\documentclass{lmcs}
\pdfoutput=1

\usepackage{lastpage}
\lmcsdoi{18}{1}{15}
\lmcsheading{}{\pageref{LastPage}}{}{}%
{Feb.~24,~2021}{Jan.~19,~2022}{}

\usepackage[utf8]{inputenc}
\usepackage{amsmath}
\newcommand\numberthis{\addtocounter{equation}{1}\tag{\theequation}}
\usepackage{amssymb} 
\usepackage{stmaryrd}
\usepackage{latexsym}
\usepackage{color}
\usepackage{verbatim}
\usepackage{relsize}
\usepackage{appendix}
\usepackage{mathtools}
\usepackage{url}
\usepackage{graphicx}
\usepackage{caption}
\usepackage{tikz}
\usetikzlibrary{arrows, automata, positioning, fit, patterns}
\usepackage{tabularx}
\usepackage{semantic}
\usepackage[]{units}
\usepackage{array}
\usepackage{empheq}
\usepackage{wasysym}
\usepackage{cite}
\usepackage{multicol}

\usepackage{listings}

\allowdisplaybreaks

\newenvironment{apx-proof}[1] 
        {\noindent\emph{Proof of #1.}}
        {\qed\\}

\newcommand{\leftmerge}{\mathbin{
                  \setlength{\unitlength}{1ex}
                  \begin{picture}(1,1.75)
                  \put(0,0){\line(1,0){1}}
                  \put(0,0){\line(0,1){1.75}}
                  \put(0.45,0){\line(0,1){1.75}}
                  \end{picture}
                 }}
\newcommand{\cmerge}{~|~}

\newcommand{\sdiff}{\ensuremath{\mathbin{\backslash}}}

\newcommand{\nil}{\mathbf{0}}
\newcommand{\parcomp}{\mathbin{\|}}
\newcommand{\Act}{\mathcal{A}}

\newcommand{\Var}{\mathcal{V}}

\newcommand{\tr}{\mathtt{T}}
\newcommand{\ctr}{\mathtt{CT}}
\newcommand{\fail}{\mathtt{F}}
\newcommand{\ready}{\mathtt{R}}
\newcommand{\ftr}{\mathtt{FT}}
\newcommand{\rtr}{\mathtt{RT}}
\newcommand{\pf}{\mathtt{PF}}

\newcommand{\s}{\mathtt{S}}
\newcommand{\cs}{\mathtt{CS}}

\newcommand{\rs}{\mathtt{RS}}
\newcommand{\B}{\mathtt{B}}

\newcommand{\trequiv}{\sim_{\tr}}

\newcommand{\fequiv}{\sim_{\fail}}
\newcommand{\requiv}{\sim_{\ready}}
\newcommand{\ftrequiv}{\sim_{\ftr}}
\newcommand{\rtrequiv}{\sim_{\rtr}}
\newcommand{\pfequiv}{\sim_{\pf}}

\newcommand{\size}{\mathrm{size}}
\newcommand{\depth}{\mathrm{depth}}
\newcommand{\norm}{\mathrm{norm}}
\newcommand{\var}{\mathrm{var}}
\newcommand{\init}{\mathtt{I}}

\newcommand{\trans}[1][]{\xrightarrow{\, {#1} \, }}
\newcommand{\ntrans}[1][]{\mathrel{{\trans[#1]}\makebox[0em][r]{$\not$\hspace{2ex}}}{\!}}

\newcommand{\rel}{\,{\mathcal R}\,}

\newcommand{\proc}{\mathcal{P}}

\newcommand{\bccsp}{\mathrm{BCCSP}_{\scriptstyle \|}}

\newcommand{\cl}{\mathrm{cl}}

\newcommand{\e}{\varepsilon}

\newcommand{\E}{\mathcal{E}}

\usepackage[normalem]{ulem}

\newlength\wantedwidth
\newcommand{\fakewidth}[2]{%
  \settowidth{\wantedwidth}{\ensuremath{#2}}%
  \makebox[\wantedwidth]{\ensuremath{#1}}%
}

\begin{document}

\title[On the axiomatisability of parallel composition]{On the Axiomatisability of Parallel Composition} 

\titlecomment{A preliminary version of this paper appeared as~\cite{ACILP20}.}

\author[L.~Aceto]{Luca Aceto\rsuper{a,b}}
\address{Reykjavik University, Iceland}
\email{\{luca,valentinac,annai,mathiasrp\}@ru.is}
\address{Gran Sasso Science Institute (GSSI), Italy}

\author[V.~Castiglioni]{Valentina Castiglioni\rsuper{a}}

\author[A.~Ing{\'o}lfsd{\'o}ttir]{Anna Ing{\'o}lfsd{\'o}ttir\rsuper{a}}

\author[B.~Luttik]{Bas Luttik\rsuper{c}}
\address{Eindhoven University of Technology, The Netherlands}
\email{s.p.luttik@tue.nl}

\author[M.~R.~Pedersen]{Mathias Ruggaard Pedersen\rsuper{a}}


\begin{abstract}
This paper studies the existence of finite equational axiomatisations of the \emph{interleaving parallel composition operator} modulo the behavioural equivalences in van~Glabbeek’s \emph{linear time-branching time spectrum}. 
In the setting of the process algebra BCCSP over a finite set of actions, we provide \emph{finite}, ground-complete axiomatisations for various simulation and (decorated) trace semantics. 
We also show that no congruence over BCCSP that includes bisimilarity and is included in possible futures equivalence has a finite, ground-complete axiomatisation; this \emph{negative result} applies to all the nested trace and nested simulation semantics. 
\end{abstract}

\maketitle


\section{Introduction}

\emph{Process algebras}~\cite{BPS01,BBR10} are prototype specification languages allowing for the description and analysis of concurrent and distributed systems, or simply \emph{processes}. 
These languages offer a variety of operators to specify composite processes from components one has already built. 
Notably, in order to model the concurrent interaction between processes, the majority of process algebras include some form of \emph{parallel composition} operator, also known as merge.

Following Milner's seminal work on CCS~\cite{M89}, the semantics of a process algebra is often defined according to a two-step approach. 
In the first step, the \emph{operational semantics}~\cite{Plo81} of a process is modelled via a \emph{labelled transition system} (LTS)~\cite{Ke76}, in which computational steps are abstracted into state-to-state transitions having actions as labels.

Behavioural equivalences have then been introduced, in the second step, as simple and elegant tools for comparing the behaviour of processes.
These are equivalence relations defined on the states of LTSs allowing one to establish whether two processes have the same \emph{observable behaviour}. 
Different notions of observability correspond to different levels of abstraction from the information carried by the LTS, which can either be considered irrelevant in a given application context, or be unavailable to an external observer.

In~\cite{vG90}, van Glabbeek presented the \emph{linear time-branching time spectrum}, i.e., a taxonomy of behavioural equivalences based on their distinguishing power.
He carried out his study in the setting of the process algebra BCCSP, which consists of the basic operators from CCS~\cite{M89} and CSP~\cite{H85}, and he proposed \emph{ground-complete axiomatisations} for most of the congruences in the spectrum over this language.
(An axiomatisation is ground-complete if it can prove all the valid equations relating terms that do not contain variables.)
The presented ground-complete axiomatisations are \emph{finite} if so is the set of actions.
For the ready simulation, ready trace and failure trace equivalences, the axiomatisation in~\cite{vG90} made use of conditional equations;
Blom, Fokkink and Nain gave purely equational, finite axiomatisations in~\cite{BFN03}.
Then, the works in~\cite{AFGI04}, on nested semantics, and in~\cite{CF08}, on impossible futures semantics, completed the studies of the axiomatisability of behavioural congruences over BCCSP by providing \emph{negative} results: neither impossible futures nor any of the nested semantics have a finite, ground-complete axiomatisation over BCCSP. 

Obtaining a complete axiomatisation of a behavioural congruence is a classic, key problem in concurrency theory, as an equational axiomatisation characterises the semantics of a process algebra in a purely syntactic fashion.
Hence, this characterisation becomes independent of the details of the definition of the process semantics of interest, allowing one to compare semantics that may have been defined in very different styles via a collection of revealing axioms.

All the results mentioned so far were obtained over the algebra BCCSP, which does not include any operator for the parallel composition of processes.
Considering the crucial role of such an operator, it is natural to ask which of those results would still hold over a process algebra including it.

In the literature, we can find a wealth of studies on the axiomatisability of parallel composition modulo \emph{bisimulation semantics}~\cite{P81}.
Briefly, in the seminal work~\cite{HM85}, Hennessy and Milner proposed a ground-complete axiomatisation of the recursion-free fragment of CCS modulo bisimilarity.
That axiomatisation, however, included infinitely many axioms, which corresponded to instances of the \emph{expansion law} used to express equationally the semantics of the merge operator.
Then, Bergstra and Klop showed in~\cite{BK84b} that a finite ground-complete axiomatisation modulo bisimilarity can be obtained by enriching CCS with two auxiliary operators, i.e., the \emph{left merge} $\leftmerge$ and the \emph{communication merge} $\!\!\!\cmerge\!\!$.
Later, Moller proved that the use of auxiliary operators is indeed necessary to obtain a finite equational axiomatisation of bisimilarity in~\cite{Mol89,Mol90a,moller90}.

To the best of our knowledge, no systematic study of the axiomatisability of the parallel composition operator modulo the other semantics in the spectrum has been presented so far.

\paragraph{Our contribution}
We consider the process algebra $\bccsp$, i.e., BCCSP enriched with the interleaving parallel composition operator, and we study the existence of finite equational axiomatisations of the behavioural congruences in the linear time-branching time spectrum over it.
Our results delineate the \emph{boundary} between finite and non-finite axiomatisability of the congruences in the spectrum over the language $\bccsp$ (see Figure~\ref{fig:spectrum}).

We start by providing a \emph{finite}, \emph{ground-complete} axiomatisation for \emph{ready simulation} semantics~\cite{BIM95}.
The axiomatisation is obtained by extending the one for BCCSP with a few axioms expressing equationally the behaviour of interleaving modulo the considered congruence.
The added axioms allow us to eliminate all occurrences of the interleaving operator from $\bccsp$ processes, thus reducing ground-completeness over $\bccsp$ to ground-completeness over BCCSP~\cite{vG90,BFN03}.
Since the axioms for the elimination of parallel composition modulo ready simulation equivalence are, of course, sound with respect to equivalences that are coarser than ready simulation equivalence, the ``reduction to ground-completeness over BCCSP'' works for all behavioural equivalences in the spectrum below ready simulation equivalence. 
Nevertheless, for those equivalences, we shall offer more elegant axioms to equationally eliminate parallel composition from closed terms.
We shall then observe a sort of parallelism between the axiomatisations for the notions of simulation and the corresponding decorated trace semantics: the axioms used to equationally express the interplay between the interleaving operator and the other operators of BCCSP in a decorated trace semantics can be seen as the \emph{linear counterpart} of those used in the corresponding notion of simulation semantics.
For instance, while the axioms for ready simulation impose constraints on the form of both arguments of the interleaving operator to facilitate equational reductions, those for ready trace equivalence impose similar constraints but only on one argument.

Finally, we complete our journey in the spectrum by showing that \emph{nested simulation} and \emph{nested trace} semantics do not have a finite axiomatisation over $\bccsp$.
To this end, firstly we adapt Moller's arguments to the effect that bisimilarity is not finitely based over CCS to obtain the \emph{negative result} for \emph{possible futures equivalence}, also known as $2$-\emph{nested trace equivalence}.
Informally, the negative result is obtained by providing an infinite family of equations that are all sound modulo possible futures equivalence but that cannot all be derived from any finite, sound axiom system.
Then, we exploit the soundness modulo bisimilarity of the equations in the family to extend the negative result to all the congruences that are finer than possible futures and coarser than bisimilarity, thus including all nested trace and nested simulation semantics.

All the results mentioned so far are obtained for a parallel composition operator that implements interleaving without synchronisation between parallel components.
As a natural extension, we then discuss the effect of extending our results to parallel composition with CCS-style synchronisation.

\paragraph{Organisation of contents}
After reviewing some basic notions on behavioural equivalences and equational logic in Section~\ref{sec:background}, we start our journey in the spectrum by providing a finite, ground-complete axiomatisation for ready simulation equivalence over $\bccsp$ in Section~\ref{sec:ready_simulation}.
In Section~\ref{sec:other_simulations} we discuss how it is possible to refine the axioms for ready simulation to obtain finite, ground-complete axiomatisations for completed simulation and simulation equivalences.
Then, in Section~\ref{sec:linear} similar refinements are provided for the (decorated) trace equivalences, thus completing the presentation of our positive results.
We end our journey in Section~\ref{sec:negative} with the presentation of the negative results, namely that the nested simulation and nested trace equivalences do not have a finite axiomatisation over $\bccsp$.
In Section~\ref{sec:communication} we modify the semantics of parallel composition to allow processes running in parallel to synchronise and we discuss the effect of this extension on the results obtained in Sections~\ref{sec:ready_simulation}--\ref{sec:negative}.
Finally, in Section~\ref{sec:conclusion} we draw some conclusions and discuss avenues for future work.

\paragraph{What's new}
A preliminary version of this paper appeared as~\cite{ACILP20}.
We have enriched our previous contribution as follows:
\begin{enumerate}
\item We provide the full proofs of our results.
\item We extend the results presented in~\cite{ACILP20} from purely interleaving parallel composition to parallel composition with synchronisation \`a la CCS. (Section~\ref{sec:communication}).
\item We present model constructions showing that the specific axioms we provide to axiomatise parallel compositions in ready simulation equivalence cannot be derived from the axiomatisations of completed simulation equivalence and ready trace equivalence.
Similarly, we show that the axiomatisations of completed simulation and failure equivalence cannot be derived from that of completed trace equivalence.
\end{enumerate}


\section{Background}
\label{sec:background}

\subsection{The language}

The language $\bccsp$ extends BCCSP with parallel composition.
Formally, $\bccsp$ consists of basic operators from CCS~\cite{M89} and CSP~\cite{H85}, with the purely \emph{interleaving} parallel composition operator $\parcomp$, and is given by the following grammar: 
\[
  t \Coloneqq \nil \;\mid\; x \;\mid\; a.t \;\mid\; t + t \;\mid\; t \parcomp t
\]
where $a$ ranges over a set of actions $\Act$ and $x$ ranges over a countably infinite set of variables $\Var$.
In what follows, we assume that the set of actions $\Act$ is \emph{finite} and non-empty.

We shall use the meta-variables $t,u,\dots$ to range over $\bccsp$ terms, and write $\var(t)$ for the collection of variables occurring in the term $t$.  
We also adopt the standard convention that prefixing binds strongest and $+$ binds weakest.
Moreover, trailing $\nil$'s will often be omitted from terms.
We use a \emph{summation} $\sum_{i\in\{1,\ldots,k\}}t_i$ to denote the term $t= t_1+\cdots+t_k$, where the empty sum represents {\bf 0}.
We can also assume that the terms $t_i$, for $i \in \{1,\dots,k\}$, do not have $+$ as head operator, and refer to them as the \emph{summands} of $t$.
The \emph{size} of a term $t$, denoted by $\size(t)$, is the number of operator symbols in it. 

A $\bccsp$ term is \emph{closed} if it does not contain any variables.
We shall, sometimes, refer to closed terms simply as \emph{processes}.
We let $\proc$ denote the set of $\bccsp$ processes and let $p,q,\dots$ range over it.
We use the \emph{Structural Operational Semantics} (SOS) framework~\cite{Plo81} to equip processes with an operational semantics.  
A \emph{literal} is an expression of the form $t \trans[a] t'$ for some process terms $t,t'$ and action $a \in \Act$.
It is \emph{closed} if both $t,t'$ are closed terms.
The inference rules for \emph{prefixing} $a.\_$, \emph{nondeterministic choice} $+$ and \emph{interleaving parallel composition} $\parcomp$ are reported in Table~\ref{tab:semantics}.
\begin{table}[t]
\centering
\begin{gather*}
\inference{}{a.x \trans[a] x} \;\,
\inference{x \trans[a] x'}{x + y \trans[a] x'} \;\,
\inference{y \trans[a] y'}{x + y \trans[a] y'} \;\,
\inference{x \trans[a] x'}{x \parcomp y \trans[a] x' \parcomp y} \;\,
\inference{y \trans[a] y'}{x \parcomp y \trans[a] x \parcomp y'}
\end{gather*}
\caption{Operational semantics of $\bccsp$.}
\label{tab:semantics}
\end{table}
A \emph{substitution} $\sigma$ is a mapping from variables to terms.
It extends to terms, literals and rules in the usual way.
A substitution is \emph{closed} if it maps every variable to a process. 

The inference rules in Table~\ref{tab:semantics} induce the $\Act$-\emph{labelled transition system}~\cite{Ke76} $(\proc,\Act,\trans[])$ whose transition relation ${\trans[]} \subseteq \proc \times \Act \times \proc$ contains exactly the closed literals that can be derived using the rules in Table~\ref{tab:semantics}.
As usual, we write $p \trans[a] p'$ in lieu of $(p,a,p') \in {\trans[]}$.
For each $p \in \proc$ and $a \in \Act$, we write $p\trans[a]$ if $p\trans[a] p'$ holds for some $p'$, and $p \ntrans[a]$ otherwise. 
The \emph{initials} of $p$ are the actions that label the outgoing transitions of $p$, that is, $\init(p) = \{a \mid p\trans[a]\}$. 
For a sequence of actions $\alpha=a_1\cdots a_k$ ($k\ge 0$), and processes $p,p'$, we write $p \trans[\alpha] p'$ if and only if there exists a sequence of transitions $p=p_0 \trans[a_1] p_1 \trans[a_2] \cdots \trans[a_k] p_k=p'$. 
If $p \trans[\alpha] p'$ holds for some process $p'$, then $\alpha$ is a {\em trace} of $p$, and $p'$ is a \emph{derivative} of $p$. 
Moreover, we say that $\alpha$ is a \emph{completed trace} of $p$ if $\init(p') = \emptyset$.
We let $\tr(p)$ denote the \emph{set of traces} of $p$, and we use $\ctr(p) \subseteq \tr(p)$ for the set of completed traces of $p$. 
We write $\e$ for the \emph{empty trace}; $|\alpha|$ stands for the \emph{length} of trace $\alpha$.
It is well known, and easy to show, that $\tr(p)$ is finite and $\ctr(p)$ is non-empty for each $\bccsp$ process $p$. 
It follows that we can define the \emph{depth} of a process $p$, denoted by $\depth(p)$, as the length of a \emph{longest} completed trace of $p$.
Formally, $\depth(p) = \max \{ |\alpha| \mid \alpha \in \ctr(p)\}$.
Similarly, the \emph{norm} of a process $p$, denoted by $\norm(p)$, is the length of a \emph{shortest} completed trace of $p$, i.e. $\norm(p) = \min \{ |\alpha| \mid \alpha \in \ctr(p)\}$.


\subsection{Behavioural equivalences}

\emph{Behavioural equivalences} have been introduced to establish whether the behaviours of two processes are \emph{indistinguishable for their observers}.
Roughly, they allow us to check whether the \emph{observable} semantics of two processes is \emph{the same}.
In the literature we can find several notions of behavioural equivalence based on the observations that an external observer can make on the process.
In his seminal article~\cite{vG90}, van Glabbeek gave a taxonomy of the behavioural equivalences discussed in the literature on concurrency theory, which is now called the \emph{linear time-branching time spectrum} (see Figure~\ref{fig:spectrum}).
 
One of the main concerns in the development of a meta-theory of process languages is to guarantee their \emph{compositionality}, i.e., that the \emph{replacement} of a component of a system with an $\rel$-equivalent one, for a chosen behavioural equivalence $\rel$, does not affect the behaviour of that system.
In algebraic terms, this is known as the \emph{congruence property} of $\rel$ with respect to all language operators, which consists in verifying whether
\[
f(t_1,\dots,t_n) \rel f(t_1',\dots,t_n') \text{ for every $n$-ary operator } f \text{ whenever } t_i \rel t_i' \text{ for all } i = 1,\dots,n.
\]

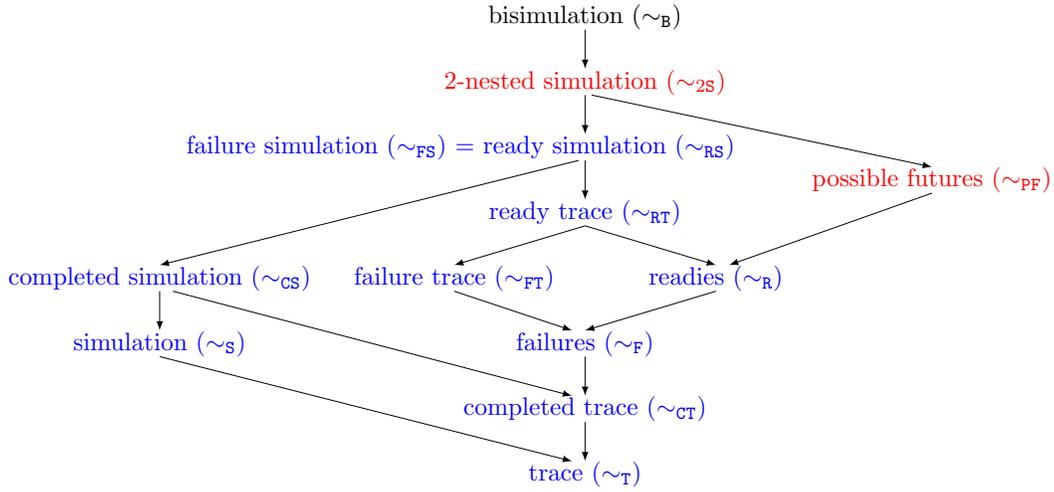
\begin{figure}
\centering
\scalebox{0.87}{
\begin{tikzpicture}
\node at (0,7){bisimulation ($\sim_\B$)}; 
\draw[-latex](0,6.8)--(0,6.2);
\node at (0,6){\textcolor{red}{$2$-nested simulation ($\sim_{2\s}$)}};
\draw[-latex](0,5.8)--(0,5.2);
\node at (-1.9,5){\textcolor{blue}{failure simulation ($\sim_{\fail\s}$) $=$ ready simulation ($\sim_{\ready\s}$)}};
\draw[-latex](0,4.8)--(0,4.2);
\node at (0,4){\textcolor{blue}{ready trace ($\sim_\rtr$)}};
\draw[-latex](0,3.8)--(-2,3.2);
\draw[-latex](0,3.8)--(2,3.2);
\node at (-2,3){\textcolor{blue}{failure trace ($\sim_\ftr$)}};
\node at (2,3){\textcolor{blue}{readies ($\sim_\ready$)}};
\draw[-latex](-2,2.8)--(-0.2,2.2);
\draw[-latex](2,2.8)--(0,2.2);
\node at (0,2){\textcolor{blue}{failures ($\sim_\fail$)}};
\draw[-latex](0,1.8)--(0,1.2);
\node at (0,1){\textcolor{blue}{completed trace ($\sim_\ctr$)}};
\draw[-latex](0,0.8)--(0,0.2);
\node at (0,0){\textcolor{blue}{trace ($\sim_\tr$)}};
\draw[-latex](-0.1,4.8)--(-6.5,3.2);
\node at (-6.5,3){\textcolor{blue}{completed simulation ($\sim_{\mathtt{C}\s}$)}};
\draw[-latex](-6.5,2.8)--(-6.5,2.2);
\draw[-latex](-6.3,2.8)--(-0.2,1.2);
\node at (-6.5,2){\textcolor{blue}{simulation ($\sim_\s$)}};
\draw[-latex](-6.5,1.8)--(-0.2,0.2);
\draw[-latex](0.1,5.8)--(5.3,4.7);
\node at (5.3,4.5){\textcolor{red}{possible futures ($\sim_\pf$)}};
\draw[-latex](5.3,4.3)--(2.2,3.2);
\end{tikzpicture}
}
\caption{\label{fig:spectrum} The linear time-branching time spectrum~\cite{vG90}.
For the equivalence relations in \textcolor{blue}{blue} we provide a finite, ground-complete axiomatization.
For the ones in \textcolor{red}{red}, we provide a negative result.
The case of bisimulation is known from the literature~\cite{Mol89,Mol90a,moller90}.}
\end{figure}

Since $\bccsp$ operators are defined by inference rules in the de Simone format~\cite{dS85}, by~\cite[Theorem 4]{vG93b} we have that all the equivalences in the spectrum in Figure~\ref{fig:spectrum} are congruences with respect to them.
Our aim in this paper is to investigate the existence of a finite equational axiomatisation  of $\bccsp$ modulo all those congruences.


\subsection{Equational Logic}
\label{sec:logic}

An \emph{axiom system} $\E$ is a collection of \emph{equations} $t \approx u$ over $\bccsp$.
An equation $t \approx u$ is \emph{derivable} from an axiom system $\E$, notation $\E \vdash t \approx u$, if there is an \emph{equational proof} for it from $\E$, namely if $t \approx u$ can be inferred from the axioms in $\E$ using the \emph{rules} of \emph{equational logic}, which express reflexivity, symmetry, transitivity, substitution and closure under $\bccsp$ contexts and are reported in Table~\ref{tab:equational_logic}.
In equational proofs, we shall write $p \stackrel{(\mathrm{A})}{\approx} q$ to highlight that the axiom denoted by A is used in that step of the proof.

\begin{table}[t]
\centering
\begin{gather*}
\scalebox{0.9}{($e_1$)}\; t \approx t 
\qquad
\scalebox{0.9}{($e_2$)}\; \inference{t \approx u}{u \approx t} 
\qquad
\scalebox{0.9}{($e_3$)}\; \inference{{t \approx u \quad u \approx v}}{{t \approx v}} 
\qquad
\scalebox{0.9}{($e_4$)}\; \inference{{t \approx u}}{{\sigma(t) \approx \sigma(u)}} 
\\
\scalebox{0.9}{($e_5$)}\; \inference{t \approx u}{a. t \approx a. u}
\qquad
\scalebox{0.9}{($e_6$)}\; \inference{t \approx  u \quad  t' \approx u'}{t+t' \approx u+u'}
\qquad
\scalebox{0.9}{($e_8$)}\; \inference{t \approx  u \quad t' \approx u'}{t \parcomp t' \approx u \parcomp u'}.
\end{gather*}
\caption{\label{tab:equational_logic} The rules of equational logic} 
\end{table}

We are interested in equations that are valid modulo some congruence relation $\rel$ over closed terms.
The equation $t \approx u$ is said to be \emph{sound} modulo $\rel$ if $\sigma(t) \,\rel\, \sigma(u)$ for all closed substitutions $\sigma$.
For simplicity, if $t \approx u$ is sound modulo $\rel$, then we write $t \,\rel\, u$.
An axiom system is \emph{sound} modulo $\rel$ if, and only if, all of its equations are sound modulo $\rel$. 
Conversely, we say that $\E$ is \emph{ground-complete} modulo $\rel$ if $p \,\rel\, q$ implies $\E \vdash p \approx q$ for all closed terms $p,q$.
We say that $\rel$ has a \emph{finite} ground-complete axiomatisation, if there is a \emph{finite} axiom system $\E$ that is sound and ground-complete modulo $\rel$.

\begin{table}[t]
\centering
\begin{tabular}{ll}
\hline
\, \scalebox{0.85}{(A0)} \; $x + \nil \approx x$ & 
\qquad \scalebox{0.85}{(P0)} \; $x \parcomp \nil \approx x$ \\
\, \scalebox{0.85}{(A1)} \; $x + y \approx y + x$ & 
\qquad \scalebox{0.85}{(P1)} \; $x \parcomp y \approx y \parcomp x$ \\
\, \scalebox{0.85}{(A2)} \; $(x + y) + z \approx x + (y+z)$ \\
\, \scalebox{0.85}{(A3)} \; $x + x \approx x$ \\
\hline 
\\[-.2cm]
\end{tabular}
\caption{Basic axioms for $\bccsp$.
We define $\E_0 = \{\mathrm{A0,A1,A2,A3}\}$ and $\E_1 = \E_0 \cup \{\mathrm{P0,P1}\}$.}
\label{tab:basic-axioms}
\end{table}

In Table~\ref{tab:basic-axioms} we present some basic axioms for $\bccsp$ that are sound with respect to all the behavioural equivalences in Figure~\ref{fig:spectrum}.
Henceforth, we will let $\E_0 = \{\mathrm{A0,A1,A2,A3}\}$, and we will denote by $\E_1$ the axiom system consisting of all the axioms in Table~\ref{tab:basic-axioms}, namely $\E_1= \E_0 \cup \{\mathrm{P0,P1}\}$.

\begin{table}[t]
\centering
\begin{tabular}{l}
\hline
\, \scalebox{0.85}{(EL1)}\; $ax \parcomp by \approx a(x \parcomp by) + b(ax \parcomp y)$ \\
\\[-.3cm]
\, \scalebox{0.85}{(EL2)}\; $\sum_{i \in I} a_ix_i \parcomp \sum_{j \in J} b_j y_j \approx \sum_{i \in I} a_i (x_i \parcomp \sum_{j \in J} b_j y_j) + \sum_{j \in J} b_j (\sum_{i \in I} a_i x_i \parcomp y_j)$ \\
\qquad with $a_i \neq a_k$ whenever $i \neq k$ and $b_j \neq b_h$ whenever $j \neq h$, $\forall\, i,k \in I, \forall\, j,h \in J$\\
\\[-.3cm]
\, \scalebox{0.85}{(EL3)}\; $\sum_{i \in I} a_ix_i \parcomp \sum_{j \in J} b_j y_j \approx \sum_{i \in I} a_i (x_i \parcomp \sum_{j \in J} b_j y_j) + \sum_{j \in J} b_j (\sum_{i \in I} a_i x_i \parcomp y_j)$ \\
\hline
\\[-.2cm] 
\end{tabular}
\caption{\label{tab:exp_law} The different instantiations of the expansion law.}
\end{table}

To be able to eliminate the interleaving parallel composition operator from closed terms we will make use of two refinements EL1 and EL2 of EL3, which is the classic expansion law~\cite{HM85} (see Table~\ref{tab:exp_law}).
We remark that the actions occurring in the three axioms in Table~\ref{tab:exp_law} are not action variables. 
Hence, when we write that an axiom system $\E$ includes one of these axioms, we mean that it includes all possible instances of that axiom with respect to the actions in $\Act$.
In particular, EL3 is a schema that generates infinitely many axioms, regardless of the cardinality of the set of actions.
This is due to the fact that we can have arbitrary summations in the two arguments of the parallel composition in the left hand side of EL3.
On the other hand, when the set of actions is assumed to be finite, we are guaranteed that there are only finitely many instances of EL1 and EL2.
Indeed, EL1 is a particular instance of EL2, i.e., the one in which both summations are over singletons.
The reason for considering both is that, as we will see, EL1 is enough to obtain the elimination result when combined with axioms allowing us to reduce any process of the form $(\sum_{i \in I} a_i p_i) \parcomp (\sum_{j \in J} b_j q_j)$ to $\sum_{i \in I, j \in J} (a_i p_i \parcomp b_j q_j)$.
Axiom EL2 is needed when this reduction is not sound modulo the considered semantics.


\section{Ready simulation}
\label{sec:ready_simulation}

In this section, we begin our journey in the spectrum by studying the equational theory of \emph{ready simulation} equivalence, whose formal definition is recalled below together with those of \emph{completed simulation} and \emph{simulation} equivalence.

\begin{defi}
[Simulation equivalences]
\phantom{no line here}
\begin{itemize}
\item A \emph{simulation} is a binary relation ${\rel} \subseteq \proc \times \proc$ such that, whenever $p \rel q$ and $p \trans[a] p'$, then there is some $q'$ such that $q \trans[a] q'$ and $p' \rel q'$.
We write $p \sqsubseteq_\s q$ if there is a simulation $\!\rel\!$ such that $p \rel q$.
We say that $p$ is \emph{simulation equivalent} to $q$, notation $p \sim_\s q$, if $p \sqsubseteq_\s q$ and $q \sqsubseteq_\s p$. 

\item A \emph{completed simulation} is a simulation $\!\rel\!$ such that, whenever $p \rel q$ and $\init(p) = \emptyset$, then $\init(q)=\emptyset$.
We write $p \sqsubseteq_\cs q$ if there is a completed simulation $\!\rel\!$ such that $p \rel q$.
We say that $p$ is \emph{completed simulation equivalent} to $q$, notation $p \sim_\cs q$, if $p \sqsubseteq_\cs q$ and $q \sqsubseteq_\cs p$. 

\item A \emph{ready simulation} is a simulation $\!\rel\!$ such that, whenever $p \rel q$ then $\init(p) = \init(q)$.
We write $p \sqsubseteq_\rs q$ if there is a ready simulation $\!\rel\!$ such that $p \rel q$.
We say that $p$ is \emph{ready simulation equivalent} to $q$, notation $p \sim_\rs q$, if $p \sqsubseteq_\rs q$ and $q \sqsubseteq_\rs p$. 
\end{itemize}
\end{defi}

In~\cite{vG93a} the notion of \emph{failure simulation} was also introduced as a simulation $\!\rel\!$ such that, whenever $p \rel q$ and $\init(p) \cap X = \emptyset$, for some $X \subseteq \Act$, then $\init(q) \cap X =\emptyset$.
Then, in~\cite{vG93b} it was proved that the notion of failure simulation coincides with that of ready simulation.

\begin{table}[t]
\centering
\scalebox{0.95}{
\begin{tabular}{l}
\hline \\[-.3cm]
\scalebox{0.85}{(RS)} \, $a(bx+by+z) \approx a(bx+by+z) + a(bx +z)$ \\[.3cm]
\scalebox{0.85}{(RSP1)} $(ax+ay+u) \parcomp (bz+bw+v) \approx (ax+u) \parcomp (bz+bw+v) + (ay+u) \parcomp (bz+bw+v) +$ \\[.05cm]
\qquad\quad\phantom{$(ax+ay+u) \parcomp (bz+bw+v) \approx$}
$+ (ax+ay+u) \parcomp (bz+v) + (ax+ay+u) \parcomp (bw+v)$ \\[.1cm]
\scalebox{0.85}{(RSP2)} $\left(\sum_{i \in I} a_i x_i\right) \parcomp (by + bz + w) \approx \left(\sum_{i \in I} a_i x_i\right) \parcomp (by + w) + \left(\sum_{i \in I} a_i x_i\right) \parcomp (bz + w) + $ \\[.1cm]
\qquad\quad\phantom{$\left(\sum_{i \in I} a_i x_i\right) \parcomp (by + bz + w) \approx$}
$\sum_{i \in I} a_i \left(x_i \parcomp (by + bz + w) \right)$ \\[.05cm]
\qquad\quad\phantom{$\left(\sum_{i \in I} a_i x_i\right) \parcomp (by + bz + w) \approx$}
where $a_j \neq a_k$ whenever $j \neq k$ for $j,k \in I$ \\[.3cm]
\; $\E_\rs = \E_1 \cup \{\textrm{RS, RSP1, RSP2, EL2}\}$
\\[.05cm]
\hline 
\\[-.2cm]
\end{tabular}
}
\caption{\label{tab:axioms_branching} Additional axioms for ready simulation equivalence.}
\end{table}

Our aim is to provide a \emph{finite}, \emph{ground-complete} axiomatisation of $\bccsp$ modulo ready simulation equivalence.
To this end, we recall that in~\cite{vG90} it was proved that the axiom system consisting of $\E_0$ together with axiom RS in Table~\ref{tab:axioms_branching} is a ground-complete axiomatisation of BCCSP, i.e., the language that is obtained from $\bccsp$ if $\parcomp$ is omitted, modulo $\sim_\rs$.
Hence, to obtain a finite, ground-complete axiomatisation of $\bccsp$ modulo $\sim_\rs$ it suffices to enrich the axiom system $\E_1 \cup \{\mathrm{RS}\}$ with finitely many axioms allowing one to eliminate all occurrences of $\parcomp$ from closed $\bccsp$ terms.
In fact, by letting $\E_\rs$ denote the axiom system $\E_1 \cup \{\mathrm{RS}\}$ suitably enriched with such elimination axioms, the elimination result consists in proving that for every closed $\bccsp$ term $p$ there is a closed BCCSP term $q$ (i.e., without any occurrence of $\parcomp$ in it) such that $\E_\rs \vdash p \approx q$.
Then, the completeness of the proposed axiom system over $\bccsp$ is a direct consequence of that over BCCSP proved in~\cite{vG90}.

Clearly, EL3 would allow us to obtain the desired elimination, but, as previously mentioned, it is a schema that finitely presents an infinite collection of equations, and thus an axiom system including it is infinite.
Instead, we include EL2, which is a variant of EL3 that generates only finitely many axioms (see Table~\ref{tab:exp_law}), and the schemata RSP1 and RSP2 that characterise the distributivity of $\parcomp$ over $+$ modulo $\sim_\rs$ (see Table~\ref{tab:axioms_branching}).

First of all, we notice that the axiom system $\E_\rs = \E_1 \cup \{\mathrm{RS,RSP1,RSP2,EL2}\}$ is sound modulo ready simulation equivalence.

\begin{thm}[$\E_\rs$ soundness]
\label{thm:RS_soundness}
The axiom system $\E_\rs$ is sound for $\bccsp$ modulo ready simulation equivalence, namely whenever $\E_\rs \vdash p \approx q$ then $p \sim_\rs q$.
\end{thm}

Let us focus now on ground-completeness.
Intuitively, RSP1 and RSP2 have been constructed in such a way that the set of initial actions of the two arguments of $\parcomp$ is preserved, while the initial term is reduced to a sum of terms of smaller size.
Briefly, according to the main features of ready simulation semantics, axiom RSP1 allows us to distribute $\parcomp$ over $+$ when both arguments of $\parcomp$ have nondeterministic choices among summands having the same initial action.
Conversely, axiom RSP2 deals with the case in which only one argument of $\parcomp$ has summands with the same initial action.
In order to preserve the branching structure of the process, which is fundamental to guarantee the soundness of the axioms modulo $\sim_\rs$, both RSP1 and RSP2 take into account the behaviour of both arguments of $\parcomp$: the terms in the right-hand side of both axioms are such that whenever the initial nondeterministic choice of one argument of $\parcomp$ is resolved, the entire behaviour of the other argument is preserved.
In fact, we stress that a simplified version of, e.g., RSP1 in which only one argument of $\parcomp$ distributes over $+$ would not be sound modulo $\sim_\rs$.
Consider, for instance, the process $p = (a + aa + b) \parcomp c$.
It is immediate to verify that $p \not \sim_\rs (a + b) \parcomp c + (aa + b) \parcomp c$ (since $p \not\sqsubseteq_\rs (a + b) \parcomp c + (aa + b) \parcomp c$).

The idea is that by (repeatedly) applying axioms RSP1 and RSP2, from left to right, we are able to reduce a process of the form $(\sum_{i \in I} p_i) \parcomp (\sum_{j \in J} p_j)$ to one of the form $\sum_{k \in K} p_k$ such that whenever $p_k$ has $\parcomp$ as head operator then $p_k = \sum_{h \in H} a_hp_h \parcomp \sum_{l \in L} b_l p_l$, with $a_h \neq a_{h'}$ for $h \neq h'$, and $b_l \neq b_{l'}$ for $l \neq l'$, for some closed $\bccsp$ terms $p_h,p_l$.
The elimination of $\parcomp$ from these terms can then proceed by means of the finitary refinement EL2 of the expansion law presented in Table~\ref{tab:exp_law}.
In particular, we notice that RSP2 is needed because RSP1 alone does not allow us to reduce all processes of the form $(\sum_{i \in I} p_i) \parcomp (\sum_{j \in J} p_j)$ into a sum of processes to which EL2 can be applied.
This is mainly due to the fact that, in order to be sound modulo $\sim_\rs$, RSP1 imposes constraints on the form of both arguments of a process $(\sum_{i \in I} p_i) \parcomp (\sum_{j \in J} p_j)$.

We can then proceed to prove the elimination result, starting from a useful remark on the form of closed BCCSP terms.

\begin{rem}
[General form of BCCSP processes]
\label{rmk:form_bccsp}
Given any closed BCCSP term $p$, we can assume, without loss of generality, that $p = \sum_{i \in I} a_i p_i$ for some finite index set $I$, actions $a_i \in \Act$, and closed BCCSP terms $p_i$, for $i \in I$.
In fact, in case $p$ is not already in this shape, then by applying axioms A0 and A1 in Table~\ref{tab:basic-axioms} we can remove superfluous occurrences of $\nil$ summands.
In particular, we remark that this transformation does not increase the number of operator symbols occurring in $p$.
\end{rem}

\begin{lem}
\label{lem:rsimpelelim}
For all closed \emph{BCCSP} terms $p$ and $q$ there exists a closed \emph{BCCSP} term $r$ such that $\E_\rs \vdash p \parcomp q\approx r$.
\end{lem}

\begin{proof}
The proof is by induction on $\size(p) + \size(q)$.
Since $p,q$ are closed BCCSP terms, we can assume that $p = \sum_{i \in I} a_ip_i$ and $q = \sum_{j \in J} b_jq_j$ (see Remark~\ref{rmk:form_bccsp}).
We proceed by a case analysis according to the cardinalities of the sets $I$ and $J$.
\begin{enumerate}
\item Case $|I| = 0$ or $|J| = 0$.
In that case we can apply axioms P0 and P1 in Table~\ref{tab:basic-axioms} to obtain that either $p \parcomp q \approx p$ or $p \parcomp q \approx q$.

\item Case $|I| = |J| = 1$.
Let $I = \{i_0\}$ and $J=\{j_0\}$.
In this case we have that
\[
p \parcomp q \stackrel{(\mathrm{EL2})}{\approx} a_{i_0}(p_{i_0} \parcomp b_{j_0}q_{j_0}) + b_{j_0}(a_{i_0}p_{i_0} \parcomp q_{j_0}),
\]
so by induction hypothesis there exist closed BCCSP terms $r_{i_0}$ and $r_{j_0}$ such that 
\[
p_{i_0} \parcomp b_{j_0}q_{j_0} \approx r_{i_0}
\quad \text{ and } \quad
a_{i_0}p_{i_0} \parcomp q_{j_0} \approx r_{j_0}.
\]
We can conclude that $p \parcomp q \approx a_{i_0}r_{i_0} + b_{j_0}r_{j_0}$, where $a_{i_0}r_{i_0} + b_{j_0}r_{j_0}$ is a closed BCCSP term.

\item Case $|I| = 1$ and $|J| > 1$.
We consider two sub-cases:
\begin{itemize}
\item {\sc There exist $j_0,j_1 \in J$ such that $j_0 \neq j_1$ and $b_{j_0} = b_{j_1}$.}

In this case we have that
\begin{align*}
p \parcomp q \stackrel{\text{(\textrm{A2})}}{\approx} & p \parcomp \left( b_{j_0}q_{j_0} + b_{j_1}q_{j_1} + \sum_{j \in J \sdiff \{j_0,j_1\}} b_jq_j \right) \\
\stackrel{\text{(\textrm{RSP2})}}{\approx} & a_{i_0}\left(p_{i_0} \parcomp (b_{j_0}q_{j_0} + b_{j_1}q_{j_1} + \sum_{j \in J \sdiff \{j_0,j_1\}} b_jq_j) \right) \\
& + p \parcomp (b_{j_0}q_{j_0} + \sum_{j \in J \sdiff \{j_0,j_1\}} b_jq_j) +  p \parcomp (b_{j_1}q_{j_1} + \sum_{j \in J \sdiff \{j_0,j_1\}} b_jq_j).
\end{align*}
By the induction hypothesis there exist closed BCCSP terms $r_1$, $r_2$, and $r_3$ such that
\begin{align*}
& p_{i_0} \parcomp (b_{j_0}q_{j_0} + b_{j_1}q_{j_1} + \sum_{j \in J \sdiff \{j_0,j_1\}} b_jq_j) \approx r_1 \\
& p \parcomp (b_{j_0}q_{j_0} + \sum_{j \in J \sdiff \{j_0,j_1\}} b_jq_j) \approx r_2 \\
& p \parcomp (b_{j_1}q_{j_1} + \sum_{j \in J \sdiff \{j_0,j_1\}} b_jq_j) \approx r_3.
\end{align*}
We have therefore obtained that $p \parcomp q \approx a_{i_0} r_1 + r_2 + r_3$ for the closed BCCSP term $a_0 r_1 + r_2 + r_3$.

\item {\sc For all $j_0,j_1 \in J$ such that $j_0 \neq j_1$ we have $b_{j_0} \neq b_{j_1}$.}

In this case we have that 
\[
p \parcomp q \stackrel{(\mathrm{EL2})}{\approx} a_{i_0}(p_{i_0} \parcomp \sum_{j \in J} b_jq_j) + \sum_{j \in J} b_j(p \parcomp q_j),
\]
and the induction hypothesis then gives, for each $j \in J$, a closed BCCSP term $r_j$ such that
\[
p \parcomp q_j \approx r_j
\]
as well as a closed BCCSP term $r'$ such that
\[
p_{i_0} \parcomp \sum_{j \in J} b_jq_j \approx r'.
\]
Therefore, $p \parcomp q$ is equivalent to the closed BCCSP term $a_{i_0}r' + \sum_{j \in J} b_jr_j$.
\end{itemize}

\item Case $|I| > 1$ and $|J| = 1$.
This case can be handled symmetrically to the case where $|I| = 1$ and $|J| > 1$ by applying axiom P1 in Table~\ref{tab:basic-axioms}.

\item Case $|I| > 1$ and $|J| > 1$.
We consider four sub-cases:
\begin{itemize}
\item {\sc There exist $i_0$ and $i_1$ such that $i_0 \neq i_1$ and $a_{i_0} = a_{i_1}$, and there exist $j_0$ and $j_1$ such that $j_0 \neq j_1$ and $b_{j_0} = b_{j_1}$.}

In this case we have that
\begin{align*}
p \parcomp q \approx{} & 
(a_{i_0}p_{i_0} + a_{i_1}p_{i_1} + \sum_{i \in I \sdiff \{i_0,i_1\}} a_ip_i) \parcomp (b_{j_0}q_{j_0} + b_{j_1}q_{j_1} + \sum_{j \in J \sdiff \{j_0,j_1\}} b_jq_j) \\
\stackrel{\text{(\textrm{RSP1})}}{\approx} & 
(a_{i_0}p_{i_0} + \sum_{i \in I \sdiff \{i_0,i_1\}} a_ip_i) \parcomp (b_{j_0}q_{j_0} + b_{j_1}q_{j_1} + \sum_{j \in J \sdiff \{j_0,j_1\}} b_jq_j) \\
& + (a_{i_1}p_{i_1} + \sum_{i \in I \sdiff \{i_0,i_1\}} a_ip_1) \parcomp (b_{j_0}q_{j_0} + b_{j_1}q_{j_1} + \sum_{j \in J \sdiff \{j_0,j_1\}} b_jq_j) \\
& + (a_{i_0}p_{i_0} + a_{i_1}p_{i_1} + \sum_{i \in I \sdiff \{i_0,i_1\}} a_ip_i) \parcomp (b_{j_0}q_{j_0} + \sum_{j \in J \sdiff \{j_0,j_1\}} b_jq_j) \\
& + (a_{i_0}p_{i_0} + a_{i_1}p_{i_1} + \sum_{i \in I \sdiff \{i_0,i_1\}} a_ip_i) \parcomp (b_{j_1}q_{j_1} + \sum_{j \in J \sdiff \{j_0,j_1\}} b_jq_j).
\end{align*}
By the induction hypothesis there are closed BCCSP terms $r_1$, $r_2$, $r_3$, and $r_4$ such that
\begin{align*}
r_1 \approx{} & (a_{i_0}p_{i_0} + \sum_{i \in I \sdiff \{i_0,i_1\}} a_ip_i) \parcomp (b_{j_0}q_{j_0} + b_{j_1}q_{j_1} + \sum_{j \in J \sdiff \{j_0,j_1\}} b_jq_j) \\
r_2 \approx{} & (a_{i_1}p_{i_1} + \sum_{i \in I \sdiff \{i_0,i_1\}} a_ip_1) \parcomp (b_{j_0}q_{j_0} + b_{j_1}q_{j_1} + \sum_{j \in J \sdiff \{j_0,j_1\}} b_jq_j) \\
r_3 \approx{} & (a_{i_0}p_{i_0} + a_{i_1}p_{i_1} + \sum_{i \in I \sdiff \{i_0,i_1\}} a_ip_i) \parcomp (b_{j_0}q_{j_0} + \sum_{j \in J \sdiff \{j_0,j_1\}} b_jq_j) \\
r_4 \approx{} & (a_{i_0}p_{i_0} + a_{i_1}p_{i_1} + \sum_{i \in I \sdiff \{i_0,i_1\}} a_ip_i) \parcomp (b_{j_1}q_{j_1} + \sum_{j \in J \sdiff \{j_0,j_1\}} b_jq_j).
\end{align*}
Hence $p \parcomp q$ is equivalent to the closed BCCSP term $r_1 + r_2 + r_3 + r_4$.

\item {\sc There exist $i_0$ and $i_1$ such that $i_0 \neq i_1$ and $a_{i_0} = a_{i_1}$, and for all $j_0$ and $j_1$ such that $j_0 \neq j_1$ we have $b_{j_0} \neq b_{j_1}$.}

In this case we have that
\begin{align*}
p \parcomp q \approx{} & (a_{i_0}p_{i_0} + a_{i_1}p_{i_1} + \sum_{i \in I \sdiff \{i_0,i_1\}} a_ip_i) \parcomp (\sum_{j \in J} b_jq_j) \\
\stackrel{(\mathrm{P1})}{\approx} & (\sum_{j \in J} b_jq_j) \parcomp (a_{i_0}p_{i_0} + a_{i_1}p_{i_1} + \sum_{i \in I \sdiff \{i_0,i_1\}} a_ip_i) \\
\stackrel{(\mathrm{RSP2})}{\approx} & \sum_{j \in J} b_j(q_j \parcomp (a_{i_0}p_{i_0} + a_{i_1}p_{i_1} + \sum_{i \in I \sdiff \{i_0,i_1\}} a_ip_i)) \\
& + (\sum_{j \in J} b_jq_j) \parcomp (a_{i_0}p_{i_0} + \sum_{i \in I \sdiff \{i_0,i_1\}} a_ip_i) + (\sum_{j \in J} b_jq_j) \parcomp (a_{i_1}p_{i_1} + \sum_{i \in I \sdiff \{i_0,i_1\}} a_ip_i),
\end{align*}
so the induction hypothesis gives, for each $j \in J$, a closed BCCSP term $r_j$ such that
\[
r_j \approx q_j \parcomp (a_{i_0}p_{i_0} + a_{i_1}p_{i_1} + \sum_{i \in I \sdiff \{i_0,i_1\}} a_ip_i),
\]
as well as closed BCCSP terms $r'$ and $r''$ such that
\begin{align*}
r' \approx{} & (\sum_{j \in J} b_jq_j) \parcomp (a_{i_0}p_{i_0} + \sum_{i \in I \sdiff \{i_0,i_1\}} a_ip_i) \\
r'' \approx{} & (\sum_{j \in J} b_jq_j) \parcomp (a_{i_1}p_{i_1} + \sum_{i \in I \sdiff \{i_0,i_1\}} a_ip_i).
\end{align*}
Thus $p \parcomp q$ is equivalent to the closed BCCSP term $\sum_{j \in J} b_jr_j + r' + r''$.

\item {\sc For all $i_0$ and $i_1$ such that $i_0 \neq i_1$ we have $a_{i_0} \neq a_{i_1}$, and there exist $j_0$ and $j_1$ such that $j_0 \neq j_1$ and $b_{j_0} = b_{j_1}$.}

This case follows by applying a symmetrical argument to that used in the previous item and it is therefore omitted.

\item {\sc For all $i_0$ and $i_1$ such that $i_0 \neq i_1$ we have $a_{i_0} \neq a_{i_1}$, and for all $j_0$ and $j_1$ such that $j_0 \neq j_1$ we have $b_{j_0} \neq b_{j_1}$.}

In this case we have that
\[
p \parcomp q \approx \sum_{i \in I} a_ip_i \parcomp \sum_{j \in J} b_jq_j
\]
and all the conditions for an application of axiom EL2 in Table\ref{tab:exp_law} are satisfied.
Hence
\[
p \parcomp q \stackrel{(\textrm{EL2})}{\approx} \sum_{i \in I} a_i(p_i \parcomp \sum_{j \in J} b_jq_j) + \sum_{j \in J} b_j(\sum_{i \in I} a_ip_i \parcomp q_j),
\]
and by the induction hypothesis there exist, for each $i \in I$, a closed BCCSP term $r_i$ such that
\[
r_i \approx p_i \parcomp \sum_{j \in J} b_jq_j,
\]
and for each $j \in J$, a closed BCCSP term $r_j'$ such that
\[
r_j' \approx \sum_{i \in I} a_ip_i \parcomp q_j.
\]
Therefore $p \parcomp q$ is equivalent to the closed BCCSP term $\sum_{i \in I} a_ir_i + \sum_{j \in J} b_jr_j'$.
\end{itemize}
\end{enumerate}
\end{proof}

\begin{prop}[$\E_\rs$ elimination]
\label{prop:rs_elim}
For every closed $\bccsp$ term $p$ there exists a \emph{BCCSP} term $q$  such that $\E_\rs\vdash p\approx q$.
\end{prop}

\begin{proof}
Straightforward by induction on the structure of $p$, using Lemma~\ref{lem:rsimpelelim} in the case that $p$ is of the form $p_1 \parcomp p_2$ for some processes $p_1$ and $p_2$.
\end{proof}

The ground-completeness of $\E_\rs$ then follows from the ground-completeness of $\E_0 \cup \{\mathrm{RS}\}$ over BCCSP~\cite{vG90}.

\begin{thm}
[$\E_\rs$ completeness]
\label{thm:RS_completeness}
The axiom system $\E_\rs$ is a ground-complete axiomatisation of $\bccsp$ modulo ready simulation equivalence, i.e., whenever $p \sim_\rs q$ then $\E_\rs \vdash p \approx q$.
\end{thm}

We remark that since axioms RSP1, RSP2, and EL2 are sound modulo ready simulation equivalence, they are automatically sound modulo all the equivalences in the spectrum that are coarser than $\sim_\rs$, namely the completed simulation, simulation, and (decorated) trace equivalences.
Hence, we can easily obtain finite, ground-complete axiomatisations of $\bccsp$ modulo each of those equivalences by adding RSP1, RSP2 and EL2 to the respective ground-complete axiomatisations of BCCSP that have been proposed in the literature~\cite{vG90,BFN03}.
However, for each of those equivalences we can provide stronger axioms that give a more elegant characterisation of the distributivity properties of $\parcomp$ over $+$. 
In particular, the axiom schema RSP2 yields $|\Act| \cdot 2^{|\Act|}$ equational axioms and EL2 yields $2^{2|\Act|}$ equational axioms.
By exploiting the various forms of distributivity of parallel composition over choice, we can obtain more concise ground-complete axiomatisations of $\bccsp$ modulo the coarser equivalences.
We devote the next two sections to the presentation of these results.


\section{Completed simulation and simulation}
\label{sec:other_simulations}

In this section we refine the axiom system $\E_\rs$ to obtain finite, ground-complete axiomatisations of $\bccsp$ modulo completed simulation and simulation equivalences.
To this end, we replace RSP1 and RSP2 with new axioms, tailored for the considered semantics, that allow us to obtain the elimination of $\parcomp$ from closed $\bccsp$ terms, while using less restrictive forms of distributivity of $\parcomp$ over $+$.

\begin{table}[t]
\centering
\scalebox{0.95}{
\begin{tabular}{l}
\hline \\[-.3cm]
\scalebox{0.85}{(CS)} \, $a(bx+y+z) \approx a(bx+y+z) + a(bx +z)$ \\[.3cm]
\scalebox{0.85}{(CSP1)} $(ax+by+u) \parcomp (cz+dw+v) \approx (ax+u) \parcomp (cz+dw+v) + (by+u) \parcomp (cz+dw+v) +$\\[.05cm]
\qquad\quad\phantom{$(ax+by+u) \parcomp (cz+dw+v) \approx$}
$+ (ax+by+u) \parcomp (cz+v) + (ax+by+u) \parcomp (dw+v)$ \\[.1cm]
\scalebox{0.85}{(CSP2)} $ax \parcomp (by+cz+w) \approx a(x \parcomp (by+cz+w)) + ax \parcomp (by+w) + ax \parcomp (cz+w)$ \\[.3cm]
\; $\E_\cs = \E_1 \cup \{\textrm{CS, CSP1, CSP2, EL1}\}$
\\[.05cm]
\hline 
\hline 
\\[-.35cm]
\scalebox{0.85}{(S)} \, $a(x+y) \approx a(x+y) + ax$ \\[.3cm]
\scalebox{0.85}{(SP1)} \, $(x+y) \parcomp (z+w) \approx x \parcomp (z+w) + y \parcomp (z+w) + (x+y) \parcomp z + (x+y) \parcomp w$ \\[.1cm]
\scalebox{0.85}{(SP2)} \, $ax \parcomp (y+z) \approx a(x \parcomp (y+z)) + ax \parcomp y + ax \parcomp z$ \\[.3cm]
\; $\E_\s = \E_1 \cup \{\textrm{S, SP1, SP2, EL1}\}$ \\[.1cm]
\hline 
\\[-.2cm]
\end{tabular}
}
\caption{\label{tab:axioms_branching_cs} Additional axioms for (completed) simulation equivalence.}
\end{table}

Let us focus first on completed simulation equivalence.
We can use axioms CSP1 and CSP2 in Table~\ref{tab:axioms_branching_cs} to characterise restricted forms of distributivity of $\parcomp$ over $+$ modulo $\sim_\cs$.
Intuitively, CSP1 is the \emph{completed simulation counterpart} of RSP1, and CSP2 is that of RSP2.
Notice that both CSP1 and CSP2 are such that, when distributing $\parcomp$ over $+$, we never obtain $\nil$ as an argument of $\parcomp$, thus guaranteeing their soundness modulo $\sim_\cs$.
Moreover, we stress that CSP1 and CSP2 are not sound modulo ready simulation equivalence.
This is due to the fact that both axioms allow for distributing $\parcomp$ over $+$ regardless of the initial actions of the summands.
It is then immediate to check that, for instance, $a \parcomp (b + c) \not\sim_\rs a \parcomp b + a \parcomp c + a \parcomp (b+c)$, whereas $a \parcomp (b + c) \sim_\cs a \parcomp b + a \parcomp c + a \parcomp (b+c)$.
Interestingly, due to the relaxed constraints on distributivity, by (repeatedly) applying CSP1 and CSP2, from left to right, we are able to reduce a $\bccsp$ process of the form $(\sum_{i \in I} p_i) \parcomp (\sum_{j \in J} p_j)$ to a $\bccsp$ process of the form $\sum_{k \in K} p_k$ such that whenever $p_k$ has $\parcomp$ as head operator then $p_k = a_kq_k \parcomp b_k q'_k$ for some $q_k,q'_k$.
We can then use the refinement EL1 of the expansion law to proceed with the elimination of $\parcomp$ from these terms.

Consider the axiom system $\E_\cs = \E_1 \cup \{\textrm{CS,CSP1,CSP2,EL1}\}$.
We can formalise the elimination result for $\sim_\cs$ as a direct consequence of the following result.

\begin{lem}
\label{lem:cssimpelelim}
For all closed \emph{BCCSP} terms $p$ and $q$ there exists a closed \emph{BCCSP} term $r$ such that $\E_\cs \vdash p \parcomp q\approx r$.
\end{lem}

\begin{proof}
The proof is by induction on $\size(p) + \size(q)$.
First note that, since $p$ and $q$ are closed BCCSP terms, we may assume that $p=\sum_{i\in I} a_ip_i$ and $q=\sum_{j\in J}b_j q_j$ (see Remark~\ref{rmk:form_bccsp}).
  
We proceed by a case analysis according to the cardinalities of the sets $I$ and $J$.

\begin{enumerate}
\item Case ${|I|} = 0$ or ${|J|}=0$.
First note that if ${|J|}=0$, i.e., $J=\emptyset$, then $q=\nil$, so $p\parcomp q \approx p$ by P0, and $p$ is the required closed BCCSP term. Similarly, if ${|I|}=0$, i.e., $I=\emptyset$, then $p=\nil$, so $p\parcomp q\approx q\parcomp p\approx q$ by P1 and P0 in Table~\ref{tab:basic-axioms}, and $q$ is the required closed BCCSP term.
  
\item Case ${|I|}={|J|}=1$.
Let $I=\{i_0\}$ and $J=\{j_0\}$, then
\begin{equation*}
 p\parcomp q\stackrel{\text{(EL1)}}{\approx}a_{i_0}(p_{i_0}\parcomp q)+b_{j_0}(p\parcomp q_{j_0}).
\end{equation*}
Since $\size(p_{i_0}) < \size(p)$ and $\size(q_{j_0}) < \size(q)$, by the induction hypothesis there exist closed BCCSP terms $r_{i_0}$ and $r_{j_0}$ such that ${p_{i_0}\parcomp q}\approx r_{i_0}$ and ${p\parcomp q_{j_0}}\approx r_{j_0}$.
It follows that $p\parcomp q\approx a_{i_0}r_{i_0}+b_{j_0}r_{j_0}$ and clearly $a_{i_0}r_{i_0}+b_{j_0}r_{j_0}$ is a closed BCCSP term.

\item Case ${|I|}=1$ and $|J|>1$.
We can then assume that $I=\{i_0\}$ and there exist $j_0,j_1\in J$ such that $j_0\neq j_1$, then
\begin{equation*} \textstyle
 p\parcomp q\stackrel{\text{(CSP2)}}{\approx} a_{i_0}(p_i\parcomp q)
  + p\parcomp \left(\sum_{j\in {J\sdiff{\{j_0\}}}}b_j q_j\right)
  + p\parcomp \left(\sum_{j\in {J\sdiff{\{j_1\}}}}b_j q_j\right).
\end{equation*}
Since $\size(p_{i_0}) < \size(p)$, 
$\size(\sum_{j\in {J\sdiff{\{j_0\}}}}b_j q_j) < \size(q)$ and 
$\size(\sum_{j\in {J\sdiff{\{j_1\}}}}b_j q_j) < \size(q)$,
by the induction hypothesis there exist closed BCCSP terms $r_{i_0}$, $r_{j_0}$ and $r_{j_1}$ such that
\[
p_{i_0}\parcomp q\approx r_{i_0} 
\qquad
p \parcomp \left(\sum_{j\in {J\sdiff{\{j_0\}}}}b_j q_j\right) \approx r_{j_0} 
\quad \text{ and } \quad
p \parcomp \left(\sum_{j\in {J\sdiff{\{j_1\}}}}b_j q_j\right) \approx r_{j_1}.
\]
So we have $p\parcomp q\approx a_{i_0}r_{i_0}+r_{j_0}+r_{j_1}$ and $a_{i_0}r_{i_0}+r_{j_0}+r_{j_1}$ is a closed BCCSP term.

\item Case $|I|>1$ and ${|J|}=1$.
The proof is similar as in the previous case, with an additional application of axiom P1 in Table~\ref{tab:basic-axioms}.

\item Case $|I|,|J|>1$.
In this case there exist $i_0,i_1\in I$ with $i_0\neq i_1$ and $j_0,j_1\in J$ with $j_0\neq j_1$.
Then
\begin{align*}
p\parcomp q \stackrel{\text{(CSP1)}}{\approx} &
\left(\sum_{i\in{I\sdiff\{i_0\}}}a_i p_i\right)\parcomp q 
+ \left(\sum_{i\in{I\sdiff\{i_1\}}}a_i p_i\right)\parcomp q \\
& + p \parcomp \left(\sum_{j\in{J\sdiff\{j_0\}}}b_j q_j\right)
+ p \parcomp \left(\sum_{j\in{J\sdiff\{j_1\}}}b_j q_i\right).
\end{align*}
Note that $\displaystyle{\size\left(\sum_{i\in{I\sdiff\{i_0\}}}a_i p_i\right)} < \size(p)$ and
$\displaystyle{\size\left(\sum_{i\in{I\sdiff\{i_1\}}}a_i p_i\right)} < \size(p)$, as well as
$\displaystyle{\size\left(\sum_{j\in{J\sdiff\{j_0\}}}b_j q_j\right)} < \size(q)$ and 
$\displaystyle{\size\left(\sum_{j\in{J\sdiff\{j_1\}}}b_j q_i\right)} < \size(q)$,
so by the induction hypothesis there exist $r_{i_0}$, $r_{i_1}$, $r_{j_0}$ and $r_{j_1}$ such that
\begin{align*}
\left(\sum_{i\in{I\sdiff\{i_0\}}}a_i p_i\right)\parcomp q\approx r_{i_0},
& & 
\left(\sum_{i\in{I\sdiff\{i_1\}}}a_i p_i\right)\parcomp q\approx r_{i_1} \\
p \parcomp \left(\sum_{j\in{J\sdiff\{j_0\}}}b_j q_j\right)\approx r_{j_0}
& &
p \parcomp \left(\sum_{j\in{J\sdiff\{j_1\}}}b_j q_i\right)\approx r_{j_1}.
\end{align*}
It follows that $p\parcomp q\approx r_{i_0}+r_{i_1}+r_{j_0}+r_{j_1}$ and $r_{i_0}+r_{i_1}+r_{j_0}+r_{j_1}$ is a closed BCCSP term. \qedhere
\end{enumerate}
\end{proof}

\begin{prop}[$\E_\cs$ elimination]
\label{prop:CSelimination}
For every closed $\bccsp$ term $p$ there exists a \emph{BCCSP} term $q$ such that $\E_\cs\vdash p\approx q$.
\end{prop}

\begin{proof}
Straightforward by induction on the structure of $p$, using Lemma~\ref{lem:cssimpelelim} in the case that $p$ is of the form $p_1 \parcomp p_2$ for some $p_1$ and $p_2$.
\end{proof}

A similar reasoning could be applied to obtain the elimination result for simulation equivalence.
Although this result could be directly derived by the soundness of CSP1 and CSP2 modulo simulation equivalence, stronger distributivity properties for parallel composition over summation hold modulo $\sim_\s$.
Instead of CSP1 and CSP2, we include SP1 and SP2 in Table~\ref{tab:axioms_branching_cs}, letting $\E_\s = \E_1 \cup \{\textrm{S,SP1,SP2,EL1}\}$.
By showing that the axioms CS, CSP1 and CSP2 are derivable from $\E_\s$, we can obtain the elimination result for simulation equivalence as an immediate corollary of that for completed simulation.
\pagebreak
\begin{lem}
\label{lem:ES_ECS}
The axioms of the system $\E_\cs$ are derivable from the axiom system $\E_\s$, namely:
\begin{enumerate}
\item \label{lem:SprovesCS}
  $\E_\s \vdash\mathrm{CS}$,
 \item \label{lem:SprovesCSP1}
  $\E_\s \vdash\mathrm{CSP1}$, and
  \item \label{lem:SprovesCSP2}
  $\E_\s\vdash\mathrm{CSP2}$.
\end{enumerate}
\end{lem}

\begin{proof}
We start with the derivation of CS.
We have the following equational derivation:
  \begin{eqnarray*}
    a(bx+y+z)
      & \stackrel{\mathrm{(A1),(A2)}}{\approx} &
    a(bx + z + y) \\
      & \stackrel{\mathrm{(S)}}{\approx} &
    a(bx + z + y) + a(bx+z) \\
      & \stackrel{\mathrm{(A1),(A2)}}{\approx} &
    a(bx + y + z) + a(bx+z).
  \end{eqnarray*}

In the case of CSP1, we have the following equational derivation:
  \begin{gather*}
     (ax+by+u)\parcomp (cz+dw+v) \\
    \quad  \fakewidth{\stackrel{\mathrm{(A1),(A3)}}{\approx}}{\stackrel{\mathrm{(A1),(A3)}}{\approx}} (ax+u+by+u) \parcomp (cz+v+dw+v) \\
    \quad \fakewidth{\stackrel{\mathrm{(SP1)}}{\approx}}{\stackrel{\mathrm{(A1),(A3)}}{\approx}}
    (ax+u)\parcomp(cz+v+dw+v)+(by+u)\parcomp(cz+v+dw+v)\\ \qquad\qquad\qquad\mbox{}+(ax+u+by+u)\parcomp(cz+v)+ (ax+u+by+u)\parcomp(dw+v)   \\
    \quad\fakewidth{\stackrel{\mathrm{(A1),(A3)}}{\approx}}{\stackrel{\mathrm{(A1),(A3)}}{\approx}}
    (ax+u)\parcomp(cz+dw+v)+(by+u)\parcomp(cz+dw+v)\\
    \qquad\qquad\qquad\mbox{}+(ax+by+u)\parcomp(cz+v)+ (ax+by+u)\parcomp(dw+v).
  \end{gather*}

Finally, for CSP2 we have the following equational derivation:
  \begin{gather*}
     ax \parcomp (by+cz+w) \\
    \quad  \fakewidth{\stackrel{\mathrm{(A1),(A3)}}{\approx}}{\stackrel{\mathrm{(A1),(A3)}}{\approx}} ax \parcomp (by+w+cz+w) \\
    \quad \fakewidth{\stackrel{\mathrm{(SP2)}}{\approx}}{\stackrel{\mathrm{(A1),(A3)}}{\approx}}
    a(x\parcomp (by+w+cz+w))+ax\parcomp (by+w) + ax \parcomp (cz+w) \\
    \quad\fakewidth{\stackrel{\mathrm{(A1),(A3)}}{\approx}}{\stackrel{\mathrm{(A1),(A3)}}{\approx}}
    a(x\parcomp (by+cz+w))+ax\parcomp(by+w) + ax\parcomp (cz+w). \qedhere
  \end{gather*}
\end{proof}

\begin{prop}
[$\E_\s$ elimination]
\label{prop:s_elim}
  For every closed $\bccsp{}$ term $p$ there exists a closed \emph{BCCSP} term $q$ such that $\E_\s \vdash p\approx q$.
\end{prop}

\begin{rem}
\label{rmk:non_derivable_RS}
As we showed in Lemma~\ref{lem:ES_ECS}, the axiom system $\E_\s$ proves all the equations in $\E_\cs$. 
A natural question at this point is whether all the equations in $\E_\rs$ can be derived from $\E_\cs$.
Indeed, that result would allow one to infer Proposition~\ref{prop:CSelimination} (the elimination result for completed simulation equivalence) from Proposition~\ref{prop:rs_elim} (the elimination result for ready simulation equivalence). 
The answer is negative, as it is not possible to derive EL2 from $\E_\cs$ .
To prove this claim, we have used Mace4~\cite{prover9-mace4} to generate a model for $\E_\cs$ in which EL2 does not hold (if there are at least two actions). 
The code can be found in Appendix~\ref{app:Mace4_cs}.
Below, we only present the model and show that it does not satisfy EL2.
The verification that the model indeed satisfies the axioms of $\E_\cs$ is lengthy and tedious.
We used the mCRL2 toolset~\cite{GM14,BGKLNVWWW19} to double check the correctness of the model produced by Mace4.

\begin{table}[t!]  
\centering
\scalebox{0.9}{
\begin{tabular}{r|rrrrr}
$a$: & 0 & 1 & 2 & 3 & 4\\
\hline
   & 2 & 2 & 3 & 4 & 4
\end{tabular}
}
\hspace{.7cm}
\scalebox{0.9}{
\begin{tabular}{r|rrrrr}
$b$: & 0 & 1 & 2 & 3 & 4\\
\hline
   & 3 & 2 & 3 & 4 & 4
\end{tabular}
} \\[.3cm] 
\scalebox{0.9}{
\begin{tabular}{r|rrrrr}
$\parcomp$: & 0 & 1 & 2 & 3 & 4\\
\hline
    0 & 0 & 1 & 2 & 3 & 4 \\
    1 & 1 & 0 & 2 & 1 & 2 \\
    2 & 2 & 2 & 3 & 4 & 4 \\
    3 & 3 & 1 & 4 & 4 & 4 \\
    4 & 4 & 2 & 4 & 4 & 4
\end{tabular} 
}
\hspace{.7cm}
\scalebox{0.9}{
\begin{tabular}{r|rrrrr}
$+$: & 0 & 1 & 2 & 3 & 4\\
\hline
    0 & 0 & 1 & 2 & 3 & 4 \\
    1 & 1 & 1 & 2 & 3 & 4 \\
    2 & 2 & 2 & 2 & 4 & 4 \\
    3 & 3 & 3 & 4 & 3 & 4 \\
    4 & 4 & 4 & 4 & 4 & 4
\end{tabular}
}
\\[.3cm]
\caption{\label{tab:model} A model for $\E_\cs$ and $\E_\ctr$.}
\end{table}

Let $\Act = \{a, b\}$, and consider the model with carrier set $\{0, 1, 2, 3, 4\}$ and the operations defined in Table~\ref{tab:model}. 
Briefly, each table reports, in cell $i,j$, the result of an application of the operator in the top left corner to the $i$-th element from the leftmost column and the $j$-th element from the top row.
Clearly, as the prefixing operators $a$ and $b$ are unary, they are applied only to the elements in the top row.
So, for instance, $a(0) = 2$, $b(0) = 3$, and $\mathbin{\|}(2, 3) = 4$.
Notice that the tables for parallel composition and nondeterministic choice are symmetric, as the two operators are commutative.
This model does not satisfy
\[
( ax + by ) \parcomp ( az + bw ) \approx 
a( x \parcomp ( az + bw ) ) + 
b( y \parcomp ( az + bw ) ) + 
a( ( ax + by ) \parcomp z ) + 
b( ( ax + by ) \parcomp w)
\]
for, e.g., the valuation $x=0$, $y=0$, $z=1$, and $w=1$.
In fact, the left-hand side term is mapped to $4$, while the right-hand term is mapped to $3$.
\end{rem}

In light of the results above, and those in~\cite{vG90} showing that $\E_0 \cup\{\mathrm{CS}\}$ and $\E_0 \cup\{\mathrm{S}\}$ are sound and ground-complete axiomatisations of BCCSP modulo $\sim_\cs$ and $\sim_\s$, respectively, we can infer that $\E_\cs$ and $\E_\s$ are ground-complete axiomatisations of $\bccsp$ modulo completed simulation equivalence and simulation equivalence, respectively.

\begin{thm}
[Soundness and completeness of $\E_\cs$ and $\E_\s$]
\label{thm:final_CS_S}
Let $\mathtt{X} \in \{\cs,\s\}$.
The axiom system $\E_{\mathtt{X}}$ is a sound, ground-complete axiomatisation of $\bccsp$ modulo $\sim_{\mathtt{X}}$, i.e., $p \sim_{\mathtt{X}} q$ if and only if $\E_{\mathtt{X}} \vdash p \approx q$.
\end{thm}

At the end of Section~\ref{sec:ready_simulation} we noticed that the size of $\E_\rs$ is exponential in $\Act$.
Now, as a final remark, we observe that the size of $\E_\cs$ is polynomial in $|\Act|$, and the size of $\E_\s$ is linear in $|\Act|$. 
In detail, $\E_\cs$ contains $|\Act|^2$ equations that are instances of CS, $|\Act|^4$ equations that are instances of CSP1, and $|\Act|^3$ equations arising from (CSP2).
On the other hand, in $\E_\s$, the axiom schemata S and SP2 yield $|\Act|$ equations each, whereas SP1 is just one equation.


\section{Linear semantics: from ready traces to traces}
\label{sec:linear}

We continue our journey in the spectrum by moving to the linear-time semantics.
In this section we consider trace semantics and all of its decorated versions, and we provide a finite, ground-complete axiomatisation for each of them (see Table~\ref{tab:axioms_linear}).

\begin{table}
\centering
\begin{tabular}{l}
\hline \\[-.3cm]
\, \scalebox{0.85}{(RT)} \; $a\left(\sum_{i=1}^{|\Act|}(b_i x_i + b_i y_i) + z\right) \approx a\left(\sum_{i=1}^{|\Act|}b_i x_i + z\right) + a\left(\sum_{i=1}^{|\Act|}b_i y_i + z\right)$ \\
\\[-.3cm]
\, \scalebox{0.85}{(FP)} \; $(ax + ay + w) \parcomp z \approx (ax + w) \parcomp z + (ay + w) \parcomp z$ \\
\\[-.3cm]
\; $\E_\rtr = \E_1 \cup \{\textrm{RT}, \textrm{FP}, \textrm{EL2}\}$
\\[.05cm]
\hline
\hline
\\[-.35cm]
\, \scalebox{0.85}{(FT)} \; $ax + ay \approx ax + ay + a(x + y)$ \\
\\[-.3cm]
\; $\E_\ftr = \E_1 \cup \{\textrm{FT},\textrm{RS},\textrm{FP},\textrm{EL2}\}$
\\[.05cm]
\hline
\hline
\\[-.35cm]
\, \scalebox{0.85}{(R)} \; $a(bx + z) + a(by + w) \approx a(bx + by + z) + a(by + w)$ \\
\\[-.3cm]
\; $\E_\ready = \E_1 \cup \{\textrm{R, FP, EL2}\}$
\\[.05cm]
\hline
\hline
\\[-.35cm]
\, \scalebox{0.85}{(F)} \; $ax + a(y + z) \approx ax + a(x + y) + a(y + z)$ \\
\\[-.3cm]
\; $\E_\fail = \E_1 \cup \{\textrm{F, R, FP, EL2}\}$
\\[.1cm]
\hline
\\[-.2cm]
\end{tabular}
\caption{Additional axioms for ready (trace) and failure (trace) equivalences.}
\label{tab:axioms_linear}
\end{table}

From a technical point of view, we can split the results of this section into two parts: 
\begin{enumerate}
\item those for ready trace, failure trace, ready, and failures equivalence, and 
\item those for completed trace, and trace equivalence.
\end{enumerate}
In both parts we prove the elimination result only for the finest semantics, namely ready trace (Proposition~\ref{prop:rtr_elim}) and completed trace (Proposition~\ref{prop:ctr_elim}) respectively.
We then obtain the remaining elimination results by showing that all the axioms in $\E_{\mathtt{X}}$ are provable from $\E_{\mathtt{Y}}$, where $\mathtt{X}$ is finer than $\mathtt{Y}$ in the considered part.


\subsection{From ready traces to failures}

First we deal with the decorated trace semantics based on the comparison of the failure and ready sets of processes.

\begin{defi}
[Readiness and failures equivalences]
\label{def:RT_F}
\begin{itemize}
\item 
A \emph{failure pair} of a process $p$ is a pair $(\alpha, X)$, with $\alpha \in \Act^{*}$ and $X \subseteq \Act$, such that $p \trans[\alpha] q$ for some process $q$ with $\init(q) \cap X = \emptyset$.
We denote by $\fail(p)$ the set of failure pairs of $p$.
Two processes $p$ and $q$ are \emph{failures equivalent}, denoted $p \fequiv q$, if $\fail(p) = \fail(q)$.
\item 
A \emph{ready pair} of a process $p$ is a pair $(\alpha, X)$, with $\alpha \in \Act^{*}$ and $X \subseteq \Act$, such that $p \trans[\alpha] q$ for some process $q$ with $\init(q) = X$.
We let $\ready(p)$ denote the set of ready pairs of $p$.
Two processes $p$ and $q$ are \emph{ready equivalent}, written $p \requiv q$, if $\ready(p) = \ready(q)$.
\item 
A \emph{failure trace} of a process $p$ is a sequence $X_0a_1X_1 \dots a_nX_n$, with $X_i \subseteq \Act$ and $a_i \in \Act$, such that there are $p_1, \dots, p_n \in \proc$ with $p = p_0 \trans[a_1] p_1 \trans[a_2] \dots \trans[a_n] p_n$ and $\init(p_i) \cap X_i = \emptyset$ for all $0 \leq i \leq n$.
We write $\ftr(p)$ for the set of failure traces of $p$.
Two processes $p$ and $q$ are \emph{failure trace equivalent}, denoted $p \ftrequiv q$, if $\ftr(p) = \ftr(q)$.
\item 
A \emph{ready trace} of a process $p$ is a sequence $X_0 a_1 X_1 \dots a_n X_n$, for $X_i \subseteq \Act$ and $a_i \in \Act$, such that there are $p_1, \dots p_n \in \proc$ with $p = p_0 \trans[a_1] p_1 \trans[a_2] \dots \trans[a_n] p_n$ and $\init(p_i) = X_i$ for all $0 \le i \le n$.
We write $\rtr(p)$ for the set of ready traces of $p$.
Two processes $p$ and $q$ are \emph{ready trace equivalent}, denoted $p \rtrequiv q$, if $\rtr(p) = \rtr(q)$.
\end{itemize}
\end{defi}

We consider first the finest equivalence among those in Definition~\ref{def:RT_F}, namely ready trace equivalence.
This can be considered as the linear counterpart of ready simulation: we focus on the current execution of the process and we require that each step is mimicked by reaching processes having the same sets of initial actions.
Interestingly, we can find a similar correlation between the axioms characterising the distributivity of $\parcomp$ over $+$ modulo the two semantics.
Consider axiom FP in Table~\ref{tab:axioms_linear}.
We can see this axiom as the \emph{linear} counterpart of RSP1: since in the linear semantics we are interested only in the current execution of a process, we can characterise the distributivity of $\parcomp$ over $+$ by treating the two arguments of $\parcomp$ independently from one another.
To obtain the elimination result for $\sim_\rtr$ we do not need to introduce the linear counterpart of axiom RSP2.
In fact, FP imposes constraints on the form of only one argument of $\parcomp$.
Hence, it is possible to use it to reduce any process of the form $(\sum_{i \in I} p_i) \parcomp (\sum_{j \in J} p_j)$ into a sum of processes to which EL2 can be applied.
We can in fact prove that the axioms in the system $\E_\rtr = \E_1 \cup \{\mathrm{RT,FP,EL2}\}$ are sufficient to eliminate all occurrences of $\parcomp$ from closed $\bccsp$ terms.

\begin{lem}
\label{lem:rtr_elim}
For all closed \emph{BCCSP} terms $p,q$ there exists a closed \emph{BCCSP} term $r$ such that $\E_\rtr \vdash p \parcomp q \approx r$.
\end{lem}

\begin{proof}
The proof proceeds by induction on the total number of operator symbols occurring in $p$ and $q$ together, not counting parentheses. 
First of all we notice that, since $p$ and $q$ are closed BCCSP terms, we can assume that $p=\sum_{i\in I} a_ip_i$ and $q=\sum_{j\in J}b_j q_j$
(see Remark~\ref{rmk:form_bccsp}).
  
We proceed by a case analysis according to the cardinalities of the sets $I$ and $J$.

\begin{enumerate}
\item Case ${|I|} = 0$ or ${|J|}=0$.
In case that ${|J|}=0$, i.e., $J=\emptyset$, then $q=\nil$, so $\E_\rtr \vdash p \parcomp q \approx p$ by P0, and $p$ is the required closed BCCSP term. 
Similarly, if ${|I|}=0$, i.e., $I=\emptyset$, then $p=\nil$, so $\E_\rtr \vdash p \parcomp q \approx q \parcomp p \approx q$ by P0 and P1 in Table~\ref{tab:basic-axioms}, and $q$ is the required closed BCCSP term.
  
\item Case ${|I|}={|J|}=1$.
Let $I=\{i_0\}$ and $J=\{j_0\}$, then
\[
p \parcomp q \stackrel{\text{(EL2)}}{\approx} 
a_{i_0} (p_{i_0} \parcomp q) + b_{j_0} (p \parcomp q_{j_0}).
\]
Since both $p_{i_0}$ and $q_{j_0}$ have fewer symbols than $p$ and $q$, respectively, by the inductive hypothesis there exist closed BCCSP terms $r_{i_0}$ and $r_{j_0}$ such that $\E_\rtr \vdash p_{i_0} \parcomp q \approx r_{i_0}$ and $\E_\rtr \vdash p\parcomp q_{j_0} \approx r_{j_0}$.
It follows that $\E_\rtr \vdash p \parcomp q \approx a_{i_0} r_{i_0} + b_{j_0} r_{j_0}$, where $a_{i_0} r_{i_0} + b_{j_0} r_{j_0}$ is a closed BCCSP term.

\item Case ${|I|}=1$ and $|J|>1$.
We can then assume that $I=\{i_0\}$ and there exist $j_0,j_1\in J$ such that $j_0\neq j_1$.
We can now distinguish two cases, according to whether $b_{j_0} = b_{j_1}$ for some $j_0,j_1 \in J$ with $j_0 \neq j_1$, or not.
\begin{itemize}
\item {\sc There are some $j_0,j_1 \in J$ such that $j_0 \neq j_1$ and $b_{j_0} = b_{j_1}$.}

In this case we have that 
\[
p \parcomp q 
\; \stackrel{\text{(P1)}}{\approx} \;  
\left(\sum_{j \in J} b_j q_j\right) \parcomp p 
\; \stackrel{\text{(FP)}}{\approx} \; 
\left(\sum_{j \in J\sdiff{\{j_0\}}} b_j q_j \right) \parcomp p 
+ \left(\sum_{j \in J\sdiff{\{j_1\}}} b_j q_j \right) \parcomp p.
\]
Since both $\sum_{j \in J\sdiff{\{j_0\}}} b_j q_j$ and $\sum_{j \in J\sdiff{\{j_1\}}} b_j q_j$ have fewer operator symbols than $q$, by the inductive hypothesis there exist closed BCCSP terms $r_{j_0}$ and $r_{j_1}$ such that 
\[
\E_\rtr \vdash \left(\sum_{j \in J\sdiff{\{j_0\}}} b_j q_j \right) \parcomp p \approx r_{j_0} 
\quad \text{ and } \quad
\E_\rtr \vdash \left(\sum_{j \in J\sdiff{\{j_1\}}} b_j q_j \right) \parcomp p \approx r_{j_1}.
\]
Therefore, we get that $\E_\rtr \vdash p \parcomp q \approx r_{j_0} + r_{j_1}$, where $r_{j_0} + r_{j_1}$ is a closed BCCSP term.

\item {\sc For all $j_0,j_1 \in J$ such that $j_0 \neq j_1$ it holds that $b_{j_0} \neq b_{j_1}$.}
In this case we have that
\[
p \parcomp q \stackrel{\text{(EL2)}}{\approx} 
a_{i_0} (p_{i_0} \parcomp q) + \sum_{j \in J} b_j (p \parcomp q_j),
\]
where $p_{i_0}$ has fewer operator symbols than $p$ and each $q_j$ has fewer operator symbols than $q$.
Hence, by the inductive hypothesis, we obtain that there is a closed BCCSP term $r_{i_0}$ such that
\[
\E_\rtr \vdash p_{i_0} \parcomp q \approx r_{i_0},
\]
and for each $j \in J$ there is a closed BCCSP term $r_j$ such that 
\[
\E_\rtr \vdash p \parcomp q_j \approx r_j.
\]
We can thus conclude that $\E_\rtr \vdash p \parcomp q \approx a_{i_0} r_{i_0} + \sum_{j \in J} b_j r_j$, where the right hand side of the equation is a closed BCCSP term.
\end{itemize}

\item Case $|I|>1$ and ${|J|}=1$.
The proof is similar as in the previous case, with an additional initial application of axiom P1 in Table~\ref{tab:basic-axioms}.

\item Case $|I|,|J|>1$.
In this case there exist $i_0,i_1\in I$ with $i_0\neq i_1$ and $j_0,j_1\in J$ with $j_0\neq j_1$.
The proof follows the same reasoning used in the case ${|I|} = 1$ and ${|J|} > 1$;
we distinguish three cases:
\begin{itemize}
\item {\sc There are $i_0,i_1 \in I$ such that $i_0 \neq i_1$ and $a_{i_0} = a_{i_1}$.}

In this case we have that 
\[
p \parcomp q \stackrel{\text{(FP)}}{\approx}
\left( \sum_{i \in I\sdiff{\{i_0\}}} a_i p_i \right) \parcomp q +
\left( \sum_{i \in I\sdiff{\{i_1\}}} a_i p_i \right) \parcomp q.
\]
As both $\sum_{i \in I\sdiff{\{i_0\}}} a_i p_i$ and $\sum_{i \in I\sdiff{\{i_1\}}} a_i p_i$ have fewer operator symbols than $p$, by the inductive hypothesis we get that there are closed BCSSP terms $r_{i_0}$ and $r_{i_1}$ such that
\[
\E_\rtr \vdash
\left( \sum_{i \in I\sdiff{\{i_0\}}} a_i p_i \right) \parcomp q \approx r_{i_0}
\quad \text{ and } \quad
\E_\rtr \vdash
\left( \sum_{i \in I\sdiff{\{i_1\}}} a_i p_i \right) \parcomp q \approx r_{i_1}.
\]
Consequently, we get that $\E_\rtr \vdash p \parcomp q \approx r_{i_0} + r_{i_1}$, where $r_{i_0} + r_{i_1}$ is a closed BCCSP term.

\item{\sc There are some $j_0,j_1 \in J$ such that $j_0 \neq j_1$ and $b_{j_0} = b_{j_1}$.}

This case can be obtained from the previous one, by an additional initial application of axiom P1 in Table~\ref{tab:basic-axioms}.

\item{\sc For all $i_0,i_1 \in I$ such that $i_0 \neq i_1$ we have that $a_{i_0} \neq a_{i_1}$ and for all $j_0,j_1 \in J$ such that $j_0 \neq j_1$ we have that $b_{j_0} \neq b_{j_1}$.}

In this case we have that 
\[
p \parcomp q \stackrel{\text{(EL2)}}{\approx}
\sum_{i \in I} a_i (p_i \parcomp q) + \sum_{j \in J} b_j (p \parcomp q_j).
\]
Since for each $i \in I$, the term $p_i$ has fewer operator symbols than $p$, by the inductive hypothesis we get that there is a closed BCCSP term $r_i$ such that 
\[
\E_\rtr \vdash p_i \parcomp q \approx r_i.
\]
Similarly, for each $j \in J$, the term $q_j$ has fewer operator symbols than $q$, so that by the inductive hypothesis we get that there is a closed BCCSP term $r_j$ such that 
\[
\E_\rtr \vdash p \parcomp q_j \approx r_j.
\]
Therefore, we can conclude that $\E_\rtr \vdash p \parcomp q \approx \sum_{i \in I} a_i r_i + \sum_{j \in J} b_j r_j$, where the right hand side of the equation is a closed BCCSP term. \qedhere
\end{itemize}
\end{enumerate}
\end{proof}

\begin{prop}[$\E_\rtr$ elimination]
\label{prop:rtr_elim}
For every closed $\bccsp$ term $p$ there is a closed \emph{BCCSP} term $q$ such that $\E_\rtr \vdash p \approx q$.
\end{prop}

\begin{proof}
The proof follows by an easy induction on the structure of $p$, using Lemma~\ref{lem:rtr_elim} in the case that $p$ is of the form $p_1 \parcomp p_2$ for some $p_1$ and $p_2$.
\end{proof}

\begin{rem}
\label{rmk:non_derivable_RS_2}
Similarly to the case of completed simulation (cf.\ Remark~\ref{rmk:non_derivable_RS}), the reason why we propose to prove directly the elimination result for ready trace equivalence is that axiom RSP2 cannot be derived from those in $\E_\rtr$, even though all its closed instantiations can be derived.
Table~\ref{tab:model_rt} presents a model, found using Mace4, of $\E_\rtr$ from which RSP2 is not derivable: e.g., it does not satisfy
\[
ax \parcomp ( ay + az + w ) = 
a ( x \parcomp ( ay + az + w ) ) + 
ax \parcomp ( ay + w ) + 
ax \parcomp ( az + w ) )
\]
for $x = 0$, $y=0$, $z=0$ and $w = 1$.

\begin{table}[t!]  
\centering 
\scalebox{0.9}{
\begin{tabular}{r|rrr}
a: & 0 & 1 & 2\\
\hline
   & 0 & 2 & 2
\end{tabular} 
}
\hspace{.5cm}
\scalebox{0.9}{
\begin{tabular}{r|rrr}
b: & 0 & 1 & 2\\
\hline
   & 0 & 0 & 0
\end{tabular} 
}
\hspace{.5cm}
\scalebox{0.9}{
\begin{tabular}{r|rrr}
$\parcomp$: & 0 & 1 & 2\\
\hline
    0 & 0 & 1 & 2 \\
    1 & 1 & 0 & 1 \\
    2 & 2 & 1 & 2
\end{tabular} 
}
\hspace{.5cm}
\scalebox{0.9}{
\begin{tabular}{r|rrr}
$+$: & 0 & 1 & 2\\
\hline
    0 & 0 & 1 & 2 \\
    1 & 1 & 1 & 2 \\
    2 & 2 & 2 & 2
\end{tabular}
}
\\[.3cm]
\caption{\label{tab:model_rt} A model for $\E_\rtr$ and $\E_\ctr$.}
\end{table}

We refer the interested reader to Appendix~\ref{app:Mace4_rt} for the Mace4 code.
\end{rem}

Interestingly, axiom FP also characterises the distributivity of $\parcomp$ over $+$ modulo $\sim_\ftr,\sim_\ready$ and $\sim_\fail$.
Consider the axiom systems $\E_\ftr = \E_1 \cup \{\textrm{FT,RS,FP,EL2}\}$, $\E_\ready = \E_1 \cup \{\textrm{R,FP,EL2}\}$ and $\E_\fail = \E_1 \cup \{\textrm{F,R,FP,EL2}\}$.
The following derivability relations among them and $\E_\rtr$ are then easy to check.

\begin{lem}
\label{lem:RT_derivations}
\begin{enumerate}
\item 
The axioms in the system $\E_\rtr$ are derivable from $\E_\ftr$, namely $\E_\ftr \vdash$ \emph{RT}.
\item 
The axioms in the system $\E_\rtr$ are derivable from $\E_\ready$, namely
$\E_\ready \vdash$ \emph{RT}.
\item 
The axioms in the system $\E_\ftr$ are derivable from $\E_\fail$, namely,
\begin{enumerate}
\item \label{lem:EF_EFT_1}
$\E_\fail \vdash$ \emph{FT}, and
\item \label{lem:EF_EFT_2}
$\E_\fail \vdash$ \emph{RS}.
\end{enumerate}
Moreover, also the axioms in the system $\E_\ready$ are derivable from $\E_\fail$.
\end{enumerate}
\end{lem}

\begin{proof}
\begin{enumerate}
\item 
To simplify notation, we introduce some abbreviations: let $|\Act| = n$, and for each $i \in \{1,\dots,n\}$ let $I_i = \{1,\dots,n\}\sdiff{\{i\}}$.

First of all, we notice that if we can prove that 
\begin{align*}
\E_\ftr \vdash{} & 
a \left( \sum_{i = 1}^n b_i x_i + z \right) + 
a \left( \sum_{i=1}^n b_i y_i + z \right)
\approx
a \left( \sum_{i = 1}^n b_i x_i + z \right) +
a \left( \sum_{i=1}^n b_i y_i + z \right) + \\
& \qquad\qquad\qquad\qquad\qquad\qquad\qquad\qquad\quad
+ a \left( \sum_{i =1}^n \big(b_i x_i + b_i y_i\big) + z \right)
\numberthis\label{axiom:eft1}
\\
\E_\ftr \vdash{} &
a \left( \sum_{i =1}^n \big(b_i x_i + b_i y_i\big) + z \right)
\approx
a \left( \sum_{i =1}^n \big(b_i x_i + b_i y_i\big) + z \right) +
a \left( \sum_{i =1}^n b_i x_i + z \right)
\numberthis\label{axiom:eft2},
\end{align*}
then $\E_\ftr \vdash$ RT, easily follows as
\begin{align*}
& a \left( \sum_{i =1}^n \big(b_i x_i + b_i y_i\big) + z \right) \\
\stackrel{\eqref{axiom:eft2}}{\approx{}}\qquad &
a \left( \sum_{i =1}^n \big(b_i x_i + b_i y_i\big) + z \right) +
a \left( \sum_{i =1}^n b_i x_i + z \right) \\
\stackrel{\textrm{(A1),(A3)},\eqref{axiom:eft2}}{\approx{}} &
a \left( \sum_{i =1}^n \big(b_i x_i + b_i y_i\big) + z \right) +
a \left( \sum_{i =1}^n b_i x_i + z \right) +
a \left( \sum_{i =1}^n b_i y_i + z \right) \\
\stackrel{(A1),\eqref{axiom:eft1}}{\approx{}}\quad &
a \left( \sum_{i =1}^n b_i x_i + z \right) +
a \left( \sum_{i =1}^n b_i y_i + z \right).
\end{align*}
We now proceed to prove the Equations~\eqref{axiom:eft1} and~\eqref{axiom:eft2} separately.
\begin{enumerate}
\item Equation~\eqref{axiom:eft1} directly follows by  
\begin{align*}
& a \left( \sum_{i = 1}^n b_i x_i + z \right) + 
a \left( \sum_{i=1}^n b_i y_i + z \right)
\\
\stackrel{\textrm{(FT)}}{\approx{}} \quad &
a \left( \sum_{i = 1}^n b_i x_i + z \right) + 
a \left( \sum_{i=1}^n b_i y_i + z \right) +
a \left( \Big(\sum_{i =1}^n b_i x_i + z \Big) + \Big(\sum_{i=1}^n b_i y_i + z \Big) \right) \\
\stackrel{\textrm{(A1),(A2)}}{\approx{}} &
a \left( \sum_{i = 1}^n b_i x_i + z \right) + 
a \left( \sum_{i=1}^n b_i y_i + z \right) +
a \left( \sum_{i =1}^n \big(b_i x_i + b_i y_i\big) + z \right).
\end{align*}

\item Before proceeding to the proof of Equation~\eqref{axiom:eft2}, we notice that
\begin{align*}
a (bx + by + z) \stackrel{\textrm{(RS)}}{\approx{}} \quad & 
a (bx + by + z) + a (bx + z) \\
\stackrel{\textrm{(A1),(RS)}}{\approx{}} &
a (bx + by + z) + a (bx + z) + a (by + z) \\
\stackrel{\textrm{(FT),(A3)}}{\approx{}} &
a (bx + z) + a (by + z),
\end{align*}
so that
\begin{equation}
\label{axiom:eft3}
\E_\ftr \vdash a (bx + by + z) \approx a (bx + z) + a (by + z).
\end{equation}
Using Equation~\eqref{axiom:eft3}, we can establish Equation~\eqref{axiom:eft2} holds for all $n\geq 1$ with induction on $n$.

If $n=1$, then
\begin{eqnarray*}
& & a \left( \sum_{i = 1}^n (b_i x_i + b_i y_i) + z \right)\\
& \stackrel{\eqref{axiom:eft3}}{\approx}
  & a(b_1x_1+z)+a(b_1y_1+z) \\
& \stackrel{\textrm{(A1),(A2),(A3)}}{\approx{}}
  & a(b_1x_1+z)+a(b_1y_1+z)+a(b_1x_1+z)\\
& \stackrel{\eqref{axiom:eft3}}{\approx} &
  a \left( \sum_{i = 1}^n (b_i x_i + b_i y_i) + z \right) + a \left( \sum_{i = 1}^n b_i x_i + z \right)
\end{eqnarray*}
Let $n\geq 1$, and suppose that Equation~\eqref{axiom:eft2} holds (IH). Then
\begin{eqnarray*}
& &
  a \left( \sum_{i =1}^{n+1} \big(b_i x_i + b_i y_i\big) + z \right) \\
& \stackrel{\textrm{(A1)},\textrm{(A2)},\eqref{axiom:eft3}}{\approx} &
  a \left( \sum_{i =1}^n \big(b_i x_i + b_i y_i\big) + b_{n+1}x_{n+1} + z \right) \\
 & & 
+ a \left( \sum_{i =1}^n \big(b_i x_i + b_i y_i\big) + b_{n+1}y_{n+1}+ z \right) \\
& \stackrel{\textrm{(IH)}}{\approx} &
  a \left( \sum_{i =1}^n \big(b_i x_i + b_i y_i\big) + b_{n+1}x_{n+1} + z \right)
+ a \left( \sum_{i =1}^n b_i x_i + b_{n+1}x_{n+1} + z \right) \\
&&\qquad\qquad
\mbox{} + a \left( \sum_{i =1}^n \big(b_i x_i + b_i y_i\big) + b_{n+1}y_{n+1}+ z \right) \\
& \stackrel{\textrm{(A1)},\textrm{(A2)},\eqref{axiom:eft3}}{\approx} &
  a \left( \sum_{i =1}^{n+1} \big(b_i x_i + b_i y_i\big) + z \right)
+ a \left( \sum_{i =1}^{n+1} b_i x_i + z \right)
\end{eqnarray*}
We conclude that Equation~\eqref{axiom:eft1} holds for all $n\geq 1$, and hence, in particular, for $n=|\Act|$.
\end{enumerate}

\item 
To prove that $\E_\ready \vdash\textrm{RT}$, we first establish that for every $n\geq 1$ we have:
\begin{multline}\label{axiom:eft4}
\E_\ready \vdash
  a\left( \sum_{i=1}^n b_ix_i+z\right) + a\left( \sum_{i=1}^n b_iy_i+w\right) \\
    \approx
  a\left( \sum_{i=1}^n (b_ix_i+b_iy_i)+z\right)
    + a\left( \sum_{i=1}^n b_iy_i+w\right).
\end{multline}
We proceed by induction on $n\geq 1$.

If $n=1$, then
\begin{eqnarray*}
& &
  a\left( \sum_{i=1}^n b_ix_i+z\right) + a\left( \sum_{i=1}^n b_iy_i+w\right) \\
& \stackrel{\textrm{(R)}}{\approx} &
  a\left( \sum_{i=1}^n (b_ix_i + b_iy_i) +z\right) + a\left( \sum_{i=1}^n b_iy_i+w\right)
\end{eqnarray*}
Let $n\geq 1$, and suppose that Equation~\eqref{axiom:eft4} holds (IH).
Then
\begin{eqnarray*}
& &
  a\left( \sum_{i=1}^{n+1} b_ix_i+z\right) + a\left( \sum_{i=1}^{n+1} b_iy_i+w\right) \\
& \stackrel{\textrm{(A1)},\textrm{(A2)}}{\approx} &
  a\left( b_{n+1}x_{n+1} + \sum_{i=1}^{n} b_ix_i+z\right) + a\left(b_{n+1}y_{n+1} + \sum_{i=1}^{n} b_iy_i+w\right) \\
& \stackrel{\textrm{(R)}}{\approx} &
  a\left(b_{n+1}x_{n+1} + b_{n+1}y_{n+1} + \sum_{i=1}^{n} b_ix_i+z\right) + a\left(b_{n+1}y_{n+1} + \sum_{i=1}^{n} b_iy_i+w\right) \\
& \stackrel{\textrm{(A1)},\textrm{(A2)}}{\approx} &
  a\left(\sum_{i=1}^{n} b_ix_i+b_{n+1}x_{n+1} + b_{n+1}y_{n+1} + z\right) + a\left(\sum_{i=1}^{n} b_iy_i+b_{n+1}y_{n+1} + w\right) \\
& \stackrel{\textrm{(IH)}}{\approx} &
  a\left(\sum_{i=1}^{n} (b_ix_i+b_iy_i)+b_{n+1}x_{n+1} + b_{n+1}y_{n+1} + z\right) \\
  & & + a\left(\sum_{i=1}^{n} b_iy_i+b_{n+1}y_{n+1} + w\right) \\
& \stackrel{\textrm{(A1)},\textrm{(A2)}}{\approx} &
  a\left(\sum_{i=1}^{n+1} (b_ix_i+b_iy_i) + z\right) + a\left(\sum_{i=1}^{n+1} b_iy_i + w\right) \\
\end{eqnarray*}
We conclude that Equation~\eqref{axiom:eft4} holds for all $n\geq 1$.

We can now derive axiom (RT) as follows:
\begin{eqnarray*}
&&
  a\left( \sum_{i=1}^{|\Act|} (b_ix_i+ b_iy_i) +z\right) \\
& \stackrel{\textrm{(A3)}}{\approx} &
  a\left( \sum_{i=1}^{|\Act|} (b_ix_i+ b_iy_i) +z\right) + \left( \sum_{i=1}^{|\Act|} (b_ix_i+ b_iy_i) +z\right)  \\
& \stackrel{\textrm{(A1)},\textrm{(A2)}}{\approx} &
  a\left( \sum_{i=1}^{|\Act|} (b_ix_i+ b_iy_i) +z\right) + \left( \sum_{i=1}^{|\Act|} b_ix_i+ \left( \sum_{i=1}^{|\Act|} b_iy_i +z\right)\right)  \\
& \stackrel{\textrm{(A1)},\textrm{(A2)},\eqref{axiom:eft4}}{\approx} &
  a\left( \sum_{i=1}^{|\Act|} b_iy_i+z\right) + \left( \sum_{i=1}^{|\Act|} b_ix_i+ \left( \sum_{i=1}^{|\Act|} b_iy_i +z\right)\right)  \\
& \stackrel{\textrm{(A1)},\textrm{(A2)}}{\approx} &
  a\left( \sum_{i=1}^{|\Act|} (b_ix_i+b_iy_i)+z\right) + a\left( \sum_{i=1}^{|\Act|} b_iy_i+z\right) \\
& \stackrel{\eqref{axiom:eft4}}{\approx} &
  a\left( \sum_{i=1}^{|\Act|} b_ix_i+z\right) + a\left( \sum_{i=1}^{|\Act|} b_iy_i+z\right).
\end{eqnarray*}

\item
The second claim, namely that the axioms in $\E_\ready$ are derivable from $\E_\fail$, follows directly from $\E_\ready \subseteq \E_\fail$.

To prove the first claim, we start by showing that the axiom FT in Table~\ref{tab:axioms_linear} is derivable from the axiom system $\E_\fail$. 
\begin{align*}
ax + ay \stackrel{\textrm{(A0)}}{\approx{}}\quad &
ax + a(y + \nil) \\
\stackrel{\textrm{(F)}}{\approx{}}\quad &
ax + a(x+y) + a(y+\nil) \\
\stackrel{\textrm{(A0),(A1)}}{\approx{}} &
ax + ay + a(x+y).
\end{align*}

We now proceed to prove that the axiom RS in Table~\ref{tab:axioms_linear} is derivable from the axiom system $\E_\fail$.
\begin{align*}
a (bx + by + z) + a (bx + z) \stackrel{\textrm{(A1),(A2)}}{\approx{}} &
a (bx + z) + a (by + (bx + z)) \\
\stackrel{\textrm{(R)}}{\approx{}}\quad &
a (bx + by + z) + a (by + (bx + z)) \\
\stackrel{\textrm{(A2),(A3)}}{\approx{}} &
a (bx + by + z). \qedhere
\end{align*}
\end{enumerate}
\end{proof}

The next proposition is then a corollary of Proposition~\ref{prop:rtr_elim} and Lemma~\ref{lem:RT_derivations}.

\begin{prop}
[$\E_\ftr$, $\E_\ready$, $\E_\fail$ elimination]
\label{prop:FT_F_elim}
Let $\mathtt{X} \in \{\ftr,\ready,\fail\}$.
For every $\bccsp$ term $p$ there is a closed \emph{BCCSP} term $q$ such that $\E_{\mathtt{X}} \vdash p \approx q$.
\end{prop}

In~\cite{BFN03} it was proved that, under the assumption that $\Act$ is finite, the axiom system $\E_0 \cup\{\textrm{RT}\}$ is a ground-complete axiomatisation of BCCSP modulo $\sim_\rtr$.
Moreover, it was also proved that $\E_0 \cup\{\textrm{FT,RS}\}$ is a ground-complete axiomatisation of BCCSP modulo $\sim_\ftr$.
The ground-completeness of $\E_0 \cup \{\mathrm{R}\}$, modulo $\sim_\ready$, and that of $\E_0 \cup \{\mathrm{F,R}\}$, modulo $\sim_\fail$, over BCCSP were proved in~\cite{vG90}.
Consequently, the soundness and ground-completeness of the proposed axioms systems can then be derived from the elimination results above and the completeness results given in~\cite{vG90,BFN03}.

\begin{thm}
[Soundness and completeness of $\E_\rtr$, $\E_\ftr$, $\E_\ready$ and $\E_\fail$]
\label{thm:final_RT_F}
Let $\mathtt{X} \in \{\rtr,\ftr,\ready,\fail\}$.
The axiom system $\E_{\mathtt{X}}$ is a sound, ground-complete axiomatisation of $\bccsp$ modulo $\sim_{\mathtt{X}}$, i.e., $p \sim_{\mathtt{X}} q$ if and only if $\E_{\mathtt{X}} \vdash p \approx q$.
\end{thm}


\subsection{Completed traces and traces}

It remains to consider completed trace equivalence and trace equivalence.

\begin{defi}
[Trace and completed trace equivalences]
Two processes $p$ and $q$ are \emph{trace equivalent}, denoted $p \trequiv q$, if $\tr(p) = \tr(q)$.
If, in addition, it holds that $\ctr(p) = \ctr(q)$, then $p$ and $q$ are \emph{completed trace equivalent}, denoted $p \sim_\ctr q$. 
\end{defi}

\begin{rem}
\label{rmk:ctr}
Since we are considering only $\bccsp$ processes with finite traces, then to prove completed trace equivalence it is enough to check whether two processes have the same sets of completed traces.
In fact, in our setting, $\ctr(p) = \ctr(q)$ implies $\tr(p) = \tr(q)$ for all processes $p,q$.
\end{rem}

\begin{table}
\centering
\begin{tabular}{l}
\hline \\[-.3cm]
\, \scalebox{0.85}{(CT)} \; $a(bx + z) + a(cy + w) \approx a(bx + cy + z + w)$ \\
\\[-.3cm]
\, \scalebox{0.85}{(CTP)} \; $(ax + by + w) \parcomp z \approx (ax + w) \parcomp z + (by + w) \parcomp z$ \\
\\[-.3cm]
\; $\E_\ctr = \E_1 \cup\{\textrm{CT, CTP, EL1}\}$
\\[.05cm]
\hline
\hline
\\[-.35cm]
\, \scalebox{0.85}{(T)} \; $ax + ay \approx a(x + y)$ \\
\\[-.3cm]
\, \scalebox{0.85}{(TP)} \; $(x + y) \parcomp z \approx x \parcomp z + y \parcomp z$ \\
\\[-.3cm]
\; $\E_\tr = \E_1 \cup\{\textrm{T, TP, EL1}\}$
\\[.1cm]
\hline
\\[-.2cm]
\end{tabular}
\caption{Additional axioms for completed trace and trace equivalences.}
\label{tab:axioms_linear_2}
\end{table}

Consider the axiom systems $\E_\ctr = \E_1 \cup\{\textrm{CT,CTP,EL1}\}$ and $\E_\tr = \E_1 \cup\{\textrm{T,TP,EL1}\}$, presented in Table~\ref{tab:axioms_linear_2}.
In the same way that axiom FP is the linear counterpart of RSP1 and RSP2, we have that CTP is the linear counterpart of CSP1 and CSP2, and TP is that of SP1 and SP2.
It is then easy to check that we can use the axioms in $\E_\ctr$ to obtain the elimination result for $\sim_\ctr$.

\begin{lem}
\label{lem:ctr_elim}
For every closed \emph{BCCSP} terms $p,q$ there is a closed \emph{BCCSP} term $r$ such that $\E_\ctr \vdash p \parcomp q \approx r$.
\end{lem}

\begin{proof}
The proof proceeds by induction on the total number of operator symbols occurring in $p$ and $q$ together, not counting parentheses. 
Since $p$ and $q$ are closed BCCSP terms, we can assume that $p=\sum_{i\in I} a_ip_i$ and  $q=\sum_{j\in J}b_j q_j$ (see Remark~\ref{rmk:form_bccsp}).
  
We proceed by a case analysis according to the cardinalities of the sets $I$ and $J$.

\begin{enumerate}
\item Case ${|I|} = 0$ or ${|J|}=0$.
In case that ${|J|}=0$, i.e., $J=\emptyset$, then $q=\nil$, so $\E_\ctr \vdash p \parcomp q \approx p$ by P0, and $p$ is the required closed BCCSP term. 
Similarly, if ${|I|}=0$, i.e., $I=\emptyset$, then $p=\nil$, so $\E_\ctr \vdash p \parcomp q \approx q \parcomp p \approx q$ by P0 and P1, and $q$ is the required closed BCCSP term.
  
\item Case ${|I|}={|J|}=1$.
Let $I=\{i_0\}$ and $J=\{j_0\}$, then
\[
p \parcomp q \stackrel{\text{(EL1)}}{\approx} 
a_{i_0} (p_{i_0} \parcomp q) + b_{j_0} (p \parcomp q_{j_0}).
\]
Since both $p_{i_0}$ and $q_{j_0}$ have fewer symbols than $p$ and $q$, respectively, by the inductive hypothesis there exist closed BCCSP terms $r_{i_0}$ and $r_{j_0}$ such that $\E_\ctr \vdash p_{i_0} \parcomp q \approx r_{i_0}$ and $\E_\ctr \vdash p\parcomp q_{j_0} \approx r_{j_0}$.
It follows that $\E_\ctr \vdash p \parcomp q \approx a_{i_0} r_{i_0} + b_{j_0} r_{j_0}$, where $a_{i_0} r_{i_0} + b_{j_0} r_{j_0}$ is a closed BCCSP term.

\item Case ${|I|}=1$ and $|J|>1$.
We can then assume that $I=\{i_0\}$ and there exist $j_0,j_1\in J$ such that $j_0\neq j_1$, then
\begin{align*}
p \parcomp q \stackrel{\textrm{(A2),(P1)}}{\approx{}} &
\left( b_{j_0} q_{j_0} + b_{j_1} q_{j_1} + \sum_{j \in J\sdiff{\{j_0,j_1\}}} b_j q_j \right) \parcomp p \\
\stackrel{\textrm{(CTP)}}{\approx{}} &
\left( b_{j_0} q_{j_0} + \sum_{j \in {J\sdiff{\{j_0,j_1\}}}} b_j q_j \right) \parcomp p + 
\left( b_{j_1} q_{j_1} + \sum_{j \in {J\sdiff{\{j_0,j_1\}}}} b_j q_j \right) \parcomp p \\
\stackrel{\textrm{(A2),(P1)}}{\approx{}} &
\left( \sum_{j \in {J\sdiff{\{j_1\}}}} b_j q_j \right) \parcomp p + 
\left( \sum_{j \in {J\sdiff{\{j_0\}}}} b_j q_j \right)
\parcomp p.
\end{align*}
Since $\sum_{j\in {J\sdiff{\{j_0\}}}}b_j q_j$ and $\sum_{j\in {J\sdiff{\{j_1\}}}}b_j q_j$ have fewer symbols than $q$, by the inductive hypothesis there exist closed BCCSP terms $r_{j_0}$ and $r_{j_1}$ such that
\begin{align*}
& \E_\ctr \vdash p \parcomp \left( \sum_{j\in {J\sdiff{\{j_0\}}}} b_j q_j \right) \approx r_{j_0} \\
& \E_\ctr \vdash p \parcomp \left( \sum_{j\in {J\sdiff{\{j_1\}}}} b_j q_j \right) \approx r_{j_0}.
\end{align*}
So we have $\E_\ctr \vdash p \parcomp q \approx r_{j_0} + r_{j_1}$ and $r_{j_0} + r_{j_1}$ is a closed BCCSP term.

\item Case $|I|>1$ and ${|J|}=1$.
This case can be obtained as the previous one, without the initial application of axiom P1.

\item Case $|I|,|J|>1$.
In this case there exist $i_0,i_1\in I$ with $i_0\neq i_1$ and $j_0,j_1\in J$ with $j_0\neq j_1$.
Then
\begin{align*}
p \parcomp q \stackrel{\text{(A2),(P1),(CTP)}}{\approx} &
\left( \sum_{i \in {I\sdiff\{i_0\}}} a_i p_i \right)\parcomp q +
\left( \sum_{i\in {I\sdiff\{i_1\}}} a_i p_i \right)\parcomp q + \\
& + p \parcomp \left( \sum_{j\in{J\sdiff\{j_0\}}} b_j q_j \right) + 
p \parcomp \left( \sum_{j\in{J\sdiff\{j_1\}}} b_j q_i \right).
\end{align*}
Note that
$\sum_{i\in{I\sdiff\{i_0\}}} a_i p_i$ and
$\sum_{i\in{I\sdiff\{i_1\}}} a_i p_i$
have fewer symbols than $p$, and
$\sum_{j\in{J\sdiff\{j_0\}}} b_j q_j$ and
$\sum_{j\in{J\sdiff\{j_1\}}} b_j q_i$
have fewer symbols than $q$, so by the inductive hypothesis there exist $r_{i_0}$, $r_{i_1}$, $r_{j_0}$ and $r_{j_1}$ such that
\begin{align*}
\E_\ctr \vdash \left( \sum_{i\in{I\sdiff\{i_0\}}} a_i p_i \right) \parcomp q \approx r_{i_0}
& & 
\E_\ctr \vdash \left( \sum_{i\in{I\sdiff\{i_1\}}} a_i p_i \right) \parcomp q \approx r_{i_1} \\
\E_\ctr \vdash p \parcomp \left( \sum_{j\in{J\sdiff\{j_0\}}} b_j q_j \right) \approx r_{j_0}
& &
\E_\ctr \vdash p \parcomp \left( \sum_{j\in{J\sdiff\{j_1\}}} b_j q_i \right) \approx r_{j_1}.
\end{align*}
Then $\E_\ctr \vdash p\parcomp q \approx r_{i_0} + r_{i_1} + r_{j_0} + r_{j_1}$ and $r_{i_0} + r_{i_1} + r_{j_0} + r_{j_1}$ is a closed BCCSP term. \qedhere
\end{enumerate}
\end{proof}

\begin{prop}[$\E_\ctr$ elimination]
\label{prop:ctr_elim}
For every closed $\bccsp$ term $p$ there is a closed BCCSP term $q$ such that $\E_\ctr \vdash p \approx q$.
\end{prop}

\begin{proof}
The proof follows by an easy induction on $p$, using Lemma~\ref{lem:ctr_elim} in the case that $p$ is of the form $p_1 \parcomp p_2$ for some $p_1$ and $p_2$.
\end{proof}

Let us now focus on trace equivalence.
We can prove that the axioms in $\E_\ctr$ are derivable from those in $\E_\tr$.

\begin{lem}
\label{lem:ET_ECT}
The axioms in the system $\E_\ctr$ are derivable from $\E_\tr$, namely,
\begin{enumerate}
\item \label{lem:TprovesCT}
$\E_\tr \vdash\mathrm{CT}$, and
\item \label{lem:TprovesCTP} 
$\E_\tr \vdash\mathrm{CTP}$.
\end{enumerate}
\end{lem}

\begin{proof}
In the case of CT, we have the following equational derivation:
  \begin{equation*}
    a(bx+z) + a(cy+w)
      \stackrel{\mathrm{(T)}}{\approx}
    a(bx + z + cy + w).
  \end{equation*}

For CTP, we have the following equational derivation:
  \begin{equation*}
     (ax+by+w)\parcomp z
    \stackrel{\mathrm{(A1),(A3)}}{\approx}
     (ax+w+by+w) \parcomp z
    \stackrel{\mathrm{(TP)}}{\approx}
      (ax+w)\parcomp z +(by+w)\parcomp z. \qedhere
  \end{equation*}
\end{proof}

Hence, the elimination result for $\sim_\tr$ can be obtained as an immediate consequence of Lemma~\ref{lem:ET_ECT} and the elimination result for completed traces (Proposition~\ref{prop:ctr_elim} above).

\begin{prop}[$\E_\tr$ elimination]
\label{prop:tr_elim}
For every closed $\bccsp{}$ term $p$ there exists a closed \emph{BCCSP} term $q$ such that $\E_\tr \vdash p\approx q$.
\end{prop}

\begin{rem}
\label{rmk:non_derivable_CS}
We can use the models in Tables~\ref{tab:model} and~\ref{tab:model_rt} to prove that, respectively, neither axiom EL2 nor axiom CSP2 are derivable from $\E_\ctr$.
In detail, in the case of EL2 the counterexample presented in Remark~\ref{rmk:non_derivable_RS} holds also in the case of $\E_\ctr$.
In the case of CSP2, we notice that the instance of RSP2 used in Remark~\ref{rmk:non_derivable_RS_2} is in fact also an instance of CSP2, and thus the counterexample for RSP2 holds also in this case.
We refer the interested reader to Appendix~\ref{app:Mace4_ct} for the Mace4 codes generating the desired models.
\end{rem}

In light of Proposition~\ref{prop:ctr_elim}, the ground-completeness of $\E_\ctr$ over $\bccsp$ modulo $\sim_\ctr$ follows from that of $\E_0 \cup \{\mathrm{CT}\}$ over BCCSP provided in~\cite{vG90}.
Similarly, the ground-completeness of $\E_0 \cup \{\mathrm{T}\}$ over BCCSP proved in~\cite{vG90} and Proposition~\ref{prop:tr_elim} give us the ground-completeness of $\E_\tr$ over $\bccsp$.

\begin{thm}
[Soundness and completeness of $\E_\ctr$ and $\E_\tr$]
\label{thm:tr-completeness}
Let $\mathtt{X} \in \{\ctr,\tr\}$.
The axiom system $\E_{\mathtt{X}}$ is a ground-complete axiomatisation of $\bccsp$ modulo $\sim_{\mathtt{X}}$, i.e., $p \sim_{\mathtt{X}} q$ if and only if $\E_{\mathtt{X}} \vdash p \approx q$.
\end{thm}


\section{The negative results}
\label{sec:negative}

We dedicate this section to the negative results: we prove that all the congruences between possible futures equivalence ($\sim_\pf$) and bisimilarity ($\sim_\B$) do not admit a finite, ground-complete axiomatisation over $\bccsp$.
This includes all the nested trace and nested simulation equivalences.
In~\cite{AFGI04} it was shown that, even if the set of actions is a singleton, the nested semantics admit no finite axiomatisation over BCCSP.
Indeed, the presence of the additional operator $\parcomp$ might, at least in principle, allow us to finitely axiomatise the equations over closed BCCSP terms that are valid modulo the considered equivalences. 
Hence, we prove these results directly.

In detail, firstly we focus on the negative result for possible futures semantics, corresponding to the $2$-nested trace semantics~\cite{HM85}.
To obtain it, we apply the general technique used by Moller to prove that interleaving is not finitely axiomatisable modulo bisimilarity~\cite{Mol89,Mol90a,moller90}.
Briefly, the main idea is to identify a \emph{witness property}.
This is a specific property of $\bccsp$ terms, say $W_N$ for $N \ge 0$, that, when $N$ is \emph{large enough}, is an invariant that is preserved by provability from finite, sound axiom systems.
Roughly, this means that if $\E$ is a finite set of axioms that are sound modulo possible futures equivalence, the equation $p \approx q$ can be derived from $\E$, and $N$ is larger than the size of all the terms in the equations in $\E$, then either both $p$ and $q$ satisfy $W_N$, or none of them does.
Then, we exhibit an infinite family of valid equations, called the \emph{witness family of equations}, in which $W_N$ is not preserved, namely it is satisfied only by one side of each equation.

Afterwards, we exploit the soundness modulo bisimilarity of the equations in the witness family to lift the negative result for $\sim_\pf$ to all congruences between $\sim_\B$ and $\sim_\pf$.

Differently from the aforementioned negative results over BCCSP, ours are obtained assuming that the set of actions contains at least two distinct elements.
In fact, when the action set is a singleton, and \emph{only} in that case, the axiom
\[
ax \parcomp (ay + az) \approx ax \parcomp (ay + a(y+z)) + ax \parcomp (az + a(y + z))
\]
is sound modulo $\sim_\pf$.
Due to this axiom we were not able to prove the negative result for $\sim_\pf$ in the case that $|\Act| = 1$, which we leave as an open problem for future work.


\subsection{Possible futures equivalence}
\label{sec:possible_futures}

According to possible futures equivalence~\cite{RB81} two processes are deemed equivalent if, by performing the same traces, they reach processes that are trace equivalent.
For this reason, possible futures equivalence is also known as the $2$-\emph{nested trace equivalence}~\cite{HM85}.

\begin{defi}
[Possible futures equivalence]
A \emph{possible future} of a process $p$ is a pair $(\alpha,X)$ where $\alpha \in \Act^{*}$ and $X \subseteq \Act^{*}$ such that $p \trans[\alpha] p'$ for some $p'$ with $X = \tr(p')$.
We write $\pf(p)$ for the set of possible futures of $p$.
Two processes $p$ and $q$ are said to be \emph{possible futures equivalent}, denoted $p \pfequiv q$, if $\pf(p) = \pf(q)$.
\end{defi}

Our order of business is to prove the following result.

\begin{thm}
\label{thm:pf_negative}
Assume that ${|\Act|}\ge 2$.
Possible futures equivalence has no finite, ground-complete, equational axiomatisation over the language $\bccsp$.
\end{thm}

In what follows, for actions $a,b \in \Act$ and $i \ge 0$, we let $b^0 a$ denote $a.0$ and $b^{i+1}a$ stand for $b(b^ia)$.  
Consider now the infinite family of equations $\{e_N \mid N \ge 1\}$ given, for $a \neq b$, by:
\begin{align*}
& p_N = \sum_{i = 1}^N b^i a 
& (N \ge 1)\phantom{.} \\
& e_N \; \colon \; a \parcomp p_N \approx a p_N + \sum_{i = 1}^{N} b (a \parcomp b^{i-1}a)   
& (N \ge 1).
\end{align*}

Notice that the equations $e_N$ are sound modulo $\sim_\pf$ for all $N \ge 1$.

We also notice that none of the summands in the right-hand side of equation $e_N$ is, alone, possible futures equivalent to $a \parcomp p_N$.
However, we now proceed to show that, when $N$ is large enough, having a summand possible futures equivalent to $a \parcomp p_N$ is an invariant under provability from finite sound axiom systems, and it will thus play the role of witness property for our negative result.

To this end, we introduce first some basic notions and results on $\sim_\pf$. 

Firstly, as a simplification, we can focus on the $\nil$ \emph{absorption properties} of $\bccsp$.
Informally, we can restrict the axiom system to a collection of equations that do not introduce unnecessary terms that are possible futures equivalent to $\nil$ in the equational proofs, namely $\nil$ summands and $\nil$ \emph{factors}.
(We refer the interested reader to Appendix~\ref{app:saturation} for further details.)

\begin{defi}
\label{def:nil-factors}
We say that a $\bccsp$ term $t$ has a {\sl $\nil$ factor} if it contains a subterm of the form $t_1 \parcomp t_2$, where either $t_1$ or $t_2$ is possible futures equivalent to $\nil$.
\end{defi}

Next, we characterise closed $\bccsp$ terms that are possible futures equivalent to $p_N$.

\begin{lem}
\label{lem:ct_closed_pn}
Let $q$ be a $\bccsp$ term that does not have $\nil$ summands or factors and such that $\ctr(q) = \ctr(p_N)$ for some $N \ge 1$.
Then $q$ does not contain any occurrence of $\parcomp$.
\end{lem}

\begin{proof}
That $q \sim_\ctr p_N$ implies that $q$ does not contain any occurrence of $\parcomp$ directly follows by observing that the completed traces of $p_N$, and thus those of $q$, contain exactly one occurrence of $a$, and this occurrence must be as the last action in the completed trace.
\end{proof}

\begin{lem}
\label{lem:pf_closed_pn}
Let $q$ be a $\bccsp$ term that does not have $\nil$ summands or factors. Then $q \sim_\pf p_N$, for some $N \ge 1$, if and only if $q = \sum_{j \in J} q_j$ for some terms $q_j$ such that none of them has $+$ as head operator and:
\begin{itemize}
\item for each $i \in \{1,\dots,N\}$ there is some $j \in J$ such that $b^i a \sim_\pf q_j$;
\item for each $j \in J$ there is some $i \in \{1,\dots,N\}$ such that $q_j \sim_\pf b^i a$.
\end{itemize}
\end{lem}

\begin{proof}
Due to a slight difference in the technical development of the proof, we treat the case of $N=1$ separately.

\begin{itemize}
\item {\sc Case $N >1$.}

\noindent($\Rightarrow$)
Let $q \sim_\pf p_N$.

First of all, we show that $q$ cannot be of the form $q = b.q'$ for some term $q'$.
In fact, in that case, since $N > 1$, we have that $p_N \trans[b] a$ and $p_N \trans[b] ba$. 
Therefore, as $\pf(q) = \pf(p_N)$, it follows that both equalities must hold, $(b,\{\e,a\}) = (b,\tr(q'))$ and $(b,\{\e,b,ba\}) = (b,\tr(q'))$
(since $q'$ is the only term reachable from $q$ via the trace $b$, it has to be trace equivalent to both $a$ and $ba$, besides all the other terms reachable from $p_N$ via the trace $b$ in the case $N > 2$), 
which is a contradiction. 

We have therefore obtained that $q = \sum_{i \in J} q_j$ for some $J$ with $|J| \ge 2$, and some $q_j$, $j \in J$ of the form $b.q_j'$.
Let $i \in \{1,\ldots,N\}$. 
Since $p_N \trans[b] b^{i-1}a$, we have that $(b,\tr(b^{i-1}a))$ is a possible future of $p_N$.
Since $q$ is possible futures equivalent to $p_N$, there is some $j$ with $\tr(q'_j)=\tr(b^{i-1}a)$. 
We claim that $q'_j$ is possible futures equivalent to $b^{i-1}a$, which yields that $q_j = b.q'_j$ is possible futures equivalent to $b^i a$. 
To see this, assume, towards a contradiction, that $\tr(q'_j)=\tr(b^{i-1}a)$ but $q_j'$ is not possible futures equivalent to $b^{i-1}a$. 
This means that there are some $k \le i-1$ and $q'$ such that $q'_j \trans[b^k] q'$ and $\tr(q')$ is strictly included in $\tr(b^{i-1-k}a)$. 
Since $\tr(q')$ is prefix closed, this means that $\tr(q') = \{ b^h \mid 0\leq h \leq \ell \}$ for some $\ell \le i-1-k$. 
This means that $q$ has $(b^{k+1},\tr(q'))$ as one of its possible futures.
On the other hand, $p_N$ has no such possible future, due to the missing trailing $a$, which contradicts our assumption that $q \sim_\pf p_N$.

Assume now, towards a contradiction, that there is a $j \in J$ such that $q_j \not\sim_\pf b^i a$ for all $i \in \{1,\dots,N\}$.
As $q_j$ is of the form $b.q'_j$, for some $q_j'$, we get that $q \trans[b] q_j'$ and $(b,\tr(q_j')) \in \pf(q)$.
We can distinguish two cases according to the form of $\ctr(q_j')$:
\begin{itemize}
\item $\ctr(q_j')$ includes one completed trace of the form $b^k a$ for some $k \in \{0,\dots,N-1\}$.
Since $p_N$ is possible futures equivalent to $q$, there is some $i \in \{1,\dots,N\}$ with $\tr(b^{i-1}a)=\tr(q_j')$. 
In this case, the same reasoning applied above yields a contradiction with $q \sim_\pf p_N$.
\item There are $b^k a, b^h a \in \ctr(q_j')$ for some $k \neq h, k,h \in \{0,\dots,N-1\}$.
Since there is no $p'$ such that $p_N \trans[b] p'$ and $\ctr(p') = \ctr(q_j')$, we get an immediate contradiction with $q \sim_\pf p_N$.
\end{itemize}

($\Leftarrow$) Assume now that $q = \sum_{j \in J} q_j$, for some terms $q_j$ that do not have $+$ as head operator and such that
\begin{enumerate}
\item for each $i \in \{1,\dots,N\}$ there is a $j \in J$ such that $b^i a \sim_\pf q_j$, and
\item for each $j \in J$ there is a $i \in \{1,\dots,N\}$ such that $q_j \sim_\pf b^i a$.
\end{enumerate}
Since possible futures equivalence is a congruence with respect to summation and, moreover, summation is an idempotent operator (axiom A3 in Table~\ref{tab:basic-axioms}), from these assumptions we can directly conclude that $q \sim_\pf p_N$.

\item {\sc Case $N = 1$.}

\noindent($\Rightarrow$)
Let $q \sim_\pf p_1$.
Hence 
\[
\pf(q) = \left\{
\Big( \e,\{\e,b,ba\} \Big), 
\Big( b,\{\e,a\} \Big), 
\Big( ba,\{\e\} \Big)
\right\}
\]
Assume that $q = b.q'$ for some term $q'$.
Then $(b,\tr(q')) \in \pf(q)$ implies $\tr(q') = \{\e,a\}$, and hence $q' \sim_\pf a$.
This gives $q = \sum_{j \in J} q_j$ for some $J$ with $|J| = 1$ and $q_j \sim_\pf ba$.
The two properties relating the summands of $q$ and $p_1$ are then immediate.

Assume now that $q = \sum_{j \in J} q_j$ for some $J$ with $|J| \ge 2$ and some $q_j$ of the form $b.q_j'$.
Assume, towards a contradiction, that there is some $q_j$ such that $q_j \not\sim_\pf ba$.
Since $q \sim_\pf p_1$ implies $\tr(q) = \tr(ba)$, we get that $q_j$, and thus $q$, has $(b,\{\e\})$ as a possible future.
This gives an immediate contradiction with $q \sim_\pf p_1$.
\\

\noindent($\Leftarrow$)
The proof of this implication follows as in the case of $N > 1$.
\end{itemize}
This concludes the proof.
\end{proof}

In light of Lemma~\ref{lem:pf_closed_pn}, we can also provide a decomposition-like characterisation of closed $\bccsp$ terms that are possible futures equivalent to $a \parcomp p_N$.
To this end, we need first to lift the notions of transition, norm and depth from closed terms to terms.
The action-labelled transition relation over $\bccsp$ terms contains exactly the literals that can derived using the rules in Table~\ref{tab:semantics}.
Clearly, $t \trans[a] t'$ implies $\sigma(t) \trans[a] \sigma(t')$ for all substitutions $\sigma$.
The norm and depth of a term are defined exactly as the norm and depth of processes, by replacing the transition relation over processes with that over terms.
Equivalently, the norm and depth of terms can be defined inductively over the structure of terms as follows:
\begin{multicols}{2}
\begin{itemize}
\item $\norm(\nil) = 0$;
\item $\norm(x) = 0$;
\item $\norm(a.t) = 1 + \norm(t)$;
\item $\norm(t+u) = \min\{\norm(t), \norm(u)\}$;
\item $\norm(t \parcomp u) = \norm(t) + \norm(u)$.
\end{itemize}
\columnbreak
\begin{itemize}
\item $\depth(\nil) = 0$;
\item $\depth(x) = 0$;
\item $\depth(a.t) = 1 + \depth(t)$;
\item $\depth(t+u) = \max\{\depth(t), \depth(u)\}$;
\item $\depth(t \parcomp u) = \depth(t) + \depth(u)$.
\end{itemize}
\end{multicols}

For $k \ge 0$, we denote by $\var_k(t)$ the set of variables occurring in the $k$-derivatives of $t$, namely $\var_k(t) = \{x \in \var(t') \mid t \trans[\alpha] t', |\alpha| = k\}$.
Finally, we write $t \sim_\pf u$ if $\sigma(t) \sim_\pf \sigma(u)$ for all closed substitutions $\sigma$.

\begin{lem}
\label{lem:pf_basic_properties}
Let $t,u$ be two $\bccsp$ terms that do not have $\nil$ summands or factors.
If $t \sim_\ctr u$ then:
\begin{enumerate}
\item \label{lem:pf_basic_properties_vark}
For each $k \ge 0$ it holds that $\var_k (t) = \var_k(u)$.
\item \label{lem:pf_basic_properties_summand}
$t$ has a summand $x$, for some variable $x$, if and only if $u$ does.
\item \label{lem:pf_basic_properties_norm_depth}
$\norm(t) = \norm(u)$ and $\depth(t) = \depth(u)$. 
\end{enumerate}
\end{lem}

\begin{proof}
\begin{enumerate}
\item Assume, towards a contradiction, that for some $k \ge 0$ there is a variable $x$ such that $x \in \var_k(t) \setminus \var_k(u)$.
In particular, this means that there is a term $t'$ such that $t \trans[\alpha] t'$ for some trace $\alpha$ with $|\alpha| = k$ and $x \in \var(t')$.
However, there is no $u'$ such that $u \trans[\beta] u'$ for some trace $\beta$ with $|\beta| = k$ and $x \in \var(u')$.
We can assume, without loss of generality, that $k$ is the largest natural number such that $x \in \var_k(t)$.

Let $n > \depth(t),\depth(u)$.
Consider the closed substitution $\sigma$ defined by
\[
\sigma(y) =
\begin{cases}
a^n & \text{ if } y = x \\
\nil & \text{ otherwise.}
\end{cases}
\]
Then $t \sim_\ctr u$ implies $t \sim_\tr u$ and thus $t \trans[\alpha] t'$ implies that $u \trans[\alpha] u'$ for some $u'$.
Let us proceed by a case analysis on the structure of $t'$ to obtain the desired contradiction.
\begin{itemize}
\item $t' = x + w$ for some term $w$.
Then we get that $\sigma(t)$ has $\alpha a^n$ as a completed trace.
However, due to the choices of $n$ and $\sigma$, we have that $\sigma(u)$ cannot perform the completed trace $a^n$ after $\alpha$, thus contradicting $t \sim_\ctr u$.
By the assumption that $k$ is the largest natural number such that $x \in \var_k(t)$ is suffices to consider two cases:
\item $t' = (x + w) \parcomp w'$ for some terms $w,w'$ with $w' \neq \nil$.
From $t \sim_\tr u$, we infer that for all $\beta \in \ctr(\sigma(w'))$ we have that $\alpha\beta$ is a trace of $\sigma(t)$ and thus of $\sigma(u)$.
However, we can proceed as in the previous case and argue that $\alpha a^n \beta \in \ctr(t) \setminus \ctr(u)$, which is a contradiction with $t \sim_\ctr u$.
\end{itemize}

\item Assume now that $t$ has a summand $x$ for some variable $x$.
Then, from item~\eqref{lem:pf_basic_properties_vark} of this lemma, $x \in \var_0(t)$ implies $x \in \var_0(u)$.
By applying similar reasoning as in the proof of the first statement of the lemma we obtain that it cannot be the case that all the occurrences of $x$ in $u$ are in the scope of prefixing.
Hence, $u$ has, at least, one summand $u'$ such that either $u' = x$ or $u' = (x + w) \parcomp w'$, for some terms $w,w'$ with $w' \neq \nil$.
Our order of business is to show that it cannot be the case that all summands $u'$ of $u$ with $x \in \var_0(u')$ are of the latter form. 
To this end, assume, towards a contradiction, that whenever $x \in \var_0(u')$ then $u' = (x + w) \parcomp w'$.
Let $n > depth(u)$ and consider the closed substitution $\sigma'$ defined by
\[
\sigma'(y) =
\begin{cases}
a^n & \text{ if } y = x \\
b & \text{ otherwise.}
\end{cases}
\]
Since $t$ has a summand $x$, we obtain that $a^n \in \ctr(\sigma'(t))$.
However, $a^n \not\in\ctr(\sigma'(u))$.
In fact, $w' \neq \nil$ and $\sigma'(y) \neq \nil$ for all variables $y \in \Var$ give that $\init(\sigma'(w')) \neq \emptyset$ for all terms $w'$ occurring in the summands $u'$ of $u$ such that $x \in \var_0(u')$.
Therefore $a^n \not \in \ctr(\sigma'(u'))$ for all such summands.
Due to choices of $\sigma'$ and $n$, we can conclude that there is no summand of $\sigma'(u)$ having $a^n$ has a completed trace, so that $a^n \not \in \ctr(\sigma'(u))$.
This gives a contradiction with $\sigma'(t) \sim_\ctr \sigma'(u)$ and thus with $t \sim_\ctr u$.

\item $\norm(t) = \norm(u)$ and $\depth(t) = \depth(u)$ follow immediately from $t \sim_\ctr u$ and items~\eqref{lem:pf_basic_properties_vark}, \eqref{lem:pf_basic_properties_summand} of this lemma. \qedhere
\end{enumerate}
\end{proof}

\begin{rem}
Since for all $\bccsp$ terms $t,u$ we have that $t \sim_\pf u$ implies $t \sim_\ctr u$, Lemma~\ref{lem:pf_basic_properties} holds also for all pairs of possible futures equivalent terms.
\end{rem}

\begin{prop}
\label{prop:pf_parcomp}
Assume that $p,q$ are two $\bccsp$ processes such that $p,q \not \sim_\pf \nil$, $p,q$ do not have $\nil$ summands or factors, and $p \parcomp q \sim_\pf a \parcomp p_N$, for some $N > 1$.
Then either $p \sim_\pf a$ and $q \sim_\pf p_N$, or $p \sim_\pf p_N$ and $q \sim_\pf a$.
\end{prop}

\begin{proof}
By Lemma~\ref{lem:pf_basic_properties}.\ref{lem:pf_basic_properties_norm_depth}, from 
\begin{equation}
\label{eq:prop_pf_parcomp}
p \parcomp q \sim_\pf a \parcomp p_N
\end{equation}
we can infer that $\norm(p \parcomp q) = 3$.
Hence, since $p,q \not\sim_\pf \nil$, we have that either $\norm(p) = 1$ and $\norm(q) = 2$, or vice versa.
We consider only the former case and prove that then $p \sim_\pf a$ and $q \sim_\pf p_N$.
In the other case it can be proved that $p \sim_\pf p_N$ and $q \sim_\pf a$ by analogous reasoning.

So assume that $\norm(p) = 1$ and $\norm(q) = 2$.
From Equation~\eqref{eq:prop_pf_parcomp} we get that $\init(p),\init(q) \subseteq \{a,b\}$.
First of all we argue that it cannot be the case that $a \in \init(p) \cap \init(q)$.
In fact, in that case $p \parcomp q$ would be able to perform the trace $a^2$, whereas $a \parcomp p_N$ cannot do so, thus giving a contradiction with Equation~\eqref{eq:prop_pf_parcomp}.
We now proceed to show that $a \in \init(p)$.
Assume now, towards a contradiction, that this is not the case and $\init(p) = \{b\}$.
This, together with $\norm(p) = 1$, gives that $p$ has a summand $b$.
Hence, given any $\alpha \in \ctr(q)$ we have that $\alpha b \in \ctr(p\parcomp q)$.
However, $\alpha b \not \in \ctr(a \parcomp p_N)$, as any completed trace of $a \parcomp p_N$ must end with an $a$, thus giving that $p \parcomp q \not\sim_\ctr a \parcomp p_N$, which contradicts Equation~\eqref{eq:prop_pf_parcomp}.
We have therefore obtained that $a \in \init(p), a \not \in \init(q)$ and, moreover, $p$ does not have a summand $b$.
Since $p$ has norm $1$, we conclude that it has a summand $a$.

We now proceed to show that also $\depth(p) = 1$, and thus $p \sim_\pf a$.
Since $p$ has a summand $a$, we have that $p \parcomp q \trans[a] \nil \parcomp q \sim_\pf q$.
By Equation~\eqref{eq:prop_pf_parcomp}, we get $a \parcomp p_N \trans[a] r$ for some $r$ such that $q \sim_\tr r$.
Since $a \parcomp p_N$ has a unique initial $a$-transition $a \parcomp p_N \trans[a] \nil \parcomp p_N$, we get that $r \sim_\pf p_N$ and $q \sim_\tr p_N$.
As a consequence, we obtain that $\depth(q) = N+1$.
Then we have
\begin{align*}
\depth(p) ={} & \depth(p \parcomp q) - \depth(q) \\
={} & \depth(a \parcomp p_N) - \depth(q) & \text{(by Eq.~\eqref{eq:prop_pf_parcomp} and Lem.~\ref{lem:pf_basic_properties}.\ref{lem:pf_basic_properties_norm_depth})} \\
={} & N+2 - (N+1) \\
={} & 1.
\end{align*}
Therefore, we have obtained that $p \sim_\pf a$ and thus, since possible futures equivalence is a congruence with respect to parallel composition, from Equation~\eqref{eq:prop_pf_parcomp} it follows that
\begin{equation}
\label{eq:prop_pf_parcomp_2}
a \parcomp q \sim_\pf a \parcomp p_N.
\end{equation}
Our order of business will now be to show that $q \sim_\pf p_N$.
We already know that $\tr(q) = \tr(p_N)$.
Moreover, as an initial $a$-transition of $a \parcomp q$ cannot stem from $q$, from $(ab^i a,\emptyset) \in \pf(a \parcomp p_N)$ for all $i \in \{1,\dots,N\}$, we get that $\ctr(p_N) \subseteq \ctr(q)$.
Assume, towards a contradiction, that $\ctr(q) \not\subseteq \ctr(p_N)$.
As $q \sim_\tr p_N$, this implies that $q$ has a completed trace of the form $b^j$ for some $1 \le j \le N$.
Hence, $(ab^j,\emptyset) \in \pf(a \parcomp q)$.
However, no completed trace of $a \parcomp p_N$ has a $b$ as last symbol and thus $a \parcomp p_N$ has no possible future of the form $(\alpha b,\emptyset)$ for some trace $\alpha$.
This gives a contradiction with Equation~\eqref{eq:prop_pf_parcomp_2}.
Therefore, we have that $\tr(q) = \tr(p_N)$ and $\ctr(q) = \ctr(p_N)$, thus giving $q \sim_\ctr p_N$.
Notice that, by Lemma~\ref{lem:ct_closed_pn}, this implies that $q$ does not contain any occurrence of $\parcomp$ and $q$ can be written in the general form $q = \sum_{j \in J} q_j$ where none of the $q_j$ has $+$ as head operator.

Assume, towards a contradiction, that $q \not\sim_\pf p_N$.
By Lemma~\ref{lem:pf_closed_pn} this implies that either there is an $i \in \{1,\dots, N\}$ such that $b^i a \not\sim_\pf q_j$ for all $j \in J$, or there is a $j \in J$ such that $q_j \not\sim_\pf b^i a$ for all $i \in \{1,\dots,N\}$.
In both cases, we can proceed as in the proof of Lemma~\ref{lem:pf_closed_pn} and obtain a contradiction with $a \parcomp p_N \sim_\pf a \parcomp q$.

We have therefore obtained that $q \sim_\pf p_N$, and the proof is thus concluded.
\end{proof}

The following lemma characterises the open $\bccsp$ terms whose substitution instances can be equivalent in possible futures semantics to terms having at least two summands of $p_N$ ($N > 1$) as their summands. 

\begin{lem}
\label{lem:pf_summand_var}
Let $t$ be a $\bccsp$ term that does not have $+$ as head operator.
Let $m > 1$ and $\sigma$ be a closed substitution such that $\sigma(t)$ has no $\nil$ summands or factors.
If $\sigma(t) \sim_\pf \sum_{k = 1}^m b^{i_k} a$, for some $1\le i_1< \dots < i_m$, then $t = x$ for some variable $x$.
\end{lem}

\begin{proof}
For simplicity of notation, we let $q_m$ denote $\sum_{k = 1}^m b^{i_k} a$.
We proceed by showing that the remaining possible forms for $t$ give a contradiction with $\sigma(t) \sim_\pf q_m$.
We remark that $\init(q_m) = \{b\}$, and $\sigma(t) \sim_\pf q_m$ implies $\ctr(\sigma(t)) = \ctr(q_m)$.
\begin{itemize}
\item $t = b.t'$ for some term $t'$.
Then $\ctr(\sigma(t')) = \{b^{i_k-1}a \mid k \in \{1,\dots,m\} \}$, and since $m>1$ we have that $|\ctr(t')| \ge 2$.
It is then immediate to verify that $(b,\tr(\sigma(t'))) \in \pf(\sigma(t))$, whereas $(b,\tr(\sigma(t'))) \not \in \pf(q_m)$, since whenever $q_m \trans[b] r$ then $\tr(r)$ includes only one trace of the form $b^j a$ for some $j \in \{i_1-1,\dots,i_m-1\}$.
This gives a contradiction with $\sigma(t) \sim_\pf q_m$. 

\item $t = t' \parcomp t''$ for some terms $t',t''$.
Since $\sigma(t)$ has no $\nil$ summands or factors, neither does $t$ thus giving $t',t'',\sigma(t'),\sigma(t'') \not \sim_\pf \nil$.
Hence, we directly get a contradiction with $\ctr(\sigma(t)) = \ctr(q_m)$, since all completed traces of $q_m$ have exactly one occurrence of $a$ and this occurrence is as the last action in the completed trace.
No process of the form $p \parcomp q$ with $\init(p),\init(q) \neq \emptyset$ can satisfy the same property. \qedhere
\end{itemize}
\end{proof}

We now have all the ingredients necessary to prove Theorem~\ref{thm:pf_negative}.
To streamline our presentation, we split the proof of into two parts: Proposition~\ref{prop:pf_substitution} deals with the preservation of the witness property under provability from the substitution rule of equational logic.
Theorem~\ref{thm:pf_axiom_derivation} builds on Proposition~\ref{prop:pf_substitution} and proves the witness property to be an invariant under provability from finite sound axiom systems.

The following lemma is immediate.

\begin{lem}
\label{lem:depth_var}
Let $t$ be a $\bccsp$ term, and let $\sigma$ be a closed substitution. 
If $x\in\var(t)$ then $\depth(\sigma(t))\geq \depth(\sigma(x))$.
\end{lem}

\begin{prop}
\label{prop:pf_substitution}
Let $t \approx u$ be an equation over $\bccsp$ that is sound modulo $\sim_\pf$. 
Let $\sigma$ be a closed substitution with $p=\sigma(t)$ and $q=\sigma(u)$. 
Suppose that $p$ and $q$ have neither $\nil$ summands nor $\nil$ factors, and that $p,q \sim_\pf a \parcomp p_N$ for some $N$ larger than the sizes of $t$ and $u$. 
If $p$ has a summand possible futures equivalent to $a \parcomp p_N$, then so does $q$.
\end{prop}

\begin{proof}
Since $\sigma(t)$ and $\sigma(u)$ have no $\nil$ summands or factors, neither do $t$ and $u$.  
We can therefore assume that, for some finite non-empty index sets $I, J$,
\begin{equation}
\label{eq:t_sum}
t = \sum_{i\in I} t_i 
\quad \text{ and } \quad 
u  = \sum_{j\in J} u_j,
\end{equation}
where none of the $t_i$ ($i \in I$) and $u_j$ ($j \in J$) is $\nil$ or has $+$ as its head operator. 
Note that, as $t$ and $u$ have no $\nil$ summands or factors, neither do $t_i$ ($i \in I$) and $u_j$ ($j \in J$).

Since $p=\sigma(t)$ has a summand which is possible futures equivalent to $a \parcomp p_N$, there is an index $i\in I$ such that
\[
\sigma(t_i) \sim_\pf a \parcomp p_N. 
\]
Our aim is now to show that there is an index $j \in J$ such that 
\[
\sigma(u_j) \sim_\pf a \parcomp p_N, 
\]
proving that $q=\sigma(u)$ has the required summand.
This we proceed to do by a case analysis on the form $t_i$ may have.
\begin{enumerate}
\item \label{case:t_variable} 
{\sc Case $t_i = x$ for some variable $x$}.
In this case, we have that $\sigma(x)\sim_\pf a \parcomp p_N$ and $t$ has $x$ as a summand. 
As $t \approx u$ is sound with respect to possible futures equivalence, from $t \sim_\pf u$ we get $t \sim_\ctr u$.
Hence, by Lemma~\ref{lem:pf_basic_properties}.\ref{lem:pf_basic_properties_summand}, we obtain that $u$ has a summand $x$ as well,
namely there is an index $j \in J$ such that $u_j = x$.
It is then immediate to conclude that $q=\sigma(u)$ has a summand which is possible futures equivalent to $a \parcomp p_N$.

\item \label{case:t_prefixing} 
{\sc Case $t_i = c t'$ for some action $c \in \{a,b\}$ and term $t'$}. 
This case is vacuous because, since $\sigma(t_i) = c \sigma(t') \trans[c] \sigma(t')$ is the only transition afforded by $\sigma(t_i)$, this term cannot be possible futures equivalent to $a \parcomp p_N$. 

\item \label{case:t_parallel} 
{\sc Case $t_i = t' \parcomp t''$ for some terms $t',t''$}. 
We have that $\sigma(t_i) = \sigma(t') \parcomp \sigma(t'') \sim_\pf a \parcomp p_N$.
As $\sigma(t_i)$ has no $\nil$ factors, it follows that $\sigma(t')\not\sim_\pf \nil$ and $\sigma(t'')\not\sim_\pf \nil$.  
Thus, by Proposition~\ref{prop:pf_parcomp}, we can infer that, without loss of generality, 
\[
\sigma(t') \sim_\pf a \text{ and } \sigma(t'')\sim_\pf p_N.
\] 
Notice that $\sigma(t'') \sim_\pf p_N$ implies $\ctr(\sigma(t'')) = \ctr(p_N)$.
Now, $t''$ can be written in the general form $t'' = v_1 + \cdots + v_l$ for some $l >0$, where none of the summands $v_h$ is $\nil$ or a sum.
By Lemma~\ref{lem:pf_closed_pn}, $\sigma(t'') \sim_\pf p_N$ implies that for each $i \in \{1,\dots,N\}$ there is a summand $r_i$ of $\sigma(t'')$ such that $b^i a \sim_\pf r_i$, and for each summand $r$ of $\sigma(t'')$ there is an $i_r \in \{1,\dots,N\}$ such that $r \sim_\pf b^i a$.
Observe that, since $N$ is larger than the size of $t$, we have that $l < N$.
Hence, there must be some $h\in\{1,\ldots,l\}$ such that 
\[
\sigma(v_h) \sim_\pf \sum_{k = 1}^m b^{i_k} a 
\]
for some $m>1$ and $1\le i_1 < \ldots < i_m\le N$. 
The term $\sigma(v_h)$ has no $\nil$ summands or factors, or else, so would $\sigma(t'')$ and $\sigma(t)$. 
By Lemma~\ref{lem:pf_summand_var}, it follows that $v_h$ can only be a variable $x$ and
\begin{equation}
\label{eq:sigma(x)}
\sigma(x) \sim_\pf  \sum_{k=1}^m b^{i_k} a.
\end{equation}
Observe, for later use, that, since $t'$ has no $\nil$ factors, the above equation yields that $x\not\in\var(t')$, or else
$\sigma(t')\not \sim_\pf a$ due to Lemma~\ref{lem:depth_var}. 
So, modulo possible futures equivalence, $t_i$ has the form $t' \parcomp (x+t''')$, for some term $t'''$, with $x\not\in\var(t')$, $\sigma(t')\sim_\pf a$ and $\sigma(x + t'') \sim_\pf p_N$.
    
Our order of business will now be to show that $\sigma(u)$ has a summand $u_j$ that is possible futures equivalent to $a \parcomp p_N$.
We recall that $t \sim_\pf u$ implies $t \sim_\ctr u$.
Thus, by Lemma~\ref{lem:pf_basic_properties}.\ref{lem:pf_basic_properties_vark} we obtain that $\var_k(t) = \var_k(u)$ for all $k \ge 0$.
Hence, from $x \in \var_0(t_i)$ we get that there is at least one $j \in J$ such that $x \in \var_0(u_j)$.

So, firstly, we show that $x$ cannot occur in the scope of prefixing in $u_j$, namely $u_j$ cannot be of the form $c.u'$ or $(c.u' + u'') \parcomp u'''$ for some $c \in \{a,b\}$ and $u'$ with $x \in \var(u')$.
We proceed by a case analysis:
\begin{enumerate}
\item $c = b$ and $u_j = (b.u' + u'') \parcomp u'''$ for some $u',u'',u''' \in \bccsp$ with $x \in \var(u')$.
As $\sigma(u)$ does not have $\nil$ summands or factors we have that $\sigma(u''') \not\sim_\pf \nil$.
Let $D = \max\{d \mid x \in \var_d(u')\}$.
From $\sigma(x) \sim_\pf \sum_{k = 1}^m b^{i_k} a$ (Equation~\eqref{eq:sigma(x)}) and $\ctr(\sigma(u)) = \ctr(a \parcomp p_N)$ we can infer that the completed traces of $\sigma(u''')$ are of the form $b^i a$, for some $i \in \{0,\dots,N-i_m-D-1\}$.
In fact, since $\sigma(x)$ can perform at least one completed trace of the form $b^{i_k}a$, for some $1 \le i_k \le N$, and the completed traces of $a \parcomp p_N$ contain exactly two occurrences of $a$, of which one as the final action of the trace, we can infer that the completed traces of $\sigma(u''')$ have to contain exactly one occurrence of $a$, and this occurrence has to be as the last symbol of the completed trace. 
Let $\alpha \in \tr(\sigma(u'))$ be such that $|\alpha| = D$ and $u' \trans[\alpha] w$ with $x \in \var(w)$.
By the choice of $D$, we can infer that $x$ does not occur in the scope of prefixing in $w$, and thus $\tr(\sigma(x)) \subseteq \tr(\sigma(w))$.
Then we get that $(b^iab\alpha,\tr(\sigma(w))) \in \pf(\sigma(u))$, where $b^ia \in \ctr(\sigma(u'''))$.
However, as $m \ge 2$, there is no $p'$ such that $a \parcomp p_N \trans[b^iab\alpha] p'$ and $\tr(\sigma(x)) \subseteq \tr(p')$, thus giving $(b^iab\alpha,\tr(\sigma(w))) \not\in \pf(a \parcomp p_N)$.
This gives a contradiction with $\sigma(u) \sim_\pf a \parcomp p_N$.
\item $c = b$ and $u_j = b.u'$ for some $\bccsp$ term $u'$ with $x \in \var(u')$.
The proof of this case is similar to, actually simpler than, that of the previous case and it is therefore omitted.
\item $c = a$ and $u_j = (a.u' + u'') \parcomp u'''$ for some $u',u'',u''' \in \bccsp$ with $x \in \var(u')$.
As $\sigma(u)$ does not have $\nil$ summands or factors we have that $\sigma(u''') \not\sim_\pf \nil$.
From $\sigma(x) \sim_\pf \sum_{k = 1}^m b^{i_k} a$ we infer that $\tr(a.\sigma(u'))$ includes traces having two occurrences of action $a$.
Since $\sigma(u) \sim_\pf a \parcomp p_N$, this implies that there is no $\alpha \in \tr(\sigma(u'''))$ such that $\alpha$ contains an occurrence of action $a$, for otherwise $\sigma(u)$ could perform a trace having 3 occurrences of that action.
In particular, this implies that the last symbol in each trace of $\sigma(u''')$ must be action $b$.
This gives that there is at least one completed trace of $\sigma(u_j)$, and thus of $\sigma(u)$, whose last symbol is action $b$.
Hence we get $\ctr(\sigma(u)) \neq \ctr(a \parcomp p_N)$, thus giving a contradiction with $\sigma(u) \sim_\pf a \parcomp p_N$.
\item $c = a$ and $u_j = a.u'$ for some $\bccsp$ term $u'$ with $x \in \var(u')$.
In this case we are going to prove a slightly weaker property, namely that not all summands $u_j$ with $x \in \var(u_j)$ can be of this form.
Despite being weaker, this property is enough because, as shown above, all other possibilities for an occurrence of $x$ in a summand $u_j$ have already been excluded.
To this end, consider the closed substitution $\sigma'$ defined by 
\[
\sigma'(y) = 
\begin{cases}
a p_N & \text{ if } y = x \\
\sigma(y) & \text{ otherwise.}
\end{cases}
\]
Then we have that $\sigma'(t_i) = \sigma'(t') \parcomp \sigma'(x) + \sigma'(t''') \trans[a] \sigma(t') \parcomp p_N \sim_\pf a \parcomp p_N$.
Since $\sigma'(t) \sim_\pf \sigma'(u)$ then there is a process $r$ such that $\sigma'(u) \trans[a] r$ and $\tr(r) = \tr(a \parcomp p_N)$.
In particular, this means that $\depth(r) = N+2$.
Hence, from the choices of $N,\sigma$ and $\sigma'$, we can infer that such an $a$-move by $\sigma'(u)$ can only stem from a summand $u_j$ such that $x \in \var(u_j)$.
Assume, towards a contradiction, that all such summands $u_j$ are of the form $a.u_j'$ for some $\bccsp$ term $u_j'$ with $x \in \var(u_j')$ and $r = \sigma'(u_j')$.
As $\depth(\sigma'(u_j')) = N+2 = \depth(\sigma'(x))$, by Lemma~\ref{lem:depth_var} we get that $u_j'$ can only be of the form $u_j' = x + w_j$ for some $\bccsp$ term $w_j$ with $\depth(\sigma'(w_j)) \le N+2$.
Notice that $\tr(\sigma'(x)) \subset \tr(a \parcomp p_N)$.
Hence $\sigma'(w_j) \neq \nil$.
More precisely, $\sigma'(x) = a p_N$ implies that $\{b\alpha \mid b\alpha \in \tr(a \parcomp p_N)\} \subseteq \tr(\sigma'(w_j)) \subseteq \tr(a \parcomp p_N)$.
Clearly, no trace starting with action $b$ can stem from $\sigma'(x)$ and we can then infer, in light of Lemma~\ref{lem:depth_var}, that $x \not \in \var(w_j)$, as $\depth(\sigma'(w_j)) \le N+2$.
This implies that $\sigma'(w_j) = \sigma(w_j)$ and thus $\{b\alpha \mid b\alpha \in \tr(a \parcomp p_N)\} \subseteq \tr(\sigma(w_j)) \subseteq \tr(a \parcomp p_N)$.
In particular, $\sigma(w_j)$ can perform at least one (completed) trace of the form $b\alpha$ where $\alpha$ contains two occurrences of action $a$.
From $\sigma(u_j) = a.(\sigma(x) + \sigma(w_j))$, we then get that $(ab\alpha,\emptyset) \in \pf(\sigma(u))$, namely $\sigma(u)$ can perform at least one (completed) trace containing 3 occurrences of action $a$.
This gives a contradiction with $\sigma(u) \sim_\pf a \parcomp p_N$.
\end{enumerate}

We have therefore obtained that $x$ does not occur in the scope of prefixing in (at least one) $u_j$.
We proceed now by a case analysis on the possible forms of this summand.
\begin{enumerate}
\item $u_j = x$.
Then $\sigma(u)$ has a summand which is possible futures equivalent to $\sum_{k = 1}^m b^{i_k} a$.
We show that this gives a contradiction with $\sigma(u) \sim_\pf a \parcomp p_N$.
This follows directly by noticing that, due to the summand $b^{i_1}a$, we have that $(b^{i_1}a,\emptyset) \in \pf(\sigma(u))$.
However, $(b^{i_1}a,\emptyset) \not\in\pf(a\parcomp p_N)$, since $a \parcomp p_N$ by performing the trace $b^{i_1}a$ can reach either a process that can perform an $a$ (in case the first $b$-move is performed by the summand $b^{i_1}a$ of $p_N$) or a $b$ (in case the first $b$-move is performed by a summand $b^i a$ of $p_N$ such that $i > i_1$).

\item $u_j = (x + w) \parcomp w'$, for some terms $w,w'$ with $w' \not \sim_\pf \nil$.
From $\sigma(u) \sim_\pf a \parcomp p_N$, we infer that $\ctr(\sigma(u_j)) \subseteq \ctr(a \parcomp p_N)$.
We recall that no completed trace of $a \parcomp p_N$ has $b$ as last symbol and, moreover, in all the completed traces of $a \parcomp p_N$ there are exactly two occurrences of $a$.
Hence, all (nonempty) completed traces of $\sigma(x),\sigma(w)$ and $\sigma(w')$ must have exactly one occurrence of $a$ and this occurrence must be as the last symbol in the completed trace.

We now proceed to show that $\sigma(w')$ has a summand $a$ and $a \not \in \init(\sigma(x) + \sigma(w))$.
We start by noticing that it cannot be the case that $a \in \init(\sigma(x) + \sigma(w)) \cap \init(\sigma(w'))$, for otherwise we would have $a^2 \in \tr(\sigma(u_j)) \subseteq \tr(\sigma(u))$, thus contradicting $\sigma(u) \sim_\pf a \parcomp p_N$.
Assume now, towards a contradiction, that $\init(\sigma(w')) = \{b\}$.
Then, due to summand $b^{i_m}a$ of $\sigma(x)$, we have that $\sigma(u_j) \trans[b^{i_m-1}] ba \parcomp \sigma(w')$ and $a\alpha \not\in\tr(ba \parcomp \sigma(w'))$ for any trace $\alpha \in \Act^{\ast}$.
Clearly, $(b^{i_m-1},\tr(ba \parcomp \sigma(w'))) \in \pf(\sigma(u_j))$, and thus it is also a possible future of $\sigma(u)$.
However, $(b^{i_m-1},\tr(ba \parcomp \sigma(w'))) \not \in \pf(a\parcomp p_N)$, as the interleaving of $p_N$ with $a$ guarantees that after an initial trace of an arbitrary number of $b$-transitions it is always possible to perform a trace starting with $a$.
This gives a contradiction with $\sigma(u) \sim_\pf a \parcomp p_N$.
We have therefore obtained that $a \in \init(\sigma(w'))$.
More precisely, from the constraints on the completed traces of $\sigma(w')$, we can infer that $\sigma(w')$ has a summand $a$.

Our order of business will now be to show that $\sigma(w') \sim_\pf a$.
Since $\sigma(w') \trans[a] \nil$, we have that $\sigma(u_j) \trans[a] (\sigma(x) + \sigma(w)) \parcomp \nil \sim_\pf \sigma(x) + \sigma(w)$.
Thus, $\sigma(u) \sim_\pf a \parcomp p_N$ implies that $a \parcomp p_N \trans[a] r$ for some $r$ with $\tr(r) = \tr(\sigma(x) + \sigma(w))$.
Since $a \parcomp p_N$ has only one possible initial $a$-transition, namely $a \parcomp p_N \trans[a] \nil \parcomp p_N$, we get that $r \sim_\pf p_N$ and thus $\tr(\sigma(x) + \sigma(w)) = \tr(p_N)$.
In particular, this implies that $\depth(\sigma(x) + \sigma(w)) = N+1$.
Therefore, we have
\begin{align*}
1 \le \depth(\sigma(w')) ={} & \depth(\sigma(u_j)) - \depth(\sigma(x)+\sigma(w)) \\
={} & \depth(\sigma(u_j)) - (N + 1) \\
\le{} & \depth(\sigma(u)) - (N + 1) \\
={} & \depth(a \parcomp p_N) - (N + 1) & \text{(by Lem.~\ref{lem:pf_basic_properties}.\ref{lem:pf_basic_properties_norm_depth})}\\
={} & N + 2 - (N + 1) \\
={} 1
\end{align*}
and we can therefore conclude that $\sigma(w') \sim_\pf a$.
Furthermore, it is not difficult to prove that $\ctr(\sigma(x) + \sigma(w)) = \ctr(p_N)$, for otherwise we get a contradiction with $\sigma(u) \sim_\pf a \parcomp p_N$.

So far we have obtained that, modulo possible futures equivalence,
\begin{align*}
&\sigma(u_j) \sim_\pf \left(\sum_{k = 1}^m b^{i_k} a + \sigma(w) \right) \parcomp a \text{ and} \\
&\ctr\left(\sum_{k = 1}^m b^{i_k} a + \sigma(w)\right) = \left\{b^i a \mid i \in \{1,\dots,N\}\right\}.
\end{align*}
To conclude the proof, we need to show that $\sum_{k=1}^m b^{i_k} a + \sigma(w) \sim_\pf p_N$.
Let $I_m = \{i_1,\dots,i_m\}$ and $I_N = \{1,\dots,N\}$.
Assume, towards a contradiction, that $\sum_{k = 1}^m b^{i_k} a + \sigma(w) \not\sim_\pf p_N$.
Since $\ctr(\sigma(x) + \sigma(w)) = \ctr(p_N)$, from Lemma~\ref{lem:ct_closed_pn} we can infer that $\sigma(w)$ does not contain any occurrence of $\parcomp$.
In particular, $\sigma(w)$ can be written in the general form $\sigma(w) = \sum_{l \in L} q_l$ for some terms $q_l$ that do not have $+$ as head operator nor contain any occurrence of $\parcomp$.
Moreover, as $\sum_{k = 1}^m b^{i_k} a + \sigma(w) \not\sim_\pf p_N$, by Lemma~\ref{lem:pf_closed_pn}, this means that either there is an $i \in I_N \setminus I_m$ such that $b^i a \not\sim_\pf q_l$ for any $l \in L$, or that there is a summand $q_l$ of $\sigma(w)$ such that $q_l \not\sim_\pf b^i a$ for any $i \in I_N$.
In both cases, we obtain that there is (at least) a summand $q_l$ of $\sigma(w)$ such that $b^k a, b^h a \in \ctr(q_l)$ for some $k\neq h, h,k \in I_N$.
We can then proceed as in the proof of Lemma~\ref{lem:pf_closed_pn} to prove that this gives the desired contradiction.
We have therefore obtained that $\sum_{k=1}^m b^{i_k} a + \sigma(w) \sim_\pf p_N$.
Hence, since possible futures equivalence is a congruence with respect to parallel composition, we get that
\[
\sigma(u_j) \sim_\pf a \parcomp p_N
\]
and we can therefore conclude that $\sigma(u)$ has the desired summand.
\end{enumerate}
\end{enumerate}
This concludes the proof.
\end{proof}

We can now proceed to prove the witness property to be an invariant under provability from finite sound axiom systems.
Theorem~\ref{thm:pf_negative} can be then obtained as a consequence of the following result.
In fact, as the left-hand side of equation $e_N$, i.e., the term $a \parcomp p_N$, has a summand possible futures equivalent to $a \parcomp p_N$, whilst the right-hand side, i.e., the term $a p_N + \sum_{i = 1}^{N} b (a \parcomp b^{i-1}a)$, does not, we can conclude that the collection of infinitely many equations $e_N$ ($N \ge 1$) is the desired witness family.

\begin{thm}
\label{thm:pf_axiom_derivation}
Let $\E$ be a finite axiom system over $\bccsp$ that is sound modulo $\sim_\pf$.
Let $N$ be larger than the size of each term in the equations in $\E$.
Let $p$ and $q$ be closed terms such that $p,q \sim_\pf a \parcomp p_N$.
Assume, moreover, that $p$ and $q$ contain no occurrences of $\nil$ as a summand or factor.
If $\E \vdash p \approx q$ and $p$ has a summand possible futures equivalent to $a \parcomp p_N$, then so does $q$.
\end{thm}

\begin{proof}
Assume that $\E$ is a finite axiom system over the language $\bccsp$ that is sound modulo possible futures equivalence, and that the statements~\eqref{statement_one}--\eqref{statement_four} below hold, for some closed terms $p$ and $q$ and positive integer $N$ larger than the size of each term in the equations in $\E$:
\begin{enumerate}
\item \label{statement_one}
$E \vdash p \approx q$,
\item $p \sim_\pf q \sim_\pf a \parcomp p_N$, 
\item $p$ and $q$ contain no occurrences of $\nil$ as a summand or factor, and
\item \label{statement_four}
$p$ has a summand possible futures equivalent to $a \parcomp p_N$.
\end{enumerate}
We prove that $q$ also has a summand possible futures equivalent to $a \parcomp p_N$ by induction on the depth of the closed proof of the equation $p \approx q$ from $\E$. 
Recall that, without loss of generality, we may assume that the closed terms involved in the proof of the equation $p \approx q$ have no $\nil$ summands or factors (by Proposition~\ref{prop:saturation}, as $\E$ may be assumed to be saturated), and that applications of symmetry happen first in equational proofs (that is, $\E$ is closed with respect to symmetry). 
  
We proceed by a case analysis on the last rule used in the proof of $p \approx q$ from $\E$. 
The case of reflexivity is trivial, and that of transitivity follows immediately by using the inductive hypothesis twice. 
Below we only consider the other possibilities.
\begin{itemize}
\item {\sc Case $E \vdash p \approx q$, because $\sigma(t)=p$ and $\sigma(u)=q$ for some equation $(t\approx u)\in E$ and closed substitution $\sigma$}. 
Since $\sigma(t)=p$ and $\sigma(u)=q$ have no $\nil$ summands or factors, and $N$ is larger than the size of each term mentioned in equations in $\E$, the claim follows by Proposition~\ref{prop:pf_substitution}.
\item {\sc Case $E \vdash p \approx q$, because $p=c p'$ and $q=c q'$ for some $p',q'$ such that $E \vdash p' \approx q'$, and for some action $c$}.
This case is vacuous because $p=c p'\not\sim_\pf a \parcomp p_N$, and thus $p$ does not have a summand possible futures equivalent to $a \parcomp p_N$.
\item {\sc Case $E \vdash p \approx q$, because $p=p'+p''$ and $q=q'+q''$ for some $p',q',p'',q''$ such that $E \vdash p' \approx q'$ and $E \vdash p'' \approx q''$}.
Since $p$ has a summand possible futures equivalent to $a \parcomp p_N$, we have that so does either $p'$ or $p''$. 
Assume, without loss of generality, that $p'$ has a summand possible futures equivalent to $a \parcomp p_N$. 
Since $p$ is possible futures equivalent to $a \parcomp p_N$, so is $p'$. Using the soundness of $\E$ modulo possible futures equivalence, it follows that $q'\sim_\pf a \parcomp p_N$. 
The inductive hypothesis now yields that $q'$ has a summand possible futures equivalent to $a \parcomp p_N$. 
Hence, $q$ has a summand possible futures equivalent to $a \parcomp p_N$, which was to be shown.
\item {\sc Case $E \vdash p \approx q$, because $p= p' \parcomp p''$ and $q=q' \parcomp q''$ for some $p',q',p'',q''$ such that $E \vdash p' \approx q'$ and $E \vdash p'' \approx q''$}. 
Since the proof involves no uses of $\nil$ as a summand or a factor, we have that $p',p''\not\sim_\pf \nil$ and $q',q''\not\sim_\pf \nil$. 
It follows that $q$ is a summand of itself. 
By our assumptions, $q' \parcomp q'' \sim_\pf a \parcomp p_N$ which, by Proposition~\ref{prop:pf_parcomp} gives that either $q' \sim_\s a$ and $q'' \sim_\s p_N$, or $q' \sim_\s p_N$ and $q'' \sim_\s a$.
In both cases, we can conclude that $q$ has itself as summand of the required form.
\end{itemize}
This completes the proof of Theorem~\ref{thm:pf_axiom_derivation}. 
\end{proof}

The proof of Theorem~\ref{thm:pf_negative} is now concluded.


\subsection{Extending the negative result}
\label{sec:nested_simulation}

It is easy to check that the equations $e_N$ ($N \ge 1$) in the witness family of the negative result for $\sim_\pf$ are all sound modulo bisimilarity, i.e., the largest symmetric simulation.
Consequently, they are also sound modulo any congruence $\rel$ such that ${\sim_\B} \subseteq {\rel} \subseteq {\sim_\pf}$.
Hence, the negative result for all these equivalences can be derived from that for $\sim_\pf$, by exploiting this fact and that any finite axiom system that is sound modulo $\rel$ is also sound modulo $\sim_\pf$.

\begin{thm}
\label{thm:nested_negative}
Assume that ${|\Act|}\ge 2$.
Let $\rel$ be a congruence such that ${\sim_\B} \subseteq {\rel} \subseteq {\sim_\pf}$.
Then $\rel$ has no finite, ground-complete, equational axiomatisation over the language $\bccsp$.
\end{thm}

\begin{proof}
Let $\E$ be a finite equational axiomatisation for $\bccsp$ that is sound modulo $\rel$. 
Since $\rel$ is included in $\sim_\pf$, we have that the axiom system $\E$ is sound modulo $\sim_\pf$.
Let $N$ be larger than the size of each term in the equations in $\E$.
Theorem~\ref{thm:pf_axiom_derivation} implies that
the equation
\[
a \parcomp p_N \approx a p_N + \sum_{i = 1}^N b(a \parcomp b^{i-1}a)
\]
cannot be derived from $\E$.
Since this equation is sound modulo $\rel$, namely
\[
a \parcomp p_N \,\rel\, a p_N + \sum_{i = 1}^N b(a \parcomp b^{i-1}a)
\]
it follows that $\E$ is not complete modulo $\rel$.
\end{proof}

Theorem~\ref{thm:nested_negative} can be applied to establish for $n \ge 2$ that the $n$-nested trace and simulation semantics have no finite, ground-complete equational axiomatisation over $\bccsp$.
The $n$-nested trace equivalences were introduced in~\cite{HM85} as an alternative tool to define bisimilarity.
The hierarchy of $n$-nested simulations, namely simulation relations contained in a (nested) simulation equivalence, was introduced in~\cite{GV92}.

\begin{defi}
[$n$-nested semantics]
For $n \ge 0$, the relation $\sim^n_\tr$ over $\proc$, called the $n$-\emph{nested trace equivalence}, is defined inductively as follows:
\begin{itemize}
\item $p \sim^0_\tr q$ for all $p,q \in \proc$,
\item $p \sim^{n+1}_\tr q$ if and only if for all traces $\alpha \in \Act^{*}$:
\begin{itemize}
\item if $p \trans[\alpha] p'$ then there is a $q'$ such that $q \trans[\alpha] q'$ and $p' \sim^n_\tr q'$, and
\item if $q \trans[\alpha] q'$ then there is a $p'$ such that $p \trans[\alpha] p'$ and $p' \sim^n_\tr q'$.
\end{itemize}
\end{itemize}
For $n \ge 0$, the relation $\sqsubseteq^n_\s$ over $\proc$ is defined inductively as follows:
\begin{itemize}
\item $p \sqsubseteq^0_\s q$ for all $p,q \in \proc$,
\item $p \sqsubseteq^{n+1}_\s q$ if and only if $p \rel q$ for some simulation $\rel$, with $\rel^{-1}$ included in $\sqsubseteq^n_\s$.
\end{itemize}
$n$-\emph{nested simulation equivalence} is the kernel of $\sqsubseteq^n_\s$, i.e., the equivalence ${\sim^n_\s} = {\sqsubseteq^n_\s} \cap {(\sqsubseteq^n_\s)^{-1}}$. 
\end{defi}

Note that $\sim^1_\tr$ corresponds to trace equivalence, $\sim^2_\tr$ is possible futures equivalence, and $\sim^1_\s$ is simulation equivalence.
The following theorem is a corollary of Theorems~\ref{thm:pf_negative} and~\ref{thm:nested_negative}.

\begin{thm}
Assume that ${|\Act|}\ge 2$.
Let $n \ge 2$.
Then, $n$-nested trace equivalence and $n$-nested simulation equivalence admit no finite, ground-complete, equational axiomatisation over the language $\bccsp$.
\end{thm}


\section{Adding CCS synchronisation}
\label{sec:communication}

The negative results provided above conclude our analysis of the axiomatisability of the purely interleaving parallel composition operator modulo the congruences in the linear time-branching time spectrum.

The most natural extension of our work consists in allowing the parallel components to synchronise.
In particular, we are interested in the \emph{CCS-style communication}.
It presupposes a bijection $\overline{\cdot}$ on $\Act$ such that $\overline{\overline{a}} = a$ and $\overline{a} \neq a$ for all $a \in \Act$.
Following~\cite{M89}, the special action symbol $\tau \not \in \Act$, will result from the synchronised occurrence of the complementary actions $a$ and $\bar{a}$. 
Let $\Act_\tau = \Act \cup \{\tau\}$.
Then, we let the metavariables $\alpha, \beta, \dots$ range over $\Act_\tau$.

The rules in Table~\ref{tab:semantics_communication} define the operational semantics of the parallel composition operator when also synchronisation is taken into account.

\begin{table}[t]
\centering
\begin{gather*}
\inference{x \trans[\alpha] x'}{x \parcomp y \trans[\alpha] x' \parcomp y} \quad
\inference{y \trans[\alpha] y'}{x \parcomp y \trans[\alpha] x \parcomp y'} \quad
\inference{x \trans[a] x' \quad y \trans[\bar{a}] y'}{x \parcomp y \trans[\tau] x' \parcomp y'} \quad
\inference{x \trans[\bar{a}] x' \quad y \trans[a] y'}{x \parcomp y \trans[\tau] x' \parcomp y'}
\end{gather*}
\caption{\label{tab:semantics_communication} Operational semantics of parallel composition with CCS communication.}
\end{table}

Our order of business for this section will then be to show if and how the analysis we carried out in the previous sections is affected by the addition of synchronisation \`a la CCS.


\subsection{The positive results}

It is not difficult to see that the arguments we used to prove the existence of finite, ground-complete axiomatisations still hold also with synchronisation.
The only changes we need to apply are reported in Table~\ref{tab:axioms_communication}.
For $\textrm{Y} \in \{1,2,3\}$, the axiom schema ELCY simply adds the terms related to communication to ELY.

\begin{table}[t]
\centering
\begin{tabular}{l}
\hline \\[-.3cm]
\, \scalebox{0.85}{(ELC1)}\; $\alpha x \parcomp \beta y \approx \alpha(x \parcomp \beta y) + \beta(\alpha x \parcomp y)$ \quad if $\alpha \neq \overline{\beta}$, or $\alpha = \tau$, or $\beta = \tau$\\
\\[-.3cm]
\, \scalebox{0.85}{(ELC1$\tau$)}\; $\alpha x \parcomp \beta y \approx \alpha(x \parcomp \beta y) + \beta(\alpha x \parcomp y) + \tau(x \parcomp y)$ \quad if $\alpha = \overline{\beta}$\\
\\[-.3cm]
\, \scalebox{0.85}{(ELC2)}\; $\sum_{i \in I} \alpha_ix_i \parcomp \sum_{j \in J} \beta_j y_j \approx \sum_{i \in I} \alpha_i (x_i \parcomp \sum_{j \in J} \beta_j y_j) + \sum_{j \in J} \beta_j (\sum_{i \in I} \alpha_i x_i \parcomp y_j) +$\\[.1cm]
\qquad\qquad\phantom{$\sum_{i \in I} \alpha_ix_i \parcomp \sum_{j \in J} \beta_j y_j \approx$} $+ \sum_{i \in I, j \in J \atop \alpha_i = \overline{\beta_j}} \tau(x_i \parcomp y_j)$\\
\qquad with $\alpha_i \neq \alpha_k$ whenever $i \neq k$ and $\beta_j \neq \beta_h$ whenever $j \neq h$, $\forall\, i,k \in I, \forall\, j,h \in J$\\
\\[-.3cm]
\, \scalebox{0.85}{(ELC3)}\; $\sum_{i \in I} \alpha_ix_i \parcomp \sum_{j \in J} \beta_j y_j \approx \sum_{i \in I} \alpha_i (x_i \parcomp \sum_{j \in J} \beta_j y_j) + \sum_{j \in J} \beta_j (\sum_{i \in I} \alpha_i x_i \parcomp y_j) +$ \\[.1cm]
\qquad\qquad\phantom{$\sum_{i \in I} \alpha_ix_i \parcomp \sum_{j \in J} \beta_j y_j \approx$} $+ \sum_{i \in I, j \in J \atop \alpha_i = \overline{\beta_j}} \tau(x_i \parcomp y_j)$
\\[.1cm]
\hline
\\[-.2cm] 
\end{tabular}
\caption{\label{tab:axioms_communication} The different instantiations of the expansion law when communication is considered.}
\end{table}

We remark that, by convention, $\sum_{i \in \emptyset} t_i = \nil$.
A finite, ground-complete axiomatisation over $\bccsp$ modulo $\sim_{\mathtt{X}}$, for $\mathtt{X} \in \{\tr,\ctr,\fail,\ready,\ftr,\rtr,\s,\cs,\rs\}$ is given by the axiom system $\E^c_{\mathtt{X}}$ obtained from $\E_{\mathtt{X}}$ as follows:
\begin{itemize}
\item we include $\{\textrm{ELC1}, \textrm{ELC1}\tau \}$ instead of EL1;
\item we include ELC2 instead of EL2;
\item all other axioms are unchanged (although the action variables occurring in them now range over $\Act_{\tau}$).
\end{itemize}
For instance, a finite, ground-complete axiomatisation over $\bccsp$ modulo ready simulation equivalence is given by the axiom system $\E^c_\rs = \E_1 \cup \{\textrm{RS},\textrm{RSP1}, \textrm{RSP2},\textrm{ELC2}\}$.


\subsection{The negative results}

In~\cite{Mol89} it was proved that bisimilarity does not admit a finite, ground-complete axiomatisation over $\bccsp$ with CCS-style synchronisation.
We now proceed to show that, by applying similar arguments to those used in Section~\ref{sec:possible_futures}, we can obtain the same negative result for possible futures equivalence.
More precisely, we prove the following:

\begin{thm}
\label{thm:pf_negative_communication}
Assume that ${|\Act_\tau|}\ge 2$.
Possible futures equivalence has no finite, ground-complete, equational axiomatisation over the language $\bccsp$ with CCS synchronisation.
\end{thm}

To this end, consider the infinite family of equations $\{e^c_N \mid N \ge 1\}$ given by:
\begin{align*}
& q_N = \sum_{i = 1}^N \tau^i a 
& (N \ge 1)\phantom{.} \\
& e^c_N \; \colon \; a \parcomp q_N \approx a q_N + \sum_{i = 1}^{N} \tau (a \parcomp \tau^{i-1}a)
& (N \ge 1).
\end{align*}
Clearly, this family of equations coincides with that used to prove the negative result in Section~\ref{sec:possible_futures} where we have substituted all occurrences of action $b$ with the special action $\tau$.
We notice that the equations $e^c_N$ are sound not only modulo possible futures equivalence, but also modulo bisimilarity (each $e^c_N$ is in fact a distinct closed instance of ELC3 in Table~\ref{tab:axioms_communication}).
This means that if we can obtain the negative result for possible futures, then we can 
proceed exactly as in Section~\ref{sec:nested_simulation} to extend it to all the congruences $\sim$ such that $\sim_\B \subseteq \sim \subseteq \sim_\pf$.

Since $\tau$ cannot communicate with any action, hence, in particular, it does not communicate with $a$, the analysis we carried out in Section~\ref{sec:possible_futures} (with $b$ replaced by $\tau$) still applies.

\begin{rem}
\label{rmk:nil_summands_communication}
Most results in Section~\ref{sec:nested_simulation} rely on the assumption that terms do not contain $\nil$ summands or factors.
Notice that due to undefined communications, the expansion law schemas in Table~\ref{tab:axioms_communication} may actually introduce terms that are possible futures equivalent to $\nil$ in the equational proofs.

However, we remark that this is not an issue.
In fact, in general, we could extend the results in Appendix~\ref{app:saturation} to deal with the possible $\nil$ summands introduced by the expansion laws.
Briefly, whenever a $\nil$ summand is introduced, we can assume that axioms 
A0 and P0 (Table~\ref{tab:basic-axioms}) are also applied in order to get rid of the unnecessary $\nil$ summands and factors.
Hence, the proof of Theorem~\ref{thm:pf_negative_communication} follows from Theorem~\ref{thm:pf_axiom_derivation} and the fact that none of the summands in the right-hand side of the equations $e^c_N$ is possible futures equivalent to $a \parcomp p_N$.
\end{rem}


\section{Concluding remarks}
\label{sec:conclusion}

We have studied the finite axiomatisability of the language $\bccsp$ modulo the behavioural equivalences in the linear time-branching time spectrum.
On the one hand we have obtained finite, ground-complete axiomatisations modulo the (decorated) trace and simulation semantics in the spectrum.
On the other hand we have proved that for all equivalences that are finer than possible futures equivalence and coarser than bisimilarity a finite ground-complete axiomatisation does not exist.

Since our ground-completeness proof for ready simulation equivalence proceeds via elimination of $\|$ from closed terms (Proposition~\ref{prop:rs_elim}), and all behavioural equivalences in the linear time-branching time spectrum that include ready simulation have a finite ground-complete axiomatisation over BCCSP, it immediately follows from the elimination result that all these behavioural equivalences have a finite ground-complete axiomatisation over $\bccsp$.

Exploiting various forms of distributivity of parallel composition over choice, we were able to present more concise and elegant axiomatisations for the coarser behavioural equivalences. 

In this paper we have considered both, a parallel composition operator that implements interleaving without synchronisation between the parallel components, and a parallel composition operator with CCS-style synchronisation.
It is natural to consider extensions of our result to parallel composition operators with other forms of synchronisation. 
We expect the extensions with ACP-style or CSP-style synchronisation to be less straightforward than the extension with CCS-style synchronisation presented in this paper, especially in the case of the negative results, and we leave these as topics for future investigations.

As previously outlined, in~\cite{AFGI04} it was proved that the nested semantics admit no finite axiomatisation over BCCSP.
However, our negative results cannot be reduced to a mere lifting of those in~\cite{AFGI04}, as the presence of the additional operator $\parcomp$ might, at least in principle, allow us to finitely axiomatise the equations over BCCSP processes that are valid modulo the considered nested semantics. 
Indeed, auxiliary operators can be added to a language in order to obtain a finite axiomatisation of some congruence relation (see, e.g. the classic example given in~\cite{BK84b}). 
Understanding whether it is possible to lift non-finite axiomatisability results among different algebras, and under which constraints this can be done, is an interesting research avenue and we aim to investigate it in future work. 
A methodology for transferring non-finite-axiomatisability results across languages was presented  in~\cite{AFIM10}, where a reduction-based approach was proposed. 
However, that method has some limitations and thus further studies are needed.

A behavioural equivalence is \emph{finitely based} if it has a finite equational axiomatisation from which all valid equations between open terms are derivable. 
In~\cite{FL00} and~\cite{AFIL09} finite bases for bisimilarity with respect to PA and $\bccsp$ extended with the auxiliary operators left merge and communication merge were presented.
Furthermore, in~\cite{CFLN08} an overview was given of which behavioural equivalences in the linear time-branching time spectrum are finitely based with respect to BCCSP. 
The negative results in Section~\ref{sec:negative} imply that none of the behavioural equivalences between possible futures equivalence and bisimilarity is finitely based with respect to $\bccsp{}$. 
An interesting question is which of the behavioural equivalences including ready simulation semantics is finitely based with respect to $\bccsp$.

In~\cite{dFGPR13} an alternative classification of the equivalences in the spectrum with respect to~\cite{vG90} was proposed.
In order to obtain a general, unified, view of process semantics, the spectrum was divided into layers, each corresponding to a different notion of constrained simulation~\cite{dFG08}.
There are pleasing connections between the different layers and the partitioning they induce of the congruences in the spectrum, as given in~\cite{dFGPR13}, and the relationships between the axioms for the interleaving operator we have presented in this study.


\section*{Acknowledgement}
This work has been supported by the project `\emph{Open Problems in the Equational Logic of Processes}' (OPEL) of the Icelandic Research Fund (grant No.~196050-051).

We thank Rob van Glabbeek for a fruitful discussion on the axiomatisability of failures equivalence, and the reviewers for their constructive feedback and careful reading of our paper.

\bibliographystyle{alphaurl}
\bibliography{spectrum_LMCS}

\newpage
\appendix



\section{\texorpdfstring{The Mace4 code for $\E_\cs$}{The Mace4 code for ECS}}
\label{app:Mace4_cs}

The following is the Mace4 code we used to generate a model for $\E_\cs$ in which EL2 (as given in \texttt{formulas(goals)} below) does not hold.

\begin{lstlisting}[breaklines=true]
set(verbose).
assign(max_megs, 1000).
assign(domain_size,5).
op(750,prefix,"a").
op(750,prefix,"b").
op(850,infix,"plus").
op(950,infix,"par").
formulas(assumptions).
  (x plus 0) = x.
  (x plus y) = (y plus x).
  ((x plus y) plus z) = (x plus (y plus z)).
  (x plus x) = x.
  (x par 0) = x.
  (x par y) = (y par x).
  
% Axiom EL1
  (a x par a y) = (a (x par a y) plus a (y par a x)).
  (a x par b y) = (a (x par b y) plus b (y par a x)).
  (b x par b y) = (b (x par b y) plus b (y par b x)).
  
% Axiom CS
  (a (a x plus (y plus z))) = (a (a x plus (y plus z)) plus a (a x plus z)).
  (a (b x plus (y plus z))) = (a (b x plus (y plus z)) plus a (b x plus z)).
  (b (a x plus (y plus z))) = (b (a x plus (y plus z)) plus b (a x plus z)).
  (b (b x plus (y plus z))) = (b (b x plus (y plus z)) plus b (b x plus z)).

% Axiom CSP1
  ((a x plus (a y plus u)) par (a z plus (a w plus v))) =
    ( ((a x plus u) par (a z plus (a w plus v)))
      plus ( ((a y plus u) par (a z plus (a w plus v)))
      plus ( ((a x plus (a y plus u)) par (a z plus v))
      plus ((a x plus (a y plus u)) par (a w plus v))
      ))).
  ((a x plus (a y plus u)) par (a z plus (b w plus v))) =
    ( ((a x plus u) par (a z plus (b w plus v)))
      plus ( ((a y plus u) par (a z plus (b w plus v)))
      plus ( ((a x plus (a y plus u)) par (a z plus v))
      plus ((a x plus (a y plus u)) par (b w plus v))
      ))).
  ((a x plus (b y plus u)) par (a z plus (b w plus v))) =
    ( ((a x plus u) par (a z plus (b w plus v)))
      plus ( ((b y plus u) par (a z plus (b w plus v)))
      plus ( ((a x plus (b y plus u)) par (a z plus v))
      plus ((a x plus (b y plus u)) par (b w plus v))
      ))).
  ((a x plus (b y plus u)) par (b z plus (b w plus v))) =
    ( ((a x plus u) par (b z plus (b w plus v)))
      plus ( ((b y plus u) par (b z plus (b w plus v)))
      plus ( ((a x plus (b y plus u)) par (b z plus v))
      plus ((a x plus (b y plus u)) par (b w plus v))
      ))).
  ((b x plus (b y plus u)) par (b z plus (b w plus v))) =
    ( ((b x plus u) par (b z plus (b w plus v)))
      plus ( ((b y plus u) par (b z plus (b w plus v)))
      plus ( ((b x plus (b y plus u)) par (b z plus v))
      plus ((b x plus (b y plus u)) par (b w plus v))
      ))).

% Axiom CSP2
  (a x par (a y plus (a z plus w))) =
    (a(x par (a y plus (a z plus w) ))
      plus ( (a x par (a y plus w))
      plus ( (a x par (a z plus w))
      ))).
  (a x par (a y plus (b z plus w))) =
    (a(x par (a y plus (b z plus w) ))
      plus ( (a x par (a y plus w))
      plus ( (a x par (b z plus w))
      ))).
  (a x par (b y plus (b z plus w))) =
    (a(x par (b y plus (b z plus w) ))
      plus ( (a x par (b y plus w))
      plus ( (a x par (b z plus w))
      ))).

  (b x par (a y plus (a z plus w))) =
    (b(x par (a y plus (a z plus w) ))
      plus ( (b x par (a y plus w))
      plus ( (b x par (a z plus w))
      ))).
  (b x par (a y plus (b z plus w))) =
    (b(x par (a y plus (b z plus w) ))
      plus ( (b x par (a y plus w))
      plus ( (b x par (b z plus w))
      ))).
  (b x par (b y plus (b z plus w))) =
    (b(x par (b y plus (b z plus w) ))
      plus ( (b x par (b y plus w))
      plus ( (b x par (b z plus w))
      ))).
end_of_list.
\end{lstlisting}

\begin{lstlisting}[breaklines=true]
formulas(goals).
  ((a x plus b y) par (a z plus b w)) = (a(x par (a z plus b w)) plus (b(y par (a z plus b w)) plus (a(z par (a x plus b y)) plus b(w par (a x plus b y))))).
end_of_list.
\end{lstlisting}


\section{\texorpdfstring{The Mace4 code for $\E_\rtr$}{The Mace4 code for ERT}}
\label{app:Mace4_rt}

The following is the Mace4 code we used to generate a model for $\E_\rtr$ in which RSP2 (as given in \texttt{formulas(goals)} below) does not hold.

\begin{lstlisting}[breaklines=true]
set(verbose).
assign(max_megs, 1000).
op(750,prefix,"a").
op(750,prefix,"b").
op(850,infix,"plus").
op(950,infix,"par").
formulas(assumptions).
  (x plus 0) = x.
  (x plus y) = (y plus x).
  ((x plus y) plus z) = (x plus (y plus z)).
  (x plus x) = x.
  (x par 0) = x.
  (x par y) = (y par x).
  
% Axiom RT
  ( a(((a x plus a y) plus (a u plus a v)) plus z) ) =
    (    a((a x plus a u) plus z)
    plus
         a((a y plus a v) plus z)
    ).
  ( a(((a x plus a y) plus (b u plus b v)) plus z) ) =
    (    a((a x plus b u) plus z)
    plus
         a((a y plus b v) plus z)
    ).
  ( a(((b x plus b y) plus (b u plus b v)) plus z) ) =
    (    a((b x plus b u) plus z)
    plus
         a((b y plus b v) plus z)
    ).
  ( b(((a x plus a y) plus (a u plus a v)) plus z) ) =
    (    b((a x plus a u) plus z)
    plus
         b((a y plus a v) plus z)
    ).
  ( b(((a x plus a y) plus (b u plus b v)) plus z) ) =
    (    b((a x plus b u) plus z)
    plus
         b((a y plus b v) plus z)
    ).
  ( b(((b x plus b y) plus (b u plus b v)) plus z) ) =
    (    b((b x plus b u) plus z)
    plus
         b((b y plus b v) plus z)
    ).
    
% Axiom FP
  ( (a x plus (a y plus w)) par z) = ( ((a x plus w) par z) plus ((a y plus w) par z) ).
  
% Axiom EL2   
  (a x par a y) = (a (x par a y) plus a (y par a x)).
  (a x par b y) = (a (x par b y) plus b (y par a x)).
  (b x par b y) = (b (x par b y) plus b (y par b x)).
  ((a x plus b y) par (a z plus b w)) =
    (      a (x par (a z plus b w))
    plus ( b (y par (a z plus b w))
    plus ( a ((a x plus b y) par z)
    plus ( b ((a x plus b y) par w)
    )))).
end_of_list.
\end{lstlisting}

\begin{lstlisting}[breaklines=true]
formulas(goals).
  ( a x par (b u plus (b v plus w) ) ) =
    (    ( a x par (b u plus w ) )
    plus (( a x par (b v plus w ) )
         plus a(x par (b u plus (b v plus w) ))
    )).
end_of_list.
\end{lstlisting}


\section{\texorpdfstring{The Mace4 code for $\E_\ctr$}{The Mace4 code for ECT}}
\label{app:Mace4_ct}

The following is the Mace4 code we used to generate a model for $\E_\ctr$ in which EL2 (as given in \texttt{formulas(goals)} below) does not hold.

\begin{lstlisting}[breaklines=true]
set(verbose).
assign(max_megs, 1000).
assign(domain_size,5).
op(750,prefix,"a").
op(750,prefix,"b").
op(850,infix,"plus").
op(950,infix,"par").
formulas(assumptions).
  (x plus 0) = x.
  (x plus y) = (y plus x).
  ((x plus y) plus z) = (x plus (y plus z)).
  (x plus x) = x.
  (x par 0) = x.
  (x par y) = (y par x).
  
% Axiom EL1
  (a x par b y) = (a (x par b y) plus b (y par a x)).
  (a x par a y) = (a (x par a y) plus a (y par a x)).
  (b x par b y) = (b (x par b y) plus b (y par b x)).
  
% Axiom CT
  (a(a x plus z) plus a(a y plus w)) = (a(a x plus (a y plus (z plus w)))).
  (a(a x plus z) plus a(b y plus w)) = (a(a x plus (b y plus (z plus w)))).
  (a(b x plus z) plus a(a y plus w)) = (a(b x plus (a y plus (z plus w)))).
  (a(b x plus z) plus a(b y plus w)) = (a(b x plus (b y plus (z plus w)))).
  (b(a x plus z) plus a(a y plus w)) = (b(a x plus (a y plus (z plus w)))).
  (b(a x plus z) plus a(b y plus w)) = (b(a x plus (b y plus (z plus w)))).
  (b(b x plus z) plus a(a y plus w)) = (b(b x plus (a y plus (z plus w)))).
  (b(b x plus z) plus a(b y plus w)) = (b(b x plus (b y plus (z plus w)))).
  
% Axiom CTP
  ((a x plus (a y plus w)) par z) = (((a x plus w) par z) plus ((a y plus w) par z)).
  ((a x plus (b y plus w)) par z) = (((a x plus w) par z) plus ((b y plus w) par z)).
  ((b x plus (b y plus w)) par z) = (((b x plus w) par z) plus ((b y plus w) par z)).
end_of_list.
\end{lstlisting}

\begin{lstlisting}[breaklines=true]
formulas(goals).
  ((a x plus b y) par (a z plus b w)) = (a(x par (a z plus b w)) plus (b(y par (a z plus b w)) plus (a(z par (a x plus b y)) plus b(w par (a x plus b y))))).
end_of_list.
\end{lstlisting}

By combining the \texttt{formulas(assumptions)} given above with the \texttt{formulas(goals)} presented below, we can generate a model for $\E_\ctr$ in which CSP2 does not hold.

\begin{lstlisting}[breaklines=true]
formulas(goals).
  ( a x par (b u plus (b v plus w) ) ) =
    (    ( a x par (b u plus w ) )
    plus (( a x par (b v plus w ) )
         plus a(x par (b u plus (b v plus w) ))
    )).
end_of_list.
\end{lstlisting}


\section{Simplifying the equational theory: saturated systems}
\label{app:saturation}

The axioms A0 and P0 in Table~\ref{tab:basic-axioms} (used from left to right) are enough to establish that each $\bccsp$ term that is possible futures equivalent to $\nil$ is also provably equal to $\nil$. 

\begin{lem}
\label{lem:prove-nil}
Let $t$ be a $\bccsp$ term. 
Then $t \sim_\pf \nil$ if and only if the equation $t \approx \nil$ is provable using axioms \emph{A0} and \emph{P0} in Table~\ref{tab:basic-axioms} from left to right.
\end{lem}

\begin{proof}
The ``if'' implication is an immediate consequence of the soundness of the equations A0 and P0 with respect to $\sim_\pf$. 
To prove the  ``only if'' implication, define, first of all, the collection $\text{NIL}$ of $\bccsp$ terms as the set of terms generated by the following grammar:
\[
t::= \nil \;\mid\; t+t \;\mid\; t \parcomp t, 
\]
We claim that each $\bccsp$ term $t$ is $\sim_\pf$ equivalent to $\nil$ if and only if $t\in\text{NIL}$.
Using this claim and structural induction on $t\in\text{NIL}$, it is a simple matter to show that if $t\sim_\pf \nil$, then $t\approx \nil$ is provable using axioms A0 and P0 from left to right, which was to be shown.
  
To complete the proof, it therefore suffices to show the above claim. 
To establish the ``if'' implication in the statement of the claim, one proves, using structural induction on $t$ and the congruence properties of $\sim_\pf$, that if $t\in\text{NIL}$, then $\sigma(t) \sim_\pf \nil$ for every
closed substitution $\sigma$.  
To show the ``only if'' implication, we establish the contrapositive statement, namely that if $t\not\in\text{NIL}$, then $\sigma(t) \not\sim_\pf \nil$ for some closed substitution $\sigma$. 
To this end, it suffices only to show, using structural induction on $t$, that if $t\not\in\text{NIL}$, then $\sigma_a(t)\trans[a]$ for some action $a \in \Act$, where $\sigma_a$ is the closed substitution mapping each variable to the closed term $a\nil$. 
The details of this argument are not hard, and are therefore left to the reader.
\end{proof}

In light of the above result, we shall assume, without loss of generality, that each axiom system we consider includes the equations in Table~\ref{tab:basic-axioms}.  
This assumption means, in particular, that our axiom systems will allow us to identify each term that is possible futures equivalent to $\nil$ with $\nil$.

We recall that a $\bccsp$ term $t$ has a {\sl $\nil$ factor} if it contains a subterm of the form $t_1 \parcomp t_2$, where either $t_1$ or $t_2$ is possible futures equivalent to $\nil$.

It is easy to see that, modulo the equations in Table~\ref{tab:basic-axioms}, every $\bccsp$ term $t$ has the form $\sum_{i\in I}t_i$, for some finite index set $I$, and terms $t_i$ ($i\in I$) that are not $\nil$ and do not have themselves the form $t'+t''$, for some terms $t'$ and $t''$.  
The terms $t_i$ ($i\in I$) will be referred to as the {\em summands} of $t$.  
Moreover, again modulo the equations in Table~\ref{tab:basic-axioms}, each of the $t_i$ can be assumed to have no $\nil$ factors.

We can now introduce the notion of \emph{saturated system}, namely an axiom system such that if a closed equation that relates two terms containing no occurrences of $\nil$ as a summand or factor, then there is a closed proof for it in which all of the terms have no occurrences of $\nil$ as a summand or factor (cf.~\cite[Proposition~5.1.5]{Mol89}).

\begin{defi}
\label{def:no_nil}
For each $\bccsp$ term $t$, we define ${t}/\nil$ thus: 
\begin{align*}
& {\nil}/\nil = \nil \qquad\qquad {x}/\nil = x \qquad\qquad {a t}/\nil = a ({t}/\nil) \\
& ({t+u})/\nil =
\begin{cases}
{u}/\nil & \text{ if } t \sim_\pf \nil \\
{t}/\nil & \text{ if } u \sim_\pf \nil \\
({t}/\nil)+ ({u}/\nil) & \text{ otherwise}
\end{cases} 
\qquad (t \parcomp u)/\nil =  
\begin{cases}
{u}/\nil & \text{ if } t \sim_\pf \nil \\
{t}/\nil & \text{ if } u \sim_\pf \nil \\
({t}/\nil) \parcomp ({u}/\nil) & \text{ otherwise.}
\end{cases}  
\end{align*}
\end{defi}

Intuitively, $t /\nil$ is the term that results by removing {\sl all} occurrences of $\nil$ as a summand or factor from $t$. 

The following lemma collects the basic properties of the above construction.
 
\begin{lem}
\label{lem:no_nil}
For each $\bccsp$ term $t$, the following statements hold:
\begin{enumerate}
\item \label{nil1} 
The equation $t \approx t /\nil$ can be proven using the equations in Table~\ref{tab:basic-axioms}, and therefore $t \sim_\pf t/\nil$.
\item \label{nil2} 
The term $t /\nil$ has no $\nil$ summands or factors.
\item \label{nil3} 
$t /\nil = t$, if $t$ has no occurrence of $\nil$ as a summand or factor.
\item \label{nil4}
$\sigma(t /\nil)/\nil = \sigma(t)/\nil$, for each substitution $\sigma$.
\end{enumerate}
\end{lem}

\begin{proof}
Immediate by structural induction over $t$.
\end{proof}

\begin{defi}
[Saturated system]
\label{def:closure}
We say that a substitution $\sigma$ is a {\sl $\nil$-substitution} if and only if $\sigma(x)\neq x$ implies that $\sigma(x)=\nil$, for each variable $x$.
Let $\E$ be an axiom system. 
We define the axiom system $\cl(\E)$ thus:
\[
\cl(\E) = \E \cup \{ \sigma(t)/\nil \approx \sigma(u)/\nil \mid (t\approx u)\in \E,~\sigma~\text{a $\nil$-substitution} \}.
\]
An axiom system $\E$ is {\sl saturated} if $\E =\cl(\E)$.
\end{defi}

The following lemma collects some basic sanity properties of the closure operator $\cl(\cdot)$. 

\begin{lem}
\label{lem:properties-closure}
Let $\E$ be an axiom system. 
Then the following statements hold. 
\begin{enumerate}
\item \label{cl-cl} 
$\cl(\E) = \cl(\cl(\E))$. 
\item \label{cl-finite} 
$\cl(\E)$ is finite, if so is $\E$.
\item $\cl(\E)$ is sound modulo possible futures equivalence, if so is $\E$.
\item $\cl(\E)$ is closed with respect to symmetry, if so is $\E$.
\item \label{cl-equi} 
$\cl(\E)$ and $\E$ prove the same equations, if $\E$ contains the equations in Table~\ref{tab:basic-axioms}.
\end{enumerate}
\end{lem} 

\begin{proof}
We limit ourselves to sketching the proofs of statements~\eqref{cl-cl} and~\eqref{cl-equi} in the lemma. 
  
In the proof of statement~\eqref{cl-cl}, the only non-trivial thing to check is that the equation
\[
\sigma(\sigma'(t)/\nil))/\nil \approx \sigma(\sigma'(u)/\nil))/\nil 
\]
is contained in $\cl(\E)$, whenever $(t \approx u)\in \E$ and $\sigma,\sigma'$ are $\nil$-substitutions. 
This follows from Lemma~\ref{lem:no_nil}.\eqref{nil4} because the collection of $\nil$-substitutions is closed under composition.
  
To show statement~\eqref{cl-equi}, it suffices only to argue that each equation $t\approx u$ that is provable from $\cl(\E)$ is also provable from $\E$, if $\E$ contains the equations in Table~\ref{tab:basic-axioms}. 
This can be done by induction on the depth of the proof of the equation $t\approx u$ from $\cl(\E)$, using Lemma~\ref{lem:no_nil}.\eqref{nil1} for the case in which $t\approx u$ is a substitution instance of an axiom in $\cl(\E)$.
\end{proof}

We are now ready to state our counterpart of~\cite[Proposition~5.1.5]{Mol89}.

\begin{prop}
\label{prop:saturation}
Assume that $\E$ is a saturated axiom system. 
Suppose furthermore that we have a closed proof from $\E$ of the closed equation $p\approx q$. 
Then replacing each term $r$ in that proof with $r/\nil$ yields a closed proof of the equation $p/\nil\approx q/\nil$. 
In particular, the proof from $\E$ of an equation $p\approx q$, where $p$ and $q$ are terms not containing occurrences of $\nil$ as a summand or factor, need not use terms containing occurrences of $\nil$ as a summand or factor.
\end{prop}

\begin{proof}
The proof follows the lines of that of~\cite[Proposition~5.1.5]{Mol89}, and is therefore omitted.
\end{proof}

\end{document}